\documentclass[a4paper,USenglish,cleveref,autoref,thm-restate]{lipics-v2021}
\usepackage[utf8]{inputenc}
\usepackage[svgnames]{xcolor}
\usepackage{amsmath,amssymb,csquotes,xcolor,array,multicol,amsthm,thmtools,bbm}
\usepackage{xparse}
\usepackage{zerosum}

\usepackage{graphicx}

\usepackage{enumitem}
\setlist[itemize,enumerate]{noitemsep,nolistsep}

\usepackage{microtype}
\usepackage{verbatim} %
\usepackage{hyperref}

\usepackage{xspace}

\usepackage{nicefrac}

\usepackage{bm}

\usepackage{marginnote}

\usepackage{cite}

\graphicspath{{figs/}}

\renewcommand{\lg}{\log}
\renewcommand{\varepsilon}{\epsilon}

\newcommand{\poly}{\operatorname{poly}}
\newcommand{\polylg}{\operatorname{polylog}}
\newcommand{\polylog}{\polylg}

\newcommand{\F}{\mathcal{S}}
\newcommand{\I}{\mathcal{I}}
\newcommand{\J}{\mathcal{J}}
\renewcommand{\S}{\F}

\newcommand{\Procedure}{\mathcal{P}}

\newcommand{\numSmall}{\textsf{small}}

\newcommand{\tO}{\tilde{O}}

\newcommand{\cut}{\mathit{cut}}
\def\iscut{\operatorname{true}}
\def\notcut{\operatorname{false}}
\def\In{\operatorname{In}}
\def\Reach{\operatorname{Reach}}

\def\Winfty{W_{\infty}}

\providecommand{\abs}[1]{\ensuremath{\left\lvert#1\right\rvert}}

\newcommand{\add}{\operatorname{add}}
\newcommand{\shift}{\operatorname{shift}}
\newcommand{\round}{\operatorname{round}}

\newcommand{\case}{\operatorname{case}}

\newcommand{\minconv}{\oplus}

\def\PTASeps{\bar{\epsilon}}

\input widetilde.tex
\usepackage{mathtools}

\newcommand{\DP}{\mathsf{DP}}
\newcommand{\tDP}{\mathsf{ADP}}
\newcommand{\ADP}{\tDP}

\renewcommand{\P}{\textsf{P}}
\newcommand{\NP}{\textsf{NP}}

\newcommand{\Prob}[1]{\mathbf{Pr}\left( #1 \right)}

\DeclareRobustCommand{\ALG}{%
	\ifmmode
		\operatorname{ON}
	\else
		\text{ON}\xspace
	\fi
}
\DeclareRobustCommand{\OFF}{%
	\ifmmode
		\operatorname{OFF}
	\else
		\text{OFF}\xspace
	\fi
}
\DeclareRobustCommand{\APPROXALGO}{%
	\ifmmode
		\operatorname{APPROX}
	\else
		\text{APPROX}\xspace
	\fi
}

\def\bold #1{{\bfseries\mathversion{bold}#1}}

\renewcommand{\ss}[1]{}
\newcommand{\stefan}[1]{}
\newcommand{\sn}[1]{}

\newcommand{\hr}[1]{}
\newcommand{\old}[1]{{}}

\def\capac{\operatorname{cap}}
\def\mincut{\operatorname{mincut}}

\def\OPT{\operatorname{OPT}}

\title{Dynamic Maintenance of Monotone Dynamic Programs and Applications}

\author{Monika Henzinger}{Faculty of Computer Science, University of Vienna}{monika.henzinger@univie.ac.at}{https://orcid.org/0000-0002-5008-6530}{This project has received funding from the European Research Council (ERC) under the European Union's Horizon 2020 research and innovation programme (Grant agreement No. 101019564 ``The Design of Modern Fully Dynamic Data Structures (MoDynStruct)'' and from the Austrian Science Fund (FWF) project ``Fast Algorithms for a Reactive Network Layer (ReactNet)'', P~33775-N, with additional funding from the \textit{netidee SCIENCE Stiftung}, 2020--2024.  \flag{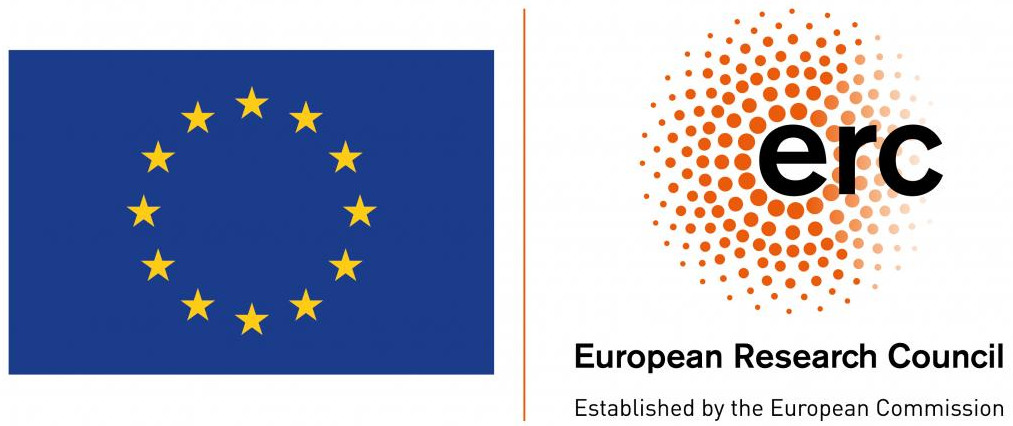}}
\author{Stefan Neumann}{KTH Royal Institute of Technology, Stockholm, Sweden}{neum@kth.se}{}{This research is supported by the the ERC Advanced Grant REBOUND~(834862) and the EC H2020 RIA project SoBigData++~(871042).}
\author{Harald Räcke}{TU Munich, Munich, Germany}{raecke@in.tum.de}{https://orcid.org/0000-0001-8797-717X}{}
\author{Stefan Schmid}{TU Berlin, Germany and Fraunhofer SIT, Germany}{stefan.schmid@tu-berlin.de}{https://orcid.org/0000-0002-7798-1711}{Research supported by Austrian Science Fund (FWF) project I 5025-N (DELTA), 2020-2024.}

\nolinenumbers

\hideLIPIcs  %

\authorrunning{M.\ Henzinger, S.\ Neumann, S.\ Schmid, H.\ Räcke} %
\Copyright{Monika Henzinger, Stefan Neumann, Stefan Schmid and Harald Räcke} %


\category{} %
\relatedversion{} %

\EventEditors{John Q. Open and Joan R. Access}
\EventNoEds{2}
\EventLongTitle{42nd Conference on Very Important Topics (CVIT 2016)}
\EventShortTitle{CVIT 2016}
\EventAcronym{CVIT}
\EventYear{2016}
\EventDate{December 24--27, 2016}
\EventLocation{Little Whinging, United Kingdom}
\EventLogo{}
\SeriesVolume{42}
\ArticleNo{23}

\begin{document}

\maketitle

\begin{abstract}
	Dynamic programming (DP) is one of the fundamental paradigms in algorithm
	design.  However, many DP algorithms have to fill in large DP tables,
	represented by two-dimensional arrays, which causes at least quadratic
	running times and space usages. This has led to the development of
        improved algorithms for special cases when the DPs satisfy additional
        properties like, e.g., the Monge property or total
	monotonicity.
	
	In this paper, we consider a new condition which assumes
	(among some other technical assumptions) that the \emph{rows} of the DP
	table are \emph{monotone}. Under this assumption, we introduce a novel data
	structure for computing $(1+\varepsilon)$-approximate DP
	solutions in \emph{near-linear time and space} in the static setting, and
	with \emph{polylogarithmic update times} when the DP entries change
	dynamically.  To the best of our knowledge, our new condition is
	incomparable to previous conditions and is the first which allows to derive
	\emph{dynamic} algorithms based on existing DPs.
	Instead of using two-dimensional arrays to store the DP tables, we
	store the rows of the DP tables using monotone piecewise constant functions.
	This allows us to store length-$n$ DP table rows with
	entries in $[0,W]$ using only $\polylog(n,W)$ bits, and to perform
	operations, such as $(\min,+)$-convolution or rounding, on these functions
	in polylogarithmic time.

	We further present several applications of our data structure.  For bicriteria
	versions of \emph{$k$-balanced graph partitioning} and
	\emph{simultaneous source location}, we obtain the first dynamic algorithms
	with subpolynomial update times, as well as the first static algorithms
	using only near-linear time and space.  Additionally, we obtain the currently
	fastest algorithm for \emph{fully dynamic knapsack}.  For $k$-balanced
	partitioning, we show how to monotonize an existing non-monotone DP by
	Feldmann and Foschini~(Algorithmica'15); for simultaneous source location,
	we obtain an efficient algorithm by considering the inverse DP function of
	the one used by Andreev, Garrod, Golovin, Maggs, and Meyerson~(TALG'09).
	Our result for fully dynamic knapsack improves upon a recent result by
	Eberle, Megow, N{\"{o}}lke, Simon and Wiese (FSTTCS'21).

\end{abstract}

\tableofcontents
\newpage

\section{Introduction}
\label{sec:introduction}
Dynamic programming (DP) is one of the fundamental paradigms in algorithm
design. 
In the DP~paradigm, a complex problem is broken up into simpler subproblems
and then the original problem is solved by combining the solutions for the subproblems.
One of the drawbacks of DP algorithms is that in practice they are often
 slow and memory-intensive: for inputs of size $n$ their running
time is typically~$\Omega(n^2)$, and when the DP table is stored using a
two-dimensional array they also need space $\Omega(n^2)$.

This motivated researchers to develop more efficient DP algorithms with
near-linear time and space. Indeed, such improvements are possible under a wide
range of conditions on the DP
tables~\cite{bein09knuth,aggarwal87geometric,eppstein88speeding,galil92dynamic,chan21near,knuth71optimum,yao80efficient,monge1781memoire,burkard96perspectives,miller11approximate},
such as the Monge property, total monotonicity, certain convexity and concavity
properties, or the Knuth--Yao quadrangle-inequality; we discuss these properties
in more detail in Appendix~\ref{sec:comparison}.  When these properties hold,
typically one does not have to compute the entire DP table but instead only has
to compute $O(n)$ DP~entries which reveal the optimal solution.

However, we are not aware of any property for DPs that yields efficient
\emph{dynamic} algorithms, i.e., algorithms that provide efficient update
operations when the input changes. One might find this somewhat surprising
because, from a conceptual point of view, many dynamic algorithms 
hierarchically partition the input and maintain solutions for subproblems;
this is quite similar to how many DP schemes are derived. Indeed, this
conceptual similarity is exploited by many ``hand-crafted''
algorithms~(e.g., \cite{compton20new,henzinger20dynamic}) which start with a
DP~scheme and then show how to maintain it dynamically under input changes.
However, such algorithms are often quite involved and their analysis often
requires sophisticated charging schemes.

Hence, it is natural to ask whether there exists a
\emph{general criterion} which, if satisfied, guarantees that a given DP can be
updated efficiently under input changes.

\smallskip
\textbf{Our Contributions.}
The main contribution of our paper is the introduction of a general criterion
which allows to approximate all entries of a DP table up to a factor of
$1+\varepsilon$. We show that if our criterion is satisfied by a DP (with
suitable parameters) then:
\begin{itemize}
	\item In the dynamic setting, we can maintain a
		$(1+\varepsilon)$-approximation of the entire DP table using
		polylogarithmic update time (see Theorem~\ref{thm:well-behaved-dynamic}).
	\item In the static setting, we can compute a
		$(1+\varepsilon)$-approximation of the DP table in near-linear time
		and space (see Theorem~\ref{thm:well-behaved-static}).
\end{itemize}

Our criterion essentially asserts that the \emph{rows} of the DP tables should
be \emph{monotone} and that the dependency graph of the DP should be a DAG,
where the sets of reachable nodes are small, among some other technical
conditions (see Definition~\ref{def:well-behaved} for the formal definition).
Our criterion is incomparable to the Monge property, total monotonicity or other
criteria from the literature (see Appendix~\ref{sec:comparison} for a more
detailed discussion).

To obtain our results, we introduce a novel data structure for maintaining DPs
which satisfy our criterion. Our data structure is based on the idea of storing
the DP rows using \emph{monotone piecewise constant} functions.  The
monotonicity of the DP rows will allow us to ensure that our functions only
contain very few pieces. Then we show that we can perform operations on such
functions very efficiently, with the running times only depending on the number
of pieces.  This is crucial because it allows us to compute an \emph{entire}
$(1+\delta)$-approximate DP row in time just $\polylg(W)$, even when the DP has
$\Omega(n)$ columns, assuming that the DP entries are from $[0,W]$. Note
that if $W\leq\poly(n)$ then this decreases the running time for computing an
entire row from $\Omega(n)$ to just $\polylog(n)$. Additionally, this also
allows us to store each row \emph{using only $\polylg(W)$ space} rather than
storing it in an array of size $\Omega(n)$. 
We present our criterion and the details of our data structure in
Section~\ref{sec:approx-conv}.

As applications of our data structure, we obtain new static and dynamic
algorithms for various problems.
We present new algorithms for \emph{$k$-balanced partitioning},
\emph{simultaneous source location} and for \emph{fully dynamic knapsack}.
Next, we describe these results in detail; we discuss more related work in
Appendix~\ref{sec:related}.

\smallskip
\textbf{Our Results for Fully Dynamic 0-1 Knapsack.}
First, we provide a novel algorithm for fully dynamic 0-1~knapsack.
In this problem, the input consists of a knapsack size $B\in\mathbb{R}_+$ and a
set of $n$~items, where each item $i\in[n]$ has a weight $w_i\in\mathbb{R}_+$
and a price $p_i\in[1,\infty)$.
The goal is to find a set of items $I$ that maximizes
$\sum_{i\in I} p_i$ while satisfying the constraint $\sum_{i\in I} w_i \leq B$.
In the dynamic version of the problem, items are inserted and deleted.
More concretely, we consider the following update operations:
\emph{insert}($p_i,w_i$), in which the price and weight of item~$i$ are set to
$p_i\in[1,\infty)$ and $w_i\in\mathbb{R}_+$, respectively, and \emph{delete}($i$),
where item~$i$ is removed from the set of items.

Our main result is a dynamic $(1+\varepsilon)$-approximation algorithm
with worst-case update time 
$\varepsilon^{-2} \cdot \log^2(nW) \cdot \polylg(1/\varepsilon,\log(nW))$, where
$W=\sum_i p_i$.
Our algorithm improves upon a recent result by Eberle, Megow, N{\"{o}}lke, Simon
and Wiese~\cite{eberle21fully} that also maintained a
$(1+\varepsilon)$-approximate solution with update time $O(\varepsilon^{-9} \lg^4(nW))$. 
\begin{restatable}{theorem}{knapsack}
\label{thm:knapsack}
	Let $\varepsilon>0$. There exists an algorithm for fully dynamic knapsack
	that maintains a $(1+\varepsilon)$-approximate solution with worst-case
	update time
	$\frac{1}{\varepsilon^2} \log^2(nW) \polylog\left( \frac{1}{\varepsilon} \log(nW)\right)$.
\end{restatable}

We will also show that we can return the maintained solution~$I$ in time
$O(\abs{I})$ and that we can answer queries whether a given item $i\in[n]$ is
contained in $I$ in time $O(1)$. This matches the query times
of~\cite{eberle21fully}.

To obtain this result, we first derive a slightly slower algorithm as a simple
application of our data structure for maintaining DPs with monotone rows.  Then
we use this algorithm together with additional ideas to obtain the theorem (see
Section~\ref{sec:knapsack}).

Since our dynamic algorithm is based on a DP, it is possible that the solution
changes significantly after each update. However, in the appendix
(Theorem~\ref{thm:recourse-knapsack}) we prove a lower bound, showing that every
dynamic $(1+\varepsilon)$-approximation algorithm for knapsack must either make
a lot of changes to the solution after each update or store many (potentially
substantially different) solutions between which it can switch after each
update. This implies that maintaining a single explicit solution with
polylogarithmic update times is not possible and the
property of our algorithm cannot be avoided.

\smallskip
\textbf{Our Results for $k$-Balanced Partitioning.}
Our most technically challenging result is for \emph{$k$-balanced graph partitioning}.
In this problem, the input consists of an integer~$k$ and an undirected weighted
graph $G=(V,E,\capac)$ with $n$~vertices, where $\capac: E \to \Winfty$ is a
weight function on the edges with weights in $\Winfty:=[1,W]\cup\{0,\infty\}$.
The goal is to find a partition $V_1,\dots,V_k$ of the vertices such that
$\abs{V_i} \leq \lceil \abs{V}/k\rceil$ for all $i$ and the weight of the edges
which are cut by the partition is minimized. More formally, we want to minimize
$\cut(V_1,\dots,V_k) := \sum_{i=1}^k \sum_{\{u,v\}\in E\cap (V_i\times (V\setminus V_i))} \capac(u,v)$.

We note that this problem is highly relevant in
theory~\cite{even99fast,andreev2006balanced,feldmann15balanced,feige02polylogarithmic}
and in
practice~\cite{karypis1997metis,sanders13experimental,buluc16recent,dong20learning}, 
where algorithms for balanced graph partitioning are often used as a
preprocessing step for large scale data analytics.  Obtaining practical
improvements for this problem is of considerable interest in applied
communities~\cite{buluc16recent} and, for instance, the popular METIS
heuristic~\cite{karypis1997metis} has 1,400+ citations.

Since the above problem is \NP-hard to approximate within a factor of
$n^{1-\varepsilon}$ for any $\varepsilon>0$ even on
trees~\cite{feldmann15balanced}, we consider bicriteria approximation
algorithms. Given an undirected weighted graph $G=(V,E,\capac)$, a
partition $V_1,\dots,V_k$ of $V$ is a \emph{bicriteria
$\emph(\alpha,\beta)$-approximate solution} if
$\abs{V_i}\leq\beta\lceil n/k\rceil$ for all $i$
and $\cut(V_1,\dots,V_k) \leq \alpha\cdot \cut(\OPT)$,
where $\OPT=(V_1^*,\dots,V_k^*)$ is the optimal solution with
$\abs{V_i^*}\leq\lceil n/k\rceil$ for all~$i$.  We note that the previously
mentioned hardness result implies that any algorithm that computes a bicriteria
$(\alpha,1+\varepsilon)$-approximation for any $\alpha\geq1$ and whose running
time depends only polynomially on $n$, must have a running time depending
super-polynomially on $1/\varepsilon$, unless $\P=\NP$.\footnote{
	If we had an algorithm that computes a bicriteria
	$(\alpha,1+\varepsilon)$-approximation in time $\poly(n,1/\varepsilon)$
	then we could set $\varepsilon=1/(2n)$ which implies that all partitions
	have size $\lceil n/k\rceil$. Thus we can compute a bicriteria
	$(\alpha,1)$-approximate solution in time $\poly(n)$ which contradicts the
	hardness result, unless $\P=\NP$.}

Our main result for the static setting is presented in the following theorem.
It gives the first algorithm with polylogarithmic approximation ratio for this
problem \emph{with near-linear running time}. More concretely, we compute a
bicriteria $(O(\lg^4 n), 1+\varepsilon)$-approximation in near-linear time for
constant $k$.
For comparison, the best approximation ratio achieved by a polynomial-time algorithm~\cite{feldmann15balanced}
is a bicriteria $(O(\lg^{1.5} n \lg \lg n), 1+\varepsilon)$-approximation
with running time $\Omega(n^4)$.
\begin{restatable}{theorem}{partitioningstatic}
\label{thm:partitioning-static}
	Let $\varepsilon > 0$ and $k\in\mathbb{N}$. Let
	$G=(V,E,\capac)$ be an undirected weighted graph with $n$~vertices and $m$~edges and
	edge weights in $\Winfty$. Then for the $k$-balanced partition problem we
	can compute:
	\begin{itemize}
		\item An $(O(\log^4 n),1+\varepsilon)$-approximation in
		time $(k/\varepsilon)^{O(\lg(1/\varepsilon)/\varepsilon)} \cdot O'(m \cdot \lg^2(W) )
			+ (k/\varepsilon)^{O(1/\varepsilon^2)}$.\footnote{
			We use the notation $O'(\cdot)$ to suppress
			factors in $\poly(\lg n, k, \lg(1/\varepsilon), \lg\lg(W))$.
		}
		\item A $(1+\varepsilon,1+\varepsilon)$-approximation in time
			$(k/\varepsilon)^{O(\lg(1/\varepsilon)/\varepsilon)} \cdot
				O'(n \cdot h^2 \cdot \lg^2(W)) + (k/\varepsilon)^{O(1/\varepsilon^2)}$
			if $G$ is a tree of height~$h$.
		\item A $(1,1+\varepsilon)$-approximation in time
			$(k/\varepsilon)^{O(\lg(1/\varepsilon)/\varepsilon)} \cdot
				O'(n^4 \cdot \lg^2(W)) + (k/\varepsilon)^{O(1/\varepsilon^2)}$
			if $G$ is a tree.
	\end{itemize}
\end{restatable}

Furthermore, we extend our results to the dynamic setting in which the graph
$G$ is undergoing edge insertions and deletions.  In the following theorem, we
present \emph{the first dynamic algorithm with subpolynomial update time} for
this problem. We again consider bicriteria approximation algorithms with
update and query times depending super-polynomially on~$1/\varepsilon$; this
cannot be avoided since if we computed $(\alpha,1)$-approximations for any
$\alpha\geq1$ or if we had a polynomial dependency on~$1/\varepsilon$, then the
hardness result from above implies that our update and query times must be
super-polynomial in $n$ (unless $\P=\NP$).
\begin{restatable}{theorem}{partitioningdynamic}
\label{thm:partitioning-dynamic}
	Let $\varepsilon > 0$ and $k\in\mathbb{N}$. Let $G=(V,E,\capac)$ 
	be an undirected weighted graph with $n$ vertices that is undergoing edge insertions and deletions.
	Then for the $k$-balanced partition problem we can maintain:
	\begin{itemize}
		\item An $(n^{o(1)},1+\varepsilon)$-approximate solution with amortized
			update time 
			$(k/\varepsilon)^{O(\lg(1/\varepsilon)/\varepsilon)} \cdot n^{o(1)}
				\cdot O'(\lg^2(W))$
			and query time $(k/\varepsilon)^{O(1/\varepsilon^2)}$
			if $G$ is unweighted.
		\item A $(1+\varepsilon,1+\varepsilon)$-approximate solution with
			worst-case update time
			$(k/\varepsilon)^{O(\lg(1/\varepsilon)/\varepsilon)} \cdot O'(h^3 \cdot \lg^2(W))$
			and query time
			$(k/\varepsilon)^{O(1/\varepsilon^2)}$
			if $G$ is a tree of height $h$.
	\end{itemize}
\end{restatable}

Our approach is inspired by the DP of Feldmann and
Foschini~\cite{feldmann15balanced}.  However, the DP rows in the algorithm
of~\cite{feldmann15balanced} are not monotone and, hence, their DP cannot
directly be sped up by our approach.  Therefore, we first simplify and
generalize the exact DP of Feldmann and Foschini to make it monotone. The DP we
obtain eventually is still slightly too complex to fit into our black-box
framework, but we show that the ideas from our framework can still be used to
obtain the result. In Section~\ref{sec:monotonizing}, we provide a technical
overview.

Again, it is possible that the solution maintained by our algorithm changes
substantially after each update. Similar to above we show in the appendix
(Theorem~\ref{thm:recourse-partitioning}) that this cannot be avoided when
considering subpolynomial update times.

\smallskip
\textbf{Our Results for Simultaneous Source Location.}
Next, we provide efficient algorithms for the simultaneous source
location problem by Andreev, Garrod, Golovin, Maggs and
Meyerson~\cite{andreev09simultaneous}. In
this problem, the input consists of an undirected graph $G=(V,E,\capac,d)$ with a capacity
function $\capac \colon E \to \Winfty$ on the edges and a demand function
$d \colon V \to \Winfty$ on the vertices. The goal is to select a
minimum set $S\subseteq V$ of \emph{sources} that can simultaneously supply all
vertex demands. More concretely, a set of sources $S$ is \emph{feasible} if
there exists a flow from the vertices in $S$ that supplies demand $d(v)$ to all
vertices $v\in V$ and that does not violate the capacity constraints on the
edges. The objective is to find a feasible set of sources of minimum size.

We will again consider bicriteria approximation algorithms.  Let $S^*$ be the
optimal solution for the simultaneous source location problem.  Then we say that
$S$ is a \emph{bicriteria $(\alpha,\beta)$-approximate solution} if
$\abs{S}\leq\alpha\abs{S^*}$ and if $S$ is a feasible set of sources when all
edge capacities are increased by a factor $\beta$.

The following theorem summarizes our main results.  It presents \emph{the first
near-linear time algorithm} for simultaneous source location
that computes a $(1+\varepsilon)$-approximate solution while only exceeding the
edge capacities by a $O(\lg^4 n)$ factor.  In comparison, the best algorithm
with arbitrary polynomial processing time
computes a bicriteria $(1,O(\lg^2 n\lg\lg n))$-approximate
solution in time $\Omega(n^3)$~\cite{andreev09simultaneous}.
\begin{restatable}{theorem}{sslstatic}
\label{thm:ssl-static}
	Let $\varepsilon > 0$. Let $G=(V,E,\capac,d)$ be an undirected weighted graph with
	$n$~vertices and $m$~edges.  Then for the simultaneous source location problem we can
	compute:
	\begin{itemize}
		\item A $(1+\varepsilon,O(\log^4(n)))$-approximation in
			time\footnote{We write $\tO(f(n,\varepsilon,W))$ to denote running
				times of the form
				$f(n,\varepsilon,W)\cdot\polylg(n,\varepsilon,\lg W)$.}
			$\tO(\frac{1}{\varepsilon^2} m)$.
		\item A $(1+\varepsilon,1)$-approximation in time
			$\tO(\frac{1}{\varepsilon^2} h^2 \cdot n)$ if $G$ is a tree of height~$h$.
	\end{itemize}
\end{restatable}

Next, we turn to dynamic versions of the problem. We consider the following
update operations: 
\emph{SetDemand}($v$,~$d$): updates the demand of vertex $v$ to $d(v)=d$,
\emph{SetCapacity}($(u,v)$,~$c$): updates the capacity of the edge $(u,v)$ to
								  $\capac(u,v)=c$,
\emph{Remove}($u,v$): removes the edge $(u,v)$,
\emph{Insert}($(u,v)$,~$c$): inserts the edge $(u,v)$ with
								capacity $\capac(u,v)=c$.

We obtain \emph{the first dynamic algorithms with subpolynomial update times}
for this problem, which exceed the edge capacities only by a small subpolynomial
factor.
\begin{restatable}{theorem}{ssldynamic}
\label{thm:ssl-dynamic}
	Let $\varepsilon > 0$. Let $G=(V,E,\capac,d)$ be a graph with $n$~vertices
	and $m$~edges that is undergoing the update operations given above. Then for
	the simultaneous source location problem we can maintain:
	\begin{itemize}
		\item A $(1+\varepsilon,n^{o(1)})$-approximation with
			amortized update time $n^{o(1)}/\varepsilon^2$ and preprocessing
			time $O(n^2/\varepsilon^2)$ if all edge capacities are~$1$.
		\item A $(1+\varepsilon,O(\log^4(n)))$-approximation with
			worst-case update time $\tO(1/\varepsilon^2)$ and preprocessing time
			$\tO(m)$ if we only allow the update operation SetDemand($v$,~$d$).
		\item A $(1+\varepsilon,O(\log^2(n)\lg\lg(n)))$-approximation with
			worst-case update time $\tO(1/\varepsilon^2)$ and preprocessing time
			$\poly(n)$ if we only allow the update operation
			SetDemand($v$,~$d$).
		\item A $(1+\varepsilon,1)$-approximate solution with worst-case update
			time $\tO(h^3/\varepsilon^2)$ and preprocessing time
			$O(n^2/\varepsilon^2)$ if $G$ is a tree of height $h$.
	\end{itemize}
\end{restatable}

To obtain these results, we use a similar DP approach as the one used by Andreev
et al.~\cite{andreev09simultaneous}. Interestingly, the DP function that we use
essentially computes \emph{the inverse function} of the one used by Andreev et
al. We sketch the details of this approach in Section~\ref{sec:inverting}.
After making these changes, the theorems become straightforward applications of
our data structure for maintaining DPs with monotone rows.

\smallskip
\textbf{Organization of Our Paper.} In Section~\ref{sec:approx-conv} we provide
the details of our condition for DPs with monotone rows. In
Section~\ref{sec:knapsack} we present our results for 0-1~Knapsack which nicely
illustrate the applicability of our black-box framework from
Section~\ref{sec:approx-conv}. We provide a technical overview of our more
involved results for $k$-Balanced Graph Partitioning and for Simultaneous Source
Location in Section~\ref{sec:technical}.  We give an overview of the appendix in
Appendix~\ref{sec:organization-appendix}. In the appendix we also present more
related work and the full proofs of our results.  We present omitted proofs from
the main text in Appendix~\ref{sec:omitted-proofs}.

\smallskip
\textbf{Open Problems and Future Work.} In the future, it will be interesting to
use our framework to obtain more dynamic algorithms based on existing DPs. We
believe that this is interesting both in theory and in practice.
Furthermore, it is intriguing to ask whether our criterion from
Definition~\ref{def:well-behaved} can be generalized.  Indeed, our approach was
built around approximating monotone functions using piecewise constant
functions, which can be viewed as piecewiese degree-0 polynomials.  An
interesting question is whether we can obtain a more general criterion if we
approximate DP rows using pieces of higher-degree polynomials, such as splines.
Results in this direction might be possible; for example, in
Appendix~\ref{sec:generalization} we give a side result for the case when the
functions contain a small number of non-monotonicities and derive a dynamic
algorithm for the $\ell_\infty$-necklace problem.

\section{Maintaining Monotone Dynamic Programming Tables}
\label{sec:approx-conv}
In this section, we introduce our notion of DP tables with \emph{monotone rows}
and the additional technical assumptions that we are making. Then we present our
data structure for efficiently maintaining DP tables that satisfy our
assumptions. In our data structure, we will store the rows of the DP using
piecewise constant functions, which we will introduce first.

\textbf{List Representation of Piecewise Constant Functions.}
Let $t \in \mathbb{R}$,
$W\in[1,\infty)$ and set $\Winfty := \{0\} \cup[1,W] \cup\{+\infty\}$.
A function $f \colon [0,t] \to \Winfty$ is \emph{piecewise constant with $p$ pieces}
if there exist real numbers $0 = x_0 < x_1 < x_2 < \dots < x_p = t$ and numbers
$y_1,\dots,y_p \in \Winfty$ such that on each interval $[x_{i-1},x_i)$, $f$ is
constant and has value $y_i$. More formally, for all $i\in\{1,\dots,p\}$ we have
$f(x) = y_i$ for all real numbers $x\in[x_{i-1},x_i)$ and $f(x_p)=y_p$.  Note
that we need the condition $f(x_p)=y_p$ such that $f$ is defined on the whole
domain.

In the \emph{list representation} of a piecewise constant function $f$, we use a
doubly linked list to store the pairs $(x_1,y_1),\dots,(x_p,y_p)$. We also store
the pairs $(x_i,y_i)$ in a binary search tree that is sorted by the
$x_i$-values, which allows us to compute a function value $f(x)$ in time
$O(\lg p)$ for all $x\in[0,t]$.  In the following, we assume that all piecewise
constant functions we consider are stored in the list representation with an
additional binary search tree.

One of the main observations we use is that many operations on
piecewise constant functions are efficient if there are only few pieces. 
The following lemma shows that several operations can be computed in time almost
linear in the number of pieces of the function, rather than in time depending on
the size of the domain of $f$.\footnote{We note that computing the operations
	themselves can be done in linear time. However, since we also store the
	pairs $(x_i,y_i)$ of the list representations in a binary search tree, the
	running times in the lemma include an additional logarithmic factor.}
For $\delta>0$ and $y\in \Winfty$, we write $\lceil y\rceil_{1+\delta}$ to
denote the smallest power of $1+\delta$ that is at least $y$, i.e.,
$\lceil y\rceil_{1+\delta} = \min\{(1+\delta)^i : (1+\delta)^i \geq y,\,\, i\in\mathbb{N}\}$;
we follow the convention that $\lceil 0 \rceil_{1+\delta} = 0$ and
$\lceil \infty \rceil_{1+\delta} = \infty$.
\begin{lemma}
\label{lem:operations}
	Let $t \in \mathbb{R}$ and
	$c\in\mathbb{R}_+$.
	Let $g, h : [0,t] \to \Winfty$ be monotone and piecewise constant functions
	with $p_g$ and $p_h$ pieces, resp. Then we can compute the following
	functions:
	\begin{itemize}
		\item $f_{\min}(x):=\min\{g(x),h(x)\}$ with at most
                $p_g+p_h$ pieces in time $O((p_g+p_h)\log(p_g+p_h))$;
		\item $f_{\shift}(x):=g(x-c)$ for $x\ge c$,
                $f_{\shift}(x)=g(0)$ for $x<c$
                with at most $p_g$
                pieces in time $O(p_g \log(p_g))$;
		\item $f_{\add}(x) := g(x)+h(x)$, with at most
                $p_g+p_h$ pieces in time $O((p_g+p_h)\log(p_g+p_h))$;
		\item $f_{\round}(x) := \lceil g(x) \rceil_{1+\delta}$ for $\delta > 0$
                with at most $2 + \lceil \log_{1+\delta}(W) \rceil$ pieces
				in time $O(p_g \log(p_g))$.
	\end{itemize}
\end{lemma}

\noindent
Note that if we set $\tilde{f} = \lceil f \rceil_{1+\delta}$
then $\tilde{f}$ is a $(1+\delta)$-approximation of $f$ in the following sense.
For $\alpha > 1$, we say that a function $\tilde{f} \colon [0,t] \to
\Winfty$ \emph{$\alpha$-approximates} a function $f \colon [0,t] \to \Winfty$ if
for all $x\in[0,t]$,
\begin{align}
\label{eq:approximation}
	f(x) \leq \tilde{f}(x) \leq \alpha \cdot f(x).
\end{align}
Furthermore, if $f$ is monotone then the rounded function
$\tilde{f}$ contains at most $O(\lg_{1+\delta}(W))$ pieces.  This will be
crucial later because this ensures that, if we perform a single rounding
operation for each row of our DP table, the resulting function will have few
pieces and operations on the function can be performed efficiently.

Next, consider functions $f_1,f_2 : [0,t] \to \Winfty$.  A function
$f:[0,t]\to\Winfty$ is the \emph{$(\min,+)$-convolution $f_1 \minconv f_2$}
if for all $x\in[0,t]$,
   $f(x)
	= (f_1 \minconv f_2)(x)
	:= \min_{\bar{x}\in[0,x]} f_1(\bar{x}) + f_2(x-\bar{x})$.
Such convolutions are highly useful for the computation of many DPs.
The following lemma shows that we can efficiently compute the convolution of
piecewise constant functions.
\begin{lemma}
\label{lem:convolution}
	Let $f_1,f_2 : [0,t] \to \Winfty$ be piecewise constant functions with at
	most $p$ pieces and assume that one of them is monotonically decreasing.
	Then we can compute the function $f : [0,t] \to \Winfty$ with
	$f = f_1 \minconv f_2$ in time $O(p^2 \lg p)$ and $f$ is a piecewise
	constant function with $O(p^2)$ pieces.
	Furthermore, after computing $f$, for any $x\in[0,t]$ we can return a value
	$\bar{x}^*\in[0,t]$ such that $f(x) = f_1(\bar{x}^*) + f_2(x-\bar{x}^*)$ in
	time $O(\lg p)$.
\end{lemma}
Now observe that Lemma~\ref{lem:convolution} has a drawback for our approach:
The number of pieces (i.e., the complexity of the functions) grows quadratically
with every application. An important property which can be used to mitigate this
issue is that the result of the convolution is still a monotone function, as we show in
Lemma~\ref{lem:monotone-convolution-is-monotone} in the appendix.
Later, to keep the number of pieces in our functions small, after each
convolution that we perform via Lemma~\ref{lem:convolution} (and that might grow
the number of pieces quadratically), we perform a rounding operation
$\lceil\cdot\rceil_{1+\delta}$ (see Lemma~\ref{lem:operations}).  This loses a
factor $1+\delta$ in approximation but guarantees that the resulting function
has $O(\lg_{1+\delta}(W))$ pieces. This will be crucial to ensure
that our functions have only few pieces.

\textbf{Maintaining DPs With Monotone Rows.}
Next, we introduce our DP scheme formally.  We consider DP tables with a finite
set of rows~$\I$ and a set of columns~$\J$, with entries taking values in
$\Winfty$.	We will consider DP tables as functions $\DP \colon \I \times \J \to
\Winfty$.\footnote{
	Even though our definition may suggest that we only consider two-dimensional
	DP tables, we do not require an order on $\I$ and we allow $\I$ to be any
	finite set. For example, in Section~\ref{sec:balanced-partitioning} we will
	set $\I$ to 3-tuples corresponding to the parameters of a four-dimensional DP.
} Further, we will associate the $i$'th row of the DP with a function
$\DP(i,\cdot) \colon \J \to \Winfty$, and we store each such function
$\DP(i,\cdot)$ using piecewise constant functions from above.

Next, we introduce the dependency graph for the \emph{rows} of our DP. More
concretely, the \emph{dependency graph} $D=(\I,E_D)$ is a directed graph that
has the rows~$\I$ as vertices and a directed edge $(i',i)$ between two rows if
for some columns $j,j'\in\J$ the entry $\DP(i',j')$ is required to compute
$\DP(i,j)$. We write $\In(i) = \{i'\in\I \colon (i',i)\in
E_D\}$ to denote the set of rows~$i'$ that are required to compute row~$i$.
For the rest of the paper we will assume that the dependency graph
is a DAG, which is the case for all applications that we study.
We will also write $\Reach(i)$ to denote the set of vertices that are reachable
from row~$i$ in $D$.

Since we assume that the dependency graph is a DAG, we can compute the $i$'th
DP row as soon as we have computed the solutions for the DP rows in $\In(i)$. 
We assume that this is done via a \emph{procedure} $\Procedure_i$ that takes as input the
DP rows $\DP(i',\cdot)$ for all $i'\in\In(i)$ and returns the row
$\DP(i,\cdot) = \Procedure_i( \{ \DP(i',\cdot) \colon i'\in\In(i)\} )$.

Next, we come to our condition which encodes when our scheme applies.  In the
definition and for the rest of the paper, we write $\ADP$ to refer to an
approximate DP table, which approximates the exact DP table $\DP$. Let
$\beta>1$. We say that $\ADP$ \emph{$\beta$-approximates} $\DP$ if
$\DP(i,j)\leq \ADP(i,j)\leq \beta\DP(i,j)$ for all $i\in\I,j\in\J$.

\begin{definition}
\label{def:well-behaved}
A DP table is $(h,\alpha,p)$-\emph{well-behaved} if it satisfies the following
conditions:
\begin{enumerate}
	\item (Monotonicity:) For all $i\in\I$, the function $\DP(i,\cdot)$ is monotone.
	\item (Dependency graph:) The dependency graph is a DAG
		and $\abs{\Reach(i)} \leq h$ for all $i\in\I$.
	\item (Sensitivity:) 
		Suppose $\beta>1$ and for all $i'\in\In(i)$, we obtain a $\beta$-approximation
		$\ADP(i',\cdot)$ of $\DP(i',\cdot)$.
		Then applying $\Procedure_i$ on the $\ADP(i',\cdot)$ yields
		a $\beta$-approximation of $\DP(i,\cdot)$, i.e.,
		\begin{align*}
			\DP(i,\cdot)
				\le \Procedure_i(\{ \ADP(i',\cdot) \colon i'\in\In(i) \} )
				\le \beta\cdot\DP(i,\cdot).
		\end{align*}
	\item (Pieces:) For each procedure $\Procedure_i$ there exists an approximate
		procedure $\tilde{\Procedure_i}$ such that:\\
		(a)~$\tilde{\Procedure}_i(\{ \ADP(i',\cdot) \colon i'\in\In(i) \} )$ is an
		$\alpha$-approximation of $\Procedure_i(\{ \ADP(i',\cdot) \colon i'\in\In(i) \})$, \\
		(b)~$\tilde{\Procedure}_i$ can be computed as the composition of
		a constant number of operations from Lemma~\ref{lem:operations} and and
		at most one application of Lemma~\ref{lem:convolution}, and \\
		(c)~$\tilde{\Procedure_i}$ returns a monotone piecewise constant
		function with at most $p$~pieces.
\end{enumerate}
\end{definition}
The definition is motivated in the following way: our operations on the
piecewise constant functions have efficient running times when the functions are
monotone and have few pieces. This is ensured by Properties~(1), 4(b), and
4(c).  Next, rounding errors cannot compound too much if each row can only reach
$h$~other rows and the sensitivity condition is satisfied. This is ensured by
Properties~(2), (3), and 4(a).

Even though the definition might look slightly technical at first glance, it
applies in many settings. In particular, Property~(2) is satisfied when the
dependency graph is a rooted tree of height~$h$ in which all edges point towards
the root; this is the case in all of our applications.
The other conditions are immediately satisfied by our DP for
0-1~Knapsack in Section~\ref{sec:knapsack} %
and the DP for simultaneous source location in
Section~\ref{sec:simultaneous-source-location}.  However, our DP for balanced
graph partitioning violates Property~(4b) of Definition~\ref{def:well-behaved}.
Hence, we will also consider a weaker assumption in
Section~\ref{sec:dps-on-trees} which, however, will not allow for nice black-box
results, such as Theorems~\ref{thm:well-behaved-static}
and~\ref{thm:well-behaved-dynamic} below.

\smallskip
Next, we state our main results. They imply that we obtain static
$(1+\varepsilon)$-approximation algorithms running in near-linear time and
space for $(\tO(1),\ln(1+\varepsilon)/\tO(1),\tO(1))$-well-behaved DPs. They
also show that under this assumption, we can dynamically maintain
$(1+\varepsilon)$-approximate DP solutions with polylogarithmic update times.

Our main theorem for static algorithms is as follows.
\begin{theorem}
\label{thm:well-behaved-static}
	Consider an $(h,\alpha,p)$-well-behaved DP. Then we can compute an
	approximate DP table $\ADP$ which $\alpha^{h+1}$-approximates $\DP$ in time and
	space $O(\abs{\I} \cdot p^2\log(p))$.
\end{theorem}
Later, we will apply the theorem to DPs with dependency trees of logarithmic
heights $h=O(\lg n)$, we will set the
approximation ratio to $\alpha=\ln(1+\varepsilon)/(h+1)$, and
the number of pieces
to $p=\polylog(W)$. This will yield our desired algorithms with near-linear
running time $\tO(\abs{\I})$ and space usage. Note that this is a big
improvement upon the brute-force running times and space usages of
$\Omega(\abs{\I}\cdot\abs{\J})$.

The proof of the theorem follows from observing that when moving from one vertex
to another in the dependency graph, we lose a multiplicative $\alpha$-factor in
the approximation ratio; as each vertex can only reach $h$~other vertices, this
will compound to at most $\alpha^{h+1}$. Combining the assumptions on the
functions $\tilde{\Procedure}_i$ and the results from
Lemmas~\ref{lem:operations} and~\ref{lem:convolution}, we get that each row
$\ADP(i,\cdot)$ can be computed in time $O(p^2\log(p))$ which gives
$O(\abs{I}\cdot p^2\log(p))$ total time.

We also give the following extension to the dynamic setting which shows that if
one of the DP rows changes, we can update \emph{the entire table} efficiently.
\begin{theorem}
\label{thm:well-behaved-dynamic}
	Consider an $(h,\alpha,p)$-well-behaved DP and suppose that row~$i$ is
	changed. Then we can update our approximate DP~table
	$\ADP$ such that after time $O(h \cdot p^2\log(p))$ it is an
	$\alpha^{h+1}$-approximation of $\DP$.
\end{theorem}
As before, we will typically use the theorem with $h=O(\lg n)$,
$\alpha=\ln(1+\varepsilon)/(h+1)$ and $p=\polylog(W)$. This will then result in
our desired polylogarithmic update times.  Note that this is a significant
speedup compared to storing the DP tables using two-dimensional arrays: in that
case even updating \emph{a single row} would take time $\Omega(\abs{\J})$, which
in many applications would already be linear in the size of the input.

The theorem follows from observing that after a row $i$ changes, we only have to
update those rows which can be reached from $i$ in the dependency graph. But
these can be at most~$h$ and each of them can be updated in time $O(p^2\log(p))$ by
Lemmas~\ref{lem:operations} and~\ref{lem:convolution}.

\section{Fully Dynamic Knapsack}
\label{sec:knapsack}

In \emph{0-1 knapsack}, the input consists of a knapsack size $B\in\mathbb{R}_+$ and a
set of $n$~items, where each item $i\in[n]$ has a weight $w_i\in\mathbb{R}_+$
and a price $p_i\in[1,\infty)$. The goal is to find a set of items $I$ that
maximizes $\sum_{i\in I} p_i$ while satisfying the constraint $\sum_{i\in I} w_i
\leq B$. For a set of items $I\subseteq[n]$, we refer to the sum $\sum_{i\in I}
w_i$ as \emph{the weight of $I$}.

For the rest of this section we set $W=\sum_i p_i$ and $t=\sum_{i\in[n]} w_i$. 

Next, we first derive a dynamic algorithm with update time
$\tO(\lg^3(n)\lg^2(W)/\varepsilon^2)$ which is based on our framework for DPs
with monotone rows.  Then we will use this algorithm as a subroutine to obtain a
faster algorithm with update time $\tO(\lg^2(nW)/\varepsilon^2)$ in
Section~\ref{sec:knapsack-small-instance}; this will prove
Theorem~\ref{thm:knapsack}.
\knapsack*

Below we will also show that we can return the maintained solution~$I$ in time
$O(\abs{I})$ and that we answer queries whether a given item $i\in[n]$ is
contained in $I$ in time $O(1)$. This matches the query times
of~\cite{eberle21fully}.

\subsection{Knapsack via Convolution of Monotone Functions}
\label{sec:knapsack-chan}
First, we give a brief recap of the knapsack approach by
Chan~\cite{chan18approximation}. We consider the more general problem of
approximating the function $f_J \colon [0,t] \to \mathbb{R}_+$, where
$J\subseteq [n]$ is a set of items and
\begin{align}
\label{eq:knapsack-general}
	f_J(x)
	= \max\left\{ \sum_{i\in I} p_i \colon
   				\sum_{i\in I} w_i \leq x, \enspace
				I \subseteq J \right\}.
\end{align}
Intuitively, the value $f_J(x)$ corresponds to the best possible
knapsack solution if we can only pick items which are contained in $J$ and if
the weight of the solution can be \emph{at most}~$x$. Therefore, $f_{[n]}(B)$
corresponds to the optimum solution of the global knapsack instance.

Note that for each $J\subseteq[n]$, $f_J(x)$ is a monotonically increasing
piecewise constant function: Indeed, consider $x'\leq x$. Any solution
$I\subseteq J$ that is feasible for $x'$ (i.e., the weight of $I$ is at
most~$x'$) is also a feasible solution for~$x$.  Thus, $f_J(x')\leq f_J(x)$ and,
therefore, $f_J$ is monotonically increasing. Furthermore, $f_J$ is piecewise
constant since each function value $f_J(x)$ corresponds to a solution
$I\subseteq J$ and the number of choices for $I\subseteq J$ is finite.

Next, note that if we have two disjoint subsets $J_1,J_2\subseteq [n]$ then it
holds that $f_{J_1 \cup J_2}$ is the $(\max,+)$-convolution of $f_{J_1}$ and
$f_{J_2}$, i.e., for all $x$ it holds that
\begin{align*}
	f_{J_1\cup J_2}(x)
	= \max_{\bar{x}} f_{J_1}(\bar{x}) + f_{J_2}(x-\bar{x}).
\end{align*}
This can be seen by observing that for each $x$, the optimum solution $I$ for
the instance $J_1\cup J_2$ with weight at most $x$ can be split into two
disjoint solutions $I_1\subseteq J_1$ and $I_2\subseteq J_2$ such that $I_1$ has
weight $\bar{x}$ and $I_2$ has knapsack weight at most $x-\bar{x}$ (for suitable
choice of $\bar{x}\in[0,x]$).  We conclude that if we have two knapsack
instances over disjoint sets of items $J_1$ and $J_2$, then we compute the
solution for the knapsack instance with items $J_1\cup J_2$ by computing the
$(\max,+)$-convolution of $f_{J_1}$ and $f_{J_2}$.

\textbf{The Exact DP.}
The previous paragraphs imply a simple way of computing the exact solution of a
knapsack instance: For each item $i\in[n]$, compute the function $f_{\{i\}}$ and
then recursively merge the solutions for sets of size $2^j$,
$j=1,\dots,\lceil\lg n\rceil$, by computing $(\max,+)$-convolutions until we
have computed the global solution $f_{[n]}$. We perform the recursive merging of
the solutions using a balanced binary tree, resulting in a tree of height $O(\lg n)$.

More concretely, we build a rooted balanced binary tree $T$ with $n$ leaf nodes,
where all edges point towards the root.  We have one leaf $f_{\{i\}}$ for each
item~$i$.  Each internal node~$u$ in~$T$ is associated with a function $f_{J_u}$
as per Equation~\eqref{eq:knapsack-general}, where $J_u$ is the set of all items
in the subtree rooted at~$u$. To simplify notation, we will also refer to
$f_{J_u}$ as $f_u$.

Now we consider the exact computation of the DP. This will reveal the
procedures $\Procedure_i$ from Definition~\ref{def:well-behaved}.
As base case, for each $i\in[n]$, the $i$'th leaf of $T$ contains the function
$f_{\{i\}}$, which is a piecewise constant function that has value $0$ on the
interval $[0,w_i)$ and value $p_i$ on the interval $[w_i,t]$.

Next, in each internal node~$u$ of $T$ with children $u_1$ and $u_2$, we set
$f_u$ to the $(\max,+)$-convolution of $f_{u_1}$ and $f_{u_2}$. By induction it
can be seen that for every node $u$ in $T$, it holds that $J_u=J_{u_1}\cup
J_{u_2}$ and thus $J_u$ is the set of all items whose corresponding leaf is
contained in the subtree $T_u$. Hence, for the root~$r$ of $T$ it holds that
$f_r = f_{[n]}$ and $f_r(B)$ is the optimal solution for the global
knapsack instance.

In the following, we check that our DP satisfies Properties~(1--3) of
Definition~\ref{def:well-behaved}.

First, note that the tree~$T$ from above is also the dependency graph of our DP.
Hence, our DP has a row for every vertex of $T$ and thus $O(n)$~rows in total.
Furthermore, since $T$ has height~$O(\lg n)$ and all edges point towards the
root, every vertex can reach at most $h=O(\lg n)$~vertices. Hence, Property~(2)
of Definition~\ref{def:well-behaved} is satisfied.

Second, we observe that in both cases above, the function $f_{\{i\}}$ and $f_u$ which
correspond to the rows of our DP table are monotonically increasing (we argued
this above for all functions $f_J$). Thus, Property~(1) is satisfied.

Third, observe that Property~(3) is also satisfied since
$(\max,+)$-convolution satisfies our sensitivity condition.

We conclude that the first three properties of Definition~\ref{def:well-behaved}
are satisfied. Unfortunately, this does not yet imply that we can obtain
efficient algorithms: Note that if we compute the exact DP bottom-up then we
compute one convolution per node and thus the total running time of this
approach is $O(n\cdot t(p))$, where $p$~is an upper bound on the number of
pieces in our functions and $t(p)$ is the time it takes to compute a
$(\max,+)$-convolution of two functions with $p$~pieces. However, observe that
computing the convolutions can potentially take a large amount of time because
the number of pieces of the functions might grow quadratically after each
convolution (see Lemma~\ref{lem:convolution}).  We will resolve this issue
below using rounding.

\textbf{The Approximate DP.}
Next, we consider approximations which will reveal the functions
$\tilde{\Procedure}_i$ from Definition~\ref{def:well-behaved}.

First, note that we need to compute $(\max,+)$-convolutions of monotonically
increasing functions efficiently. We observe that this can be done efficiently
using our subroutine from Lemma~\ref{lem:convolution} for the
$(\min,+)$-convolution of monotonically decreasing functions:  Indeed, suppose
that $f$ is the $(\max,+)$-convolution of two monotonically increasing
functions~$g$ and $h$, then for all $x$ it holds that
\begin{align*}
	f(x)
	= \max_{\bar{x}} \{ g(\bar{x}) + h(x-\bar{x}) \}
	= - \min_{\bar{x}} \{-g(\bar{x}) + (-h(x-\bar{x})) \}.
\end{align*}
Now observe that $-g$ and $-h$ are monotonically decreasing functions and,
therefore, $f = - ( (-g) \minconv (-h))$, where $\minconv$ denotes the
$(\min,+)$-convolution. Thus, we can use the efficient routine for
$(\min,+)$-convolution from Lemma~\ref{lem:convolution} with the same running
time.\footnote{We note that, formally, Lemma~\ref{lem:convolution} can only be
	applied on functions with non-negative values. However, this can be achieved
	by adding a number $C$ to $-g$ and $-h$, which is an upper bound on the
	values taken by $g$ and $h$, and at the end we 
	subtract the constant function $2C$, i.e., we set
	$f = - ( (-g+C) \minconv (-h+C)) - 2C$.
}

Now we can define the subroutines $\tilde{\Procedure}_i$. Let $\delta>0$ be a
parameter that we set later. Whenever we compute a function $f_u$ via a
$(\max,+)$-convolution, we use the efficient subroutine from
Lemma~\ref{lem:convolution}. After computing the convolution, we set $f_u =
\lceil f_u \rceil_{1+\delta}$ via the subroutine from
Lemma~\ref{lem:operations}.

Observe that this approach satisfies Property~(4a) of
Definition~\ref{def:well-behaved} with $\alpha=1+\delta$. 
Furthermore, Property~(4b) is satisfied since we only use a single convolution
and a single rounding step. Finally, Property~(4c) is also satisfied because the
resulting function is monotone and has $p=O(\log_{1+\delta}(W))$ after the
rounding.

The above arguments show that our DP is $(h,\alpha,p)$-well-behaved for 
$h=\lceil\lg n\rceil$, $\alpha=1+\delta$,
$\delta=\ln(1+\varepsilon)/\lceil\log n\rceil$
and $p=O(\log_{1+\delta}(W)) = O(\log(W)/\delta)$.
Hence, Theorem~\ref{thm:well-behaved-dynamic} immediately implies the following
lemma.
\begin{lemma}
\label{lem:knapsack-ptas}
	Let $\varepsilon>0$.
	There exists an algorithm that computes a $(1+\varepsilon)$-approximate
	solution for 0-1 knapsack in time
	$n\cdot \frac{1}{\varepsilon^2} \log^2(n) \log^2(W)
			\cdot \polylog(\frac{1}{\varepsilon}\log(nW))$.
\end{lemma}

We note that we can return our solution~$I$ in time
$\abs{I} \lg(n) \cdot \polylog(\frac{1}{\varepsilon}\log(nW))$
as follows. Recall that our global objective function value is
achieved by $f_r(B)$ and that $f_r(B) = f_{u_1}(\bar{x}^*) +
f_{u_2}(B-\bar{x}^*)$, where $u_1$ and $u_2$ are the nodes below the root
node~$r$ of the dependency tree. Now using the second part of
Lemma~\ref{lem:convolution} we can get the value of $\bar{x}^*$ in time
$O(\lg p)$. If $f_{u_1}(\bar{x}^*)>0$ we recurse on $f_{u_1}(\bar{x}^*)$ and if
$f_{u_2}(B-\bar{x}^*)>0$ we also recurse on
$f_{u_2}(B-\bar{x}^*)$. At some point we will reach a leaf node~$i$ and we
include $i$ in the solution iff $f_{\{i\}}(x)>0$. Note that since we only recurse
for function values which are strictly larger than zero, for each item that we
include into the solution we have to follow a single path in the dependency tree
of height $O(\lg n)$ and our work in each internal node is bounded by $O(\lg p)$.
This gives the total time of $O(\abs{I}\lg(n)\lg(p))$ and our claim follows from
our choice of~$p$ above.

\textbf{Extension to the Dynamic Setting.}
Next, we extend our result to the dynamic setting.
For the sake of simplicity, we assume that $n$ is
an upper bound on the maximum number of available items (items in $S$)
and given to our algorithm in the beginning.\footnote{It is possible to drop
	this assumption using an amortization argument. More concretely, every time
	the number of items is less than $n/2$ or more than $n$, we rebuild the
	data structure with a new value of~$n$. Each rebuild can be done in time
	$O(n t(n))$, where $t(n)$ is our update time. Since this only happens after
	$\Omega(n)$ updates occured, we can amortize this cost over the updates that
	appeared since the last rebuild.}
We consider update
operations that insert and delete items from the set.  More concretely, we
consider the following update operations:
\begin{itemize}
	\item \emph{insert}($p_i,w_i$), in which $i$ is added to $S$ by setting the
	price and weight of item~$i$ to $p_i\in\Winfty$ and $w_i\in\mathbb{R}_+$,
	respectively, and
	\item \emph{delete}($i$), where item~$i$ is removed from the set of items. 
\end{itemize}

Our implementation is as follows. In the preprocessing phase, we build the same
tree~$T$ as above and use the subroutine from above to compute the function
$f_{\{i\}}$. For the operation \emph{delete}($i$), we set $p_i=0$ and $w_i=0$,
which changes exactly one row of our DP table.  For the operation
\emph{insert}($p_i,w_i$), we set the price and weight of item $i$ to $p_i$ and
$w_i$, resp., which again changes a single row in our DP table.  After changing
such a row, we recompute the global DP solution via
Theorem~\ref{thm:well-behaved-dynamic}.  Since the DP is
$(h,\alpha,p)$-well-behaved with the same parameters as above, the theorem
implies the following proposition.

\begin{restatable}{proposition}{knapsackProposition}
\label{prop:dynamic-knapsack-chan}
	Let $\varepsilon>0$. There exists an algorithm for the fully dynamic
	knapsack problem
	that maintains a $(1+\varepsilon)$-approximate solution with worst-case
	update time
	$\frac{1}{\varepsilon^2} \log^3(n) \log^2(W) \cdot \polylog\left( \frac{1}{\varepsilon} \log(nW)\right)$.
\end{restatable}

Observe that with the same procedure as for the static algorithm, we can return
our solution~$I$ in time 
$\abs{I} \lg(n) \cdot \polylog(\frac{1}{\varepsilon}\log(nW))$.
Furthermore, given an item $i\in[n]$, we can return whether $i\in I$ in time
$\lg(n) \cdot \polylog(\frac{1}{\varepsilon}\log(nW))$.
This can be done by using the same query procedure as in the static setting,
where we only recurse on the unique subtree in the depedency tree that contains
the node for item~$i$.

We note that the above proposition already improves upon the update time in the
result of Eberle et al.~\cite{eberle21fully} in terms of the dependency on
$\frac{1}{\varepsilon}$ but it has a worse dependency on $\log(nW)$. However,
our query time is slower than the $O(1)$-time query operation
in~\cite{eberle21fully}. We will resolve these issues in the next subsection,
where we will use the algorithm from
Proposition~\ref{prop:dynamic-knapsack-chan} as a subroutine.

\subsection{Dynamically Maintaining a Small Instance}
\label{sec:knapsack-small-instance}
Next, we we obtain a faster dynamic algorithm with update time
$\tO(\frac{1}{\varepsilon^2} \log^2(nW))$ by combining the algorithm from
Proposition~\ref{prop:dynamic-knapsack-chan} and with ideas from Eberle et
al.~\cite{eberle21fully}.
Our high-level approach is as follows. First, we partition the items into
a small number of \emph{price classes}. Then we take a few items of small weight
from each price class. This will give a very small knapsack instance~$X$ for
which we maintain an almost optimal solution using the subroutine from
Proposition~\ref{prop:dynamic-knapsack-chan}; since this instance is very small
(i.e., $\abs{X} \ll n$), the update time for maintaining this instance
essentially becomes $O(\frac{1}{\varepsilon^2} \log^2(W))$, i.e., we lose the
$O(\lg^3 n)$ term that made the update time in the proposition too costly.  For
the rest of the items which are not contained in~$X$, we show that we can
compute a good solution using fractional knapsack, which can be easily solved
using a set of binary search trees. Then it remains to show that the combination
of the two solutions is a $(1+\varepsilon)$-approximation.

The main differences of our algorithm and the one by Eberle et
al.~\cite{eberle21fully} are as follows. Eberle et al.\ also partition the items
into a small number of price classes. They also combine solutions for a small
set of heavy items~$X$ and solutions based on fractional knapsack for the other
items.  However, they have to enumerate many different sets~$X$ and they also
guess the approximate price of the fractional knapsack solution; more
concretely, they enumerate $\Theta(\frac{1}{\varepsilon^2} \lg(W))$ choices
for~$X$ and the number of guesses they have to make for the fractional knapsack
solution is $\Theta(\frac{1}{\varepsilon} \lg(W))$. Thus they have to consider
$\Theta(\frac{1}{\varepsilon^3} \lg^2(W))$ guesses and for each of them they
have to compute approximate solutions, which takes time
$\Theta(\frac{1}{\varepsilon^4})$ for each~$X$ since they have to run a static
algorithm from scratch. In our approach, we only have to consider a single
set~$X$ which we maintain in our data structure from
Proposition~\ref{prop:dynamic-knapsack-chan}, which saves us a lot of time.
Furthermore, the piecewise constant function, in which we store the solution for
$X$, essentially ``guides'' our $\Theta(\frac{1}{\varepsilon} \lg(W))$ guesses
for the weight of fractional knapsack solution.  In our analysis we have to be
slightly more careful to ensure that our guesses for the weight of the
fractional knapsack solution guarantee the correct approximation ratio.

\textbf{Definitions.}
We assume that $\varepsilon < 1$ and that $1/\varepsilon$ is an integer. More
concretely, we run the algorithm with
$\varepsilon' = \max\{ \frac{1}{i} \colon \frac{1}{i} \leq \varepsilon, i\in\mathbb{N} \}$.
Set $L = \lceil \lg_{1+\varepsilon}(W) \rceil$ and 
recall that we set $W=\sum_{i} p_i$. %

We define the \emph{price classes}
$V_\ell = \{ i \colon (1+\varepsilon)^\ell \leq p_i < (1+\varepsilon)^{\ell+1} \}$. 
In the following, we assume that all items from price class $V_\ell$ have
price exactly $(1+\varepsilon)^{\ell+1}$. We only lose a factor of
$1+\varepsilon$ by making this assumption.
Furthermore, we set $V^{1/\varepsilon}_\ell$ to the set of $1/\varepsilon$ items
from $V_\ell$ with smallest weights $w_i$ (breaking ties arbitrarily). We also define
$V_\ell' = V_\ell \setminus V^{1/\varepsilon}_\ell$.

Next, we set $X = \bigcup_{\ell\geq 0} V^{1/\varepsilon}_\ell$ and 
$Y = \bigcup_{\ell\geq 0} V_\ell'$ for all $\ell\geq0$. Note that $X$ and $Y$
partition the set of
items and $\abs{X} = \frac{1}{\varepsilon} \cdot L = O(\varepsilon^{-2} \log(W))$.

Now our strategy is to use our algorithm from
Proposition~\ref{prop:dynamic-knapsack-chan} to maintain a solution for the
items in~$X$. Then we show how we can combine the solution for~$X$ with a
solution for~$Y$ that is based on fractional knapsack and a charging argument.

\textbf{Data Structures.}
For each $\ell\in[L]$, we maintain $V_\ell$ sorted non-decreasingly by weight.

We also maintain the set~$X$ in a binary search tree, in which we sort the items
by their index, and we maintain our data structure from
Proposition~\ref{prop:dynamic-knapsack-chan} on the items in~$X$.

Furthermore, let $U_\ell = \bigcup_{\ell'\leq\ell} V_{\ell'}'$ denote the set of
all items that are not contained in $X$ and of price class at most~$\ell$. For
each $\ell$, we maintain the set~$U_\ell$ in a binary search tree~$T$ in which
the items are stored as leaves and sorted by their density $\frac{p_i}{w_i}$. In
each internal node~$u$ of~$T$, we store the total weight of the items in the
subtree~$T_u$ rooted at~$u$ and the total profit of the items in~$T_u$. Observe
that this allows us to answer queries of the type: ``Given a budget~$b$, what is
the value of the optimal fractional\footnote{
   In fractional knapsack, we may use items fractionally.
   An optimal solution is achieved by sorting the items
   items by their density and greedily adding items to the solution until we
   have used up our budget~$b$. This approach uses at most one item
   fractionally (namely, the one at which we use up our budget).}
knapsack solution in $U_{\ell}$ with weight
at most~$b$?'' in time $O(\lg n)$.

\textbf{Updates.}
Now consider an item insertion or deletion and suppose that the updated item is
of price class $V_\ell$. We first update the sets $V_\ell$, $U_{\ell'}$ for
$\ell'\leq\ell$ and the sets~$X$ and~$Y$. Note that for each of these sets at
most one item can be removed and inserted. Thus, these steps can be done in time
$O(\ell \cdot \lg(n)) = O(\varepsilon^{-1} \lg(W)\lg(n))$.

Next, if $X$ changed in the previous step, then we also perform the
corresponding updates in the data structure from
Proposition~\ref{prop:dynamic-knapsack-chan}. Since
$\abs{X}=O(\varepsilon^{-2}\log(W))$ holds by construction of~$X$, the update
operations for the data structure maintaing the knapsack solution for~$X$ take a
total time of
\begin{align*}
	&O\left(\varepsilon^{-2} \log^3(\abs{X}) \log^2(W) \cdot
				\polylog\left(\frac{1}{\varepsilon}\log(\abs{X}W)\right)\right) \\
	&= O\left(\varepsilon^{-2} \log^2(W) \cdot
				\polylog\left(\frac{1}{\varepsilon}\log(nW)\right)\right).
\end{align*}
Furthermore, we can explicitly write down our solution~$I_X$ for the items in
$X$ in time $\varepsilon^{-2}\log(W) \cdot \polylog(\frac{1}{\varepsilon}\log(nW))$
since $\abs{X}=O(\varepsilon^{-2}\log(W))$. Also, for each $i\in I_X$, we can
set a bit indicating that $i\in I_X$. Note that the time for writing down $I_X$
and setting the bits is subsumed by the update time above.

\textbf{Queries.}
\emph{Returning the value of a solution:}
We return the value of a global knapsack solution as follows.

Consider the data structure from Proposition~\ref{prop:dynamic-knapsack-chan}
which maintains a solution for the items in~$X$. Note that this solution is
stored as a piecewise constant function with $p\leq L$~pieces and consider the
list representation $(x_1,y_1),\dots,(x_p,y_p)$ of this function.

Our strategy is as follows: For each $i=1,\dots,p$, we consider a solution which
spends budget~$x_i$ on items in~$X$ and budget~$B-x_i$ on items in~$Y$. Then we
take the maximum over all of the solutions we have considered. More concretely,
for given~$i=1,\dots,p$, we obtain our solution as follows. We pick $\ell_i$
such that
$(1+\varepsilon)^{\ell_i} = \lceil \varepsilon \cdot y_i \rceil_{1+\varepsilon}$
(see Lemma~\ref{lem:choice-l-i} below for a justification of this choice).  Now we use
the binary search tree for $U_{\ell_i}$ to find the highest profit that we can
obtain from fractional knapsack on items in~$U_{\ell_i}\subseteq Y$ if we can
spend budget at most $b = B-x_i$. Let $y_i'$ be the value of this query after
removing any profit that we gain from the (at most one) fractionally cut item.
We also store the density of the final item that is contained in the fractional
knapsack solution.  Now we return the maximum of $y_i+y_i'$ over all
$i=1,\dots,p$.

Note that since the solution for~$X$ has at most $L=O(\varepsilon^{-1} \log(W))$
pieces and for each of them we perform a single query in a binary search tree,
the total time for return the solution value is $O(\varepsilon^{-1} \lg(W) \lg(n))$.
Note that this time is subsumed by the update time.

\emph{Returning the entire solution:}
Now we can return our global solution~$I$ in time $O(\abs{I})$ as follows.
Observe that $I$ is composed of the solution $I_X$ for the items in $X$ and of
the items in the fractional knapsack solution. During our updates, we already
stored the items in $I_X$ and can write them down in time $O(\abs{I_X})$. Next,
to return the items from the fractional knapsack solution, recall that we stored
the density of the final item in the fractional knapsack solution. Thus, we only
have to output the items ordered non-decreasingly by their density, while we are
above the desired density-threshold. This can be done in time linear in the size
of the fractional knapsack solution. This is essentially the same query
procedure as in~\cite{eberle21fully}.

\emph{Returning whether an item is in the solution:}
Furthermore, observe that the above implies that we can answer whether an item
$i\in[n]$ is contained in our solution in time $O(1)$: If $i\in X$ then we
already stored a bit whether $i\in I_X$. If $i\not\in X$ then we can check
whether $i$ is in the fractional knapsack solution by checking whether its
density is above or below the threshold given by the final item in the
fractional knapsack solution.

\textbf{Analysis.}
We start by making some simplifications to $\OPT$. We let $\OPT'$ denote the
version of $\OPT$ in which for each $\ell\in[L]$, we pick the
$\abs{\OPT\cap V_\ell}$ items of smallest weight from $V_{\ell}$. This
only loses a factor of $1+\varepsilon$.
Next, define $\OPT_X' = \OPT' \cap X$ and $\OPT_Y' = \OPT' \cap Y$.
Observe that by how we picked $\OPT'$, it holds that
$\OPT_Y'\cap V_\ell \neq \emptyset$ iff
$\abs{\OPT'\cap V_\ell} > 1/\varepsilon$.

Let $p_X$ denote the total price of items in $\OPT_X'$ and let $w_X$ denote the
total weight of the items in $\OPT_X'$. Let~$f$ denote the piecewise constant
function that stores the solution for the items in~$X$.  Observe that by
Proposition~\ref{prop:dynamic-knapsack-chan} we have that
$$ p_X \leq f(w_X) \leq (1+\varepsilon)p_X. $$
Also, the function value $f(w_X)$ is part of a piece
$(x_{i^*},y_{i^*})$ with $x_{i^*} \leq w_X$ and $y_{i^*} = f(w_X)$.

The next lemma justifies why we set $\ell_i$ such that
$(1+\varepsilon)^{\ell_i} = \lceil \varepsilon \cdot y_i \rceil_{1+\varepsilon}$
in our algorithm. To this end, let $\ell_{i^*}$ be such that
$(1+\varepsilon)^{\ell_{i^*}} = \lceil \varepsilon \cdot y_{i^*} \rceil_{1+\varepsilon}$
and let $\ell_Y$ be the price class of the most valuable item in $\OPT_Y'$. In
the lemma we show that $\ell_{i^*}\geq\ell_Y$.  We will use this to show that
our solution for $X$ of profit $y_{i^*}$ is valuable enough such that we can
charge a fractionally cut item from fractional knapsack onto the solution
from~$X$ and only lose a factor of $(1+\varepsilon)^2$.
\begin{lemma}
\label{lem:choice-l-i}
	It holds that $\ell_{i^*} \geq \ell_Y$.
\end{lemma}
\begin{proof}
	Since $\OPT_Y' \cap V_{\ell_Y}' \neq \emptyset$,
	$\abs{\OPT'\cap V_{\ell_Y}} > 1/\varepsilon$ and thus
	$\OPT_X'$ contains all $1/\varepsilon$ items from
	$V^{1/\varepsilon}_{\ell_Y}$. Hence, $p_X \geq \frac{1}{\varepsilon} \cdot
	(1+\varepsilon)^{\ell_Y}$.
	From above we get $f(w_X)=y_{i^*}$ and $f(w_X) \geq p_X$. By choice of
	$\ell_{i^*}$,
	\begin{align*}
		(1+\varepsilon)^{\ell_{i^*}}
		= \lceil \varepsilon \cdot y_{i^*} \rceil_{1+\varepsilon}
		= \lceil \varepsilon \cdot f(w_X) \rceil_{1+\varepsilon}
		\geq \left\lceil \varepsilon \cdot p_X \right\rceil_{1+\varepsilon}
		\geq \left\lceil \varepsilon \cdot \frac{1}{\varepsilon} (1+\varepsilon)^{\ell_Y} \right\rceil_{1+\varepsilon}
		= (1+\varepsilon)^{\ell_Y}.
	\end{align*}
	This implies $\ell_{i^*}\geq \ell_Y$.
\end{proof}

Next, consider the the fractional knapsack solution that we obtain from our
query. Note that this solution has a profit that is at least as large as the
profit of $\OPT_Y'$ (since fractional knapsack is a relaxation of 0-1 knapsack). 
Furthermore, the fractional solution uses at most one item fractionally and this
item is from $U_{\ell_{i^*}}$ and has value at most
$(1+\varepsilon)^{\ell_{i^*}}
= \lceil \varepsilon \cdot y_{i^*} \rceil_{1+\varepsilon}
\leq (1+\varepsilon) \varepsilon \cdot y_{i^*}$.
Thus, we can charge this item on $\OPT_X'$ and lose a factor of at most
$(1+\varepsilon)^2$.

We conclude that the solution $y_{i^*} + y_{i^*}'$ is a
$(1+\varepsilon)^{O(1)}$-approximation of $\OPT$. Combining this with our
previous running time analysis, we obtain Theorem~\ref{thm:knapsack}.

\section{Technical Overview}
\label{sec:technical}
We now present an overview of two techniques for making DPs fit our framework.
We will briefly discuss how we \emph{monotonized} the DP for $k$-balanced
partitioning and how we \emph{inverted} the DP for simultaneous source location.
Due to space constraints, we only present excerpts of our algorithms and we only
consider special cases.  More concretely, for both problems we will consider the
special case when the input graph is a binary tree. In the appendix we will show
that the results can be extended to general graphs.

\subsection{Monotonizing the DP of Feldmann and Foschini}
\label{sec:monotonizing}
We start by considering the $k$-balanced graph partitioning problem.  Recall
that in this problem, the input is a graph $G=(V,E,\capac)$, where
$\capac: E \to \Winfty$ is a weight function on the
edges, and an integer~$k$. As discussed in the introduction, we assume that we
can violate the partition sizes by a $(1+\varepsilon)$-factor and our goal is to
find a partition $V_1,\dots,V_k$ of the vertices such that $\abs{V_i} \leq \lceil
(1+\varepsilon) \abs{V}/k\rceil$ for all $i$ and such that we minimize
$\cut(V_1,\dots,V_k) := \sum_{i=1}^k \sum_{\{u,v\}\in E\cap (V_i\times (V\setminus V_i))} \capac(u,v)$.

For the sake of better exposition, here we only consider the special case in
which \emph{$G$ is a binary tree}; in Appendix~\ref{sec:extension-balanced} we
show how to drop this assumption.

In the following we present a DP in which the rows are monotone and
we show how to efficiently perform operations on these solution vectors using
monotone piecewise constant functions. Our DP is related to the DP by Feldmann
and Foschini~\cite{feldmann15balanced} which is \emph{non-monotone} and thus our
DP can be viewed as the \emph{monotonization} of the DP by Feldmann and
Foschini. We believe that our technique to monotonize the DP will have further
applications in the future.

\textbf{High-Level Description of the DP.}
Our DP is computed bottom-up starting at the leaves of the tree and then moving
up in the tree.  For each vertex $v$, we will compute a DP solution of minimum
cost that encodes whether the edge to the parent~$p$ of $v$ is cut and which
edges shall be cut inside the subtree $T_v$ that is rooted at $v$. Note that the
removal of the cut edges in our solution will decompose the tree into disjoint
connected components and exactly one of them contains $v$'s parent~$p$.
Additionally, we store information about the number of vertices that are
still connected to the parent $p$ (and, therefore, to the outside of $T_v$)
after the cut edges are removed.  We will assume that when we compute the DP
cell for a vertex $v$, we have access to the solutions for both of its children.

More concretely, when we have computed a solution for a subtree $T_v$, i.e., we
know which edges incident to nodes in this subtree we are going to remove (note
that the edge leading to the parent of $v$ is incident to $T_v$ and thus we
consider it as part of this solution), we store the following information in the DP
table.  First, we store its cost, i.e., the total capacity of all edges that are
incident to vertices in $T_v$ and that are cut.  As described above, we would
also like to store the number of vertices that are connected to the parent
of $v$ and the sizes of connected components inside $T_v$. However, there are
two difficulties: (1)~We cannot store the number of vertices that are
connected to the root exactly because this would result in a too large DP table.
Instead, we store the cheapest solution in which vertices of \emph{at most} some
given number are still connected to the parent of $v$. As we will see, this
approach gives rise to \emph{monotonically decreasing} functions and allows for a
very efficient computation of the DP table.  (2)~We store implicitly the
\emph{size} of all connected components that are created after the cut edges are
removed and that lie completely inside $T_v$. As before, storing these sizes
exactly would result in a very large DP table and, therefore, we store them
concisely using the concept of a \emph{signature}. The signatures will help us
to characterize the sizes of the components inside $T_v$ very efficiently.

\textbf{Signatures.} We call a connected component in $T_v$ \emph{large} if it
contains at least $\varepsilon \lceil \abs{V}/k\rceil$ vertices and
otherwise we call it \emph{small}.  Let $t=\lceil
\log_{1+\varepsilon}(1/\varepsilon) \rceil+1$, and let $M = \lceil k/\varepsilon
\rceil+1$. A \emph{signature} is a vector $g=(g_0,\dots,g_{t-1}) \in [M-1]^t$.
Observe that each $\Procedure_i$ is an integer between $0$ and $M-1$ and hence there are
$M^t = (k/\varepsilon)^{O(\varepsilon^{-1} \log(1/\varepsilon))}$ different
signatures. Intuitively, an entry $\Procedure_i$ in $g$ tells us roughly how many
components of size $(1+\varepsilon)^i\cdot\varepsilon\lceil \abs{V}/k\rceil$
there are in the DP solutions that we consider. Due to space constraints, we
refer to the appendix for the formal definition.

For $x\in\mathbb{N}$, we let $e(x) \in [M-1]^t$ denote the signature of a single
component with $x$ vertices. More precisely, we set $e(x)$ to the vector that
has $e(x)_{j}=1$ for
$j=\arg\min\{j\in\mathbb{N} \colon x\le (1+\epsilon)^j\cdot\epsilon\lceil \abs{V}/k\rceil\}$ and
$e(x)_j=0$, otherwise. If $x<\epsilon\lceil \abs{V}/k\rceil$, we define
$e(x)=\smash{\vec{0}}$.

\textbf{Formal DP Definition.}
Now we describe the DP formally.  An entry $\DP(v,g,\cut,x) \in \Winfty$ in the
DP table for a vertex $v$ is indexed by a signature $g$, a Boolean value $\cut$
and $x\in[n]$. We will consider the tuples $(v,g,\cut)$ as the rows~$\I$
of the DP table and $x$ as the columns; we associate each such row with a
function $\DP(v,g,\cut,\cdot) \colon [n] \to \Winfty$. Note that our DP
has $\abs{V} \cdot M^t \cdot 2
	= (k/\varepsilon)^{O(\varepsilon^{-1} \log(1/\varepsilon))} \cdot n$
rows. Also, note that it has columns $n$; later, even though $x$ 
only takes discrete values, we will allow $x$ to take values in~$[0,\infty)$.

An entry $\DP(v,g,\cut,x)$ describes the optimum cost of cutting edges incident on
the subtree $T_v$ (including the cost of maybe cutting the edge to the parent
of~$v$). We will refer to the set of vertices in $T_v$ that are still connected
to the parent of $v$ after the cut edges are removed as the \emph{root component}.
We impose the following conditions on $\DP(v,g,\cut,x)$:
\begin{itemize}
	\item Once the cut edges are removed, the root component $U\subseteq T_v$
	has \emph{at most} $x$ vertices, i.e., $\abs{U}\leq x$.
	\item If $\cut$ is set to true then the edge between $v$ and its parent
	is cut, otherwise it is kept.
	\item The vertices inside $T_v$ that (once the cut edges are removed) are
	\emph{not} connected to the parent of~$v$ form connected components that are
	consistent with the signature $g$.
\end{itemize}

Next, we observe that if we fix a vertex~$v$, a signature~$g$ and a value for
$\cut$, then the resulting function $\DP(v,g,\cut,\cdot)$ is monotonically
decreasing in $x$.
\begin{observation}
\label{obs:monotone-intro}
	Let $v\in V$, $g\in[M-1]^t$ be a signature and $\cut\in\{\iscut,\notcut\}$.
	Then the function $\DP(v,g,\cut,\cdot) : [0,\infty) \to \mathbb{R}_+$ is
	monotonically decreasing.
\end{observation}
\begin{proof}
	By definition, $\DP(v,g,\cut,x)$ stores the cost of the optimum solution in
	which there are \emph{at most} $x$ vertices in the root component. Since
	$x\leq x'$, the solution $\DP(v,g,\cut,x)$ is also a feasible solution for
	$\DP(v,g,\cut,x')$.  Hence, $\DP(v,g,\cut,\cdot)$ is monotonically
	decreasing.
\end{proof}

\textbf{Comparison With the DP by Feldmann and Foschini.}
When comparing our DP with the one by Feldmann and
Foschini~\cite{feldmann15balanced} then one of the crucial changes is that in
our DP, $x$ encodes an \emph{upper bound} on the number of vertices in the root
component. Previously, Feldmann of Foschini considered root components
with \emph{exactly} $x$ vertices.  This is why their DP was
\emph{non-monotone} and why one can view our DP as the \emph{monotonization} of
the DP in~\cite{feldmann15balanced}. However, we also generalize the DP to the
setting with vertex weights and, as we will see below, parts of our algorithm
for computing the DP approximately are rather involved.

\subsubsection{Computing the DP}
We now give a flavor of what our algorithms for computing the DP look like.  We
start by showing how to compute the exact DP solution $\DP(v,\cdot,\cdot,\cdot)$
for a vertex~$v$ of the tree, where $v$ has parent $p$ and children $v_l$, $v_r$
and it is connected to them via edges $e_p$, $e_l$ and $e_r$, respectively.

Computing the DP is based on several case distinctions; here, we only consider
the case \emph{in which $v$ is an internal vertex of the tree we do not cut the
edges $e_l$ and $e_r$}.  All other cases are presented in the appendix.

When computing a DP row given by $\DP(v,\cdot,\cdot,\cdot)$, we will only
require access to the DP rows $\DP(v_l,\cdot,\cdot,\cdot)$ and
$\DP(v_r,\cdot,\cdot,\cdot)$. This implies that the dependency tree of the DP
is a tree and has the same height as our input graph~$G$ (recall that here we
assume that $G$ is a binary tree). Note that the height of the tree also implies
an upper bound on the number of reachable nodes.

\bold{Exact Computation.} 
We start with the exact computation. %
Here, we can afford to iterate over all
values $x\in[n]$ and $g\in[M-1]^t$ to compute $\DP(v,\cdot,\cdot,\cdot)$.
Therefore, we consider the values for $x$ and $g$ as input to our algorithm.

Since we assume that we do not cut the edges $e_l$ and $e_r$, we have to select
subsolutions for $T_{v_l}$ and $T_{v_r}$, where each subsolution is
characterized by the upper bound $x_l$ (resp.\ $x_r$) and its signature $g_l$
(resp.\ $g_r$).

First, suppose that we cut the edge $e_p$. If we let $x_l$ and $x_r$ denote the
number of vertices of the root components for the subsolutions, then the vertex $v$
will be included in a component of size $x_l+x_r+1$ afterwards. Hence, we can
combine the subsolutions to a solution for signature $g$ as long as
$g_l+g_r+e(x_l+x_r+1)=g$.  Consequently we set for every $x\in[0,\infty)$,
\begin{equation*}
	\DP_B(v,g,\iscut,x) =
		\capac(v,p)
		+ \min_{x_l,x_r,g_l+g_r=g-e(x_l+x_r+1)}
				\DP(v_l,g_l,\notcut,x_l)+\DP(v_r,g_r,\notcut,x_r).
\end{equation*}

Second, suppose that we do not cut $e_p$. Again we have to set
$\DP_B(v,g,\notcut,x) = \infty$ for all signatures $g$ and all $x\in[0,1)$,
because $v$ can reach $p$.  For $x\ge 1$ we have
to select $x_l$ and $x_r$  such that they sum to $x-1$ as this guarantees
that at most $x$ vertices can reach the parent $p$. Consequently, we
set for all $x\in[1,\infty)$
\begin{equation*}
	\DP_B(v,g,\notcut,x) =
	\min_{g_l+g_r=g,x_l+x_r=x-1}
		\DP(v_l,g_l,\notcut,x_l)+\DP(v_r,g_r,\notcut,x_r).
\end{equation*}

Here, we can afford to exhaustively enumerate all $O(M^{t} n^2)$ possibilities
in the $\min$-operations above.

\bold{Approximate Computation.} Now let us consider the approximate computation.
We denote the approximate DP solution by $\ADP$. We assume that we have
already computed the children solutions $\ADP(v_l,g,\cut,\cdot)$ and
$\ADP(v_r,g,\cut,\cdot)$ and that they are stored using our data structure from
Section~\ref{sec:approx-conv}. We will
maintain as an invariant that each of these functions has at most
$p=O(\lg_{1+\delta}(W))$ pieces and we will ensure this \emph{by rounding our
solution} at the end of every step, i.e., by setting $\ADP(v,g,\cut,\cdot) =
\lceil \ADP(v,g,\cut,\cdot)\rceil_{1+\delta}$ using the rounding procedure from
Lemma~\ref{lem:operations}. This will ensure the following two properties:
(1)~The functions $\ADP(v,g,\cut,\cdot)$ never have more than
$O(\lg_{1+\delta}(W))$ pieces by Lemma~\ref{lem:operations}. Thus, we can
perform all of our operations very efficiently. (2)~For the function at the root
of the tree, the approximation error is at most $(1+\delta)^h$, where $h$ is the
height of the tree. By picking $\delta = O(\varepsilon/h)$, we will achieve that
we obtain a $(1+\varepsilon)$-approximate solution at the root.  Now we proceed
to the explanation of our computation.

\emph{If we do not cut the edge to the parent of $v$}, we proceed similar to the exact
DP above. We start by setting $\tDP_B(v,g,\notcut,x)=\infty$ for all $x\in[0,1)$.
Next, for $x\in[1,\infty)$ we wish to set
\begin{align}
	\tDP_B(v,g,\notcut,x) &=
	\min_{g_l+g_r=g,x_l+x_r=x-1}
		\tDP(v_l,g_l,\notcut,x_l)+\tDP(v_r,g_r,\notcut,x_r)\\
	&= 
	\min_{g_l+g_r=g}
	\min_{x_l+x_r=x-1}
		\tDP(v_l,g_l,\notcut,x_l)+\tDP(v_r,g_r,\notcut,x_r).
\label{eq:b-line-intro}
\end{align}
Note that for fixed $g_l$ and $g_r$, the inner $\min$-operation in the second
line describes a $(\min,+)$-convolution due to the constraint $x_l+x_r=x-1$.
Therefore, in the inner $\min$-operation we compute a convolution
$\tDP(v_l,g_l,\notcut,\cdot)\minconv \tDP(v_r,g_r,\notcut,\cdot)$ and shift the
result by $1$ via the shift operation from Lemma~\ref{lem:operations} (where
for $x\in[0,1)$ we set $\tDP_B(v,g,\notcut,x) =\infty$).
We need time
$O(p^2\log p)$ for computing the convolution according to
Lemma~\ref{lem:convolution}.
To compute the outer minimum in Equation~\eqref{eq:b-line-intro},
we iterate over all $g_l\in[M-1]^t$ using Lemma~\ref{lem:operations}
and thus perform $O(M^{t})$ minimum
computations over piecewise constant functions with at most $p^2$~pieces. Hence,
we need time $O(M^{t} p^2 \lg(M^t p^2))$ according to Lemma~\ref{lem:multimin}.
By Lemma~\ref{lem:monotone-convolution-is-monotone}, $\tDP_B(v,g,\notcut,\cdot)$
is monotonically decreasing since it is the minimum over convolutions of two
monotonically decreasing functions.

\emph{If we cut the edge to the parent of $v$}, then for all $x\in[0,\infty)$ we would
like to set
\begin{align*}
	\tDP_B(v,g,\iscut,x) =
	\capac(v,p)
	+\min_{x_l,x_r,g_l+g_r=g-e(x_l+x_r+1)}
		\tDP(v_l,g_l,\notcut,x_l)+\tDP(v_r,g_r,\notcut,x_r).
\end{align*}
Note that here we need to be careful as the range of $g_l$ and $g_r$ depends on
the choice of $x_l+x_r$. Since there are $\Omega(n)$ possible values for
$x_l+x_r$, we cannot afford to iterate over all values that $x_l+x_r$ can take.
Instead, we will show that we only need to consider
$O(\lg(k/\varepsilon)/\varepsilon)$ different pairs $(x_l,x_r)$ by exploiting
the monotonicity of $\ADP(v_l,g_l,\notcut,\cdot)$ and
$\ADP(v_r,g_r,\notcut,\cdot)$.

First, observe that we can assume $x_l\le \abs{T_{v_l}}$ and $x_r\le \abs{T_{v_r}}$:
increasing the upper bounds on the number of vertices of the root component further would
mean that the root component contains than \emph{all} vertices
inside the sub-tree, which is impossible. Thus, $x_l+x_r+1\in[1,n]$.

Second, we partition the interval $[1,n]$ into $O(\log(k/\epsilon)/\epsilon)$
intervals. We have intervals $I_j=(\xi_{j-1},\xi_{j}]$ with
$\xi_j = (1+\varepsilon)^j\varepsilon\lceil n/k\rceil$ for all
$j=1,\dots,\log_{1+\epsilon}(k/\epsilon)$. In addition, we add an
\enquote{interval} $I_0:=[\epsilon\lceil n/k\rceil,\epsilon\lceil
n/k\rceil]$ and the interval $I_{-1}:=[1,\epsilon\lceil n/k\rceil)$.
We set $\xi_{0}=\epsilon\lceil n/k\rceil$ and we set $\xi_{-1}$ to the
largest integer that is less than $\epsilon\lceil n/k\rceil$.  Observe that
for all $j\ge -1$ and $x\in I_j$, we have $e(x)=e(\xi_j)$, i.e., the value of
$e(x)$ does not change inside the interval $I_j$.  Below, this property will
allow us to separate the conditions on $x_l+x_r$ and on $g_l+g_r$.

Now we can rewrite the above expression as
\begin{align*}
&\tDP_B(v,g,\iscut,x) = \\
&\capac(v,p)+
\min_j
\min_{x_l+x_r+1\in I_j}
\min_{g_l+g_r=g-e(\xi_j)}\tDP(v_l,g_l,\notcut,x_l)+\tDP(v_r,g_r,\notcut,x_r).
\end{align*}

Third, note that now the two $\min$-operations only depend on the choice of $j$
and, importantly, the minimum over $g_l$ and $g_r$ does not depend on the choice
of $x_l+x_r$ anymore. Therefore, we can swap the order of the two
$\min$-operations.
Furthermore, since $\tDP_B(v,g,\notcut,x)$ is monotonically decreasing with $x$, we can
restrict the choice of $x_l$ and $x_r$ such that $x_l+x_r+1$ is the largest
number in the corresponding interval $I_j$, i.e., $x_l+x_r+1=\xi_j$. Thus,
\begin{equation*}
\begin{split}
&\tDP_B(v,g,\iscut,x) = \\
&\capac(v,p)+
\min_j
\min_{g_l+g_r=g-e(\xi_j)}\min_{x_l+x_r+1=\xi_j}
\tDP(v_l,g_l,\notcut,x_l)+\tDP(v_r,g_r,\notcut,\xi_j-x_l-1).
\end{split}
\end{equation*}

Next, we explain how the above expression can be computed efficiently.
Let us first argue how we can efficiently compute the inner $\min$-operation of
the above expression. We start by observing that this $\min$-operation is
\emph{not} a convolution since in the constraint we sum up to $\xi_i$ which is a
constant (rather than to the variable $x$). Now recall that
$\tDP(v_l,g_l,\notcut,\cdot)$ and $\tDP(v_r,g_r,\notcut,\cdot)$ are piecewise
constant functions with $O(p)$ pieces by our invariants. Since $x_l,x_r\geq 0$
this implies that there are only $O(p^2)$ choices for $x_l$ and $x_r$ such that
$x_l,x_r \in I_j$ and either a new piece starts in $\tDP(v_l,g_l,\notcut,x_l)$
or in $\tDP(v_r,g_r,\notcut,x_r)$. Thus, we can iterate over all these pairs
$(x_l,x_r)$ and evaluate $\tDP(v_l,g_l,\notcut,x_l)+\tDP(v_r,g_r,\notcut,x_r)$,
where $x_r=\xi_j-x_l-1$.
Thus, we can compute the inner $\min$-operation in time $O(p^2 \lg p)$.
We note that since this $\min$-operation is considering a super-constant number
of terms, this DP is not well-behaved (it violates Property~(4b) of
Definition~\ref{def:well-behaved}). This is why in our analysis we will use the
more general notion from Section~\ref{sec:dps-on-trees}.

Next, we can compute the outer two $\min$-operations by simply iterating over
$j$ and all choices for $g_l$ and setting $g_r=g-e(\xi_j)-g_l$ as above in
$O(M^{t}\cdot\log(k/\epsilon)/\epsilon)$ iterations.  Hence, we obtain 
a running time of $O(M^{t}p^2\log p\cdot\log(k/\epsilon)/\epsilon)$.

Finally, we note that as $\tDP_B(v,g,\iscut,x)$ is independent of $x$, it is a
constant. Thus, $\tDP_B(v,g,\iscut,x)$ is a piecewise constant function with a
single piece and it is monotonically decreasing. 

\emph{Rounding Step.}
As noted earlier, after computing the solutions $\ADP_B(v,g,\notcut,\cdot)$ and
$\ADP_B(v,g,\iscut,\cdot)$, we also round the solution by setting
$\ADP_B(v,g,\cut,\cdot) = \lceil \ADP_B(v,g,\cut,\cdot)\rceil_{1+\delta}$ for
$\cut\in\{\iscut,\notcut\}$ to ensure that we only have $p=O(\lg_{1+\delta}(W))$
pieces in the resulting function.
Note that this is the only approximate operation we perform and all other
operations above have been exact.

\subsection{Inverting the DP of Andreev et al.}
\label{sec:inverting}
Now we briefly describe our DP for simultaneous source location. Recall that
in this problem, the input consists of an undirected graph $G=(V,E,\capac,d)$
with a capacity function $\capac \colon E \to \Winfty$ and a demand
function $d \colon V \to \Winfty$. The goal is to select a
minimum set $S\subseteq V$ of \emph{sources} that can simultaneously supply all
vertex demands. More concretely, a set of sources $S$ is \emph{feasible} if
there exists a flow from the vertices in $S$ that supplies demand $d(v)$ to all
vertices $v\in V$ and that does not violate the capacity constraints on the
edges. The objective is to find a feasible set of sources of minimum size.

Here, we will again assume the special case in which $G$ \emph{is a binary
tree}; we show in Appendix~\ref{sec:ssl-general} how to drop this assumption.

\textbf{DP Definition.} %
Given a vertex $v$ and a
value $x\in\mathbb{R}$, we let $\DP(v,x)$ denote the minimum number of sources
that we need to place in the subtree $T_v$ such that when $v$ receives flow at
most $x$ from its parent then all demands in $T_v$ can be satisfied.  We note
that $x$ can take positive and negative values: for $x\geq0$ this corresponds to
the setting in which flow is sent from the parent of $v$ into $T_v$ and for
$x<0$ this corresponds to the setting in which flow is sent from $T_v$ towards
the parent of $v$. We further follow the convention that when the demands in
$T_v$ cannot be satisfied when $v$ receives flow $x$ from its parent, then we
set $\DP(v,x)=\infty$.

Observe that this DP has rows $\I=V$ and columns $\J=\mathbb{R}$. Furthermore,
$\DP(v,\cdot)$ is monotonically decreasing since for $x < x'$, any
solution in which $T_v$ receives flow at most $x$ from the parent of $v$ is also
feasible when $T_v$ receives flow at most $x'$ from the parent of~$v$.
This satisfies Property~(1) of Definition~\ref{def:well-behaved}.

\textbf{The Inverse DP.}
Interestingly, our DP is very related to the one by Andreev et
al.~\cite{andreev09simultaneous}. They defined a function $f(v,i)$ which, given
a vertex $v$ and an integer $i\in\mathbb{N}$, denotes the \emph{minimum amount
of flow} that $v$ needs to receive from its parent if all demands in $T_v$ need
to be satisfied and if we can place $i$~sources in the subtree $T_v$.  Similar
to above, $f(v,i)$ takes positive values if the demand in $T_v$ can only be
satisified by receiving flow from the parent of $v$ and it takes negative values
if the demand in $T_v$ is already satisfied by the sources in the subtree~$T_v$
and $v$ can send flow to its parent.

Now observe that our DP can essentially be viewed as the ``inverse'' of
$f(v,i)$. More formally, observe that
$\DP(v,x)=f^{-1}(v,x):=\min\{ i \colon f(v,i)\leq x \}$.

The reason why we chose the inverse formulation for our DP is as follows. To
ensure that our algorithms are efficient, we have to make sure that our monotone
piecewise constant functions have only few pieces. One natural way to do is
using rounding. However, since the function values of $f$ are positive and
negative, it is not clear how we should perform the rounding. For example, to
only use a small number of pieces for representing~$f$, we would have to use
different rounding mechanisms for those function values in $[-1,1]$ and those in
$[-W,W]\setminus[-1,1]$, where $W$ is the largest edge capacity:  Indeed, if we
rounded the values of~$f$ to powers of $(1+\delta)^j$ then there are only
$O(\lg_{1+\delta}(W))$ function values in $[-W,W]\setminus[-1,1]$ but there are
infinitely many function values in $[-1,1]$.  Similarly, if we rounded to
multiples of $\delta$ then there are only $O(1/\delta)$ function values in
$[-1,1]$ but this would lead to $O(W/\delta)$ function values in
$[-W,W]\setminus[-1,1]$. In both cases, our functions would have too many pieces
and we would have to pick a rounding function which provides a tradeoff
between these two cases. Furthermore, we would have to find an analysis that
shows that this ``more involved'' rounding function does not introduce much too
error.

In our DP we bypass these issues because we move the negative numbers into the
domain of the function $\DP(v,\cdot) \colon \mathbb{R} \to [n+1]$. Then in the
codomain we only have non-negative numbers to which we can apply the standard
rounding function $\lceil \cdot \rceil_{1+\delta}$ in a straightforward way.
This also has the positive side effects that instead of getting factors of
$\polylog(W)$ in our running times, we only get factors of $\polylog(n)$ because
our codomain became $[n+1]$ rather than some potentially large interval
$[-W,W]$.  We believe this technique of considering inverse DPs will be useful
in the future to compute approximate solutions for DPs that can take positive
\emph{and negative} values.

Due to lack of space, we present the details for computing the DP in
Appendix~\ref{sec:simultaneous-source-location}.

{
\bibliographystyle{plainurl}
\bibliography{main}
}

\newpage
\appendix

\section{Organization of the Appendix}
\label{sec:organization-appendix}
Our appendix is organized as follows:
\begin{itemize}
\item In Appendix~\ref{sec:related} we discuss more related work.
\item Appendix~\ref{sec:preliminaries} introduces preliminaries.
\item Appendix~\ref{sec:balanced-partitioning} presents our results for 
$k$-balanced partitioning.
\item Appendix~\ref{sec:simultaneous-source-location} presents our results for
simultaneous source location.
\item Appendix~\ref{sec:recourse} presents our recourse lower bounds for
algorithms which only maintain few solutions.
\item Appendix~\ref{sec:generalization} presents our
generalization to functions with non-monotonicities and our results for
$\ell_\infty$-necklace.
\item Appendix~\ref{sec:omitted-proofs} presents missing proofs.
\end{itemize}

\section{Further Related Work}
\label{sec:related}
\label{sec:comparison}

Speeding up DP algorithms is a well-studied topic, which has received attention
for several
decades~\cite{bein09knuth,aggarwal87geometric,eppstein88speeding,galil92dynamic,chan21near,knuth71optimum,yao80efficient,monge1781memoire,burkard96perspectives,miller11approximate}.
This line of work has led to several conditions, which, if satisfied, imply that
the underlying DP can be solved more efficiently. These conditions include, for
example, the Monge property, total monotonicity, certain convexity and concavity
properties, or the Knuth--Yao quadrangle-inequality, which are often related to
each other. For example, it is known that DP tables which satisfy the Monge
property are also totally monotone. One of the most popular methods in this area
is the SMAWK algorithm~\cite{aggarwal87geometric} which runs in near-linear time
in the number of columns of the DP table if the DP table is \emph{totally monotone}.
More concretely, a DP table is totally monotone if for each submatrix $A$ of the
DP table and for every pair of consecutive rows~$i$ and $i+1$ in $A$, the
minimum entry for row~$i+1$ appears in a column that is equal to or greater
than the minimum entry for row~$i$.

However, these conditions are quite different from our conditions in
Definition~\ref{def:well-behaved} and they are essentially incomparable.  For
the purpose of illustration, we will briefly argue this for total monotonicity
and Definition~\ref{def:well-behaved}; similar arguments can also be made for
the Monge property and other criteria.  On one hand, the totally monotone
matrices do not imply that the rows of the DP table are monotone. Indeed, when
the rows are monotone then finding the columns with the minimum entries is
trivial (they are always in the first or last column, depending on whether we
consider monotonically increasing or decreasing rows, respectively). Hence,
total monotonicity does not imply our condition from
Definition~\ref{def:well-behaved}. On the other hand, the ordering of the rows
is highly important for the conditions above: just swapping two rows of a
totally monotone DP table can break total monotonicity.  In our case, the
rows can be ordered arbitrarily in the DP table, as long as their dependency
graph has good properties. Hence, our property does not imply total
monotonicity. This shows that these definitions are incomparable.

Recently, Varma and Yoshida~\cite{varma21average} and Kumabe and
Yoshida~\cite{kumabe22average} studied the sensitivity of graph algorithms and
of DP algorithms. They studied how much the solutions of such algorithms change
when a random element from the input is deleted. For several problems including
knapsack they showed that these algorithms have small sensitivity. However, we
show in Section~\ref{sec:recourse} that when insertions are allowed, dynamic
algorithms must have high recourse or they have to maintain many different
solutions.

The $k$-balanced graph partitioning problem has received a lot of attention in
the theory
community~\cite{even99fast,andreev2006balanced,feldmann15balanced,feige02polylogarithmic}.
The problem is also highly relevant in
practice~\cite{karypis1997metis,sanders13experimental,buluc16recent,dong20learning},
where algorithms for balanced graph partitioning are often used as a
preprocessing step for large scale data analytics.
For the special case of $k=2$, this corresponds to
the minimum bisection problem and Feige and
Krauthgamer~\cite{feige02polylogarithmic} presented polynomial-time algorithms
with polylogarithmic approximation ratios. For $k\geq3$, Andreev and
Räcke~\cite{andreev2006balanced} showed that no polynomial-time algorithm can
achieve a finite approximation ratio unless $\P=\NP$. They also showed how to
compute a bicriteria $(O(\lg^{1.5}(n)/\varepsilon^2),1+\varepsilon)$-approximate
solution in polynomial time. Feldmann and Foschini~\cite{feldmann15balanced}
obtained a polynomial-time bicriteria
$(O(\lg^{1.5}(n) \lg \lg n),1+\varepsilon)$-approximation algorithm which has
the advantage that the approximation ratio does not depend on the
parameter~$\varepsilon$ of the partition sizes.  Even et al.~\cite{even99fast}
showed that one can compute a bicriteria $(O(\lg n),2)$-approximation in
polynomial time.

The simultaneous source location problem that we study is closely related to the
source location problem introduced by Tamura et
al.~\cite{tamura1990location,tamura1992some}, in which a minimum number of
sources must be selected to be able to satisfy any single demand in an
undirected edge-capacitated graph. Arata et al. \cite{arata2002locating} showed
that the problem is NP-hard and presented an exact algorithm for the variant
with uniform vertex costs.  In the simultaneous source location problem that was
introduced by Andreev et al. \cite{andreev2006balanced} and that we study in
this paper, \emph{all} demands must be satisfied simultaneously. Andreev et al.\
provide an $O(\log D)$-approximation algorithm, where $D$ is the sum of demands,
and a matching hardness result for this problem in general graphs.  They also
present an exact polynomial-time algorithm when the input graph is a tree and
show that this result can be extended to general graphs when the edge capacities
can be violated by a $O(\lg^2 n \lg \lg n)$-factor, where $n$ is the number of
vertices in the graph.

Chan~\cite{chan18approximation} showed that one can consider
the solutions for the 0-1 knapsack as monotone piecewise constant functions and
used this insight to obtain faster algorithms.  Recently, these results were
improved by Jin~\cite{jin19improved} who showed how to compute a
$(1+\varepsilon)$-approximation for 0-1 knapsack with $n$~items in time
$\tO(n + \varepsilon^{-9/4})$. Bringmann and Cassis~\cite{bringmann22faster}
derived faster exact algorithms for 0-1 knapsack using bounded monotone
min-plus-convolution. Aouad and Segev~\cite{aouad22approximate} study the
incremental knapsack problem, where the capacity constraint is increased over
time and the goal is to find nested subsets of items which maximize the average
profit; we note that this is different from our setting, where the goal is to
obtain efficient update times, while the solutions may change arbitrarily over
time.

An $\ell_1$-necklace alignment problem was first considered by Toussaint
\cite{toussaint2004geometry},  motivated by computational music theory and
rhythmic similarity \cite{toussaint2004comparison}. Toussaint focused on a
scenario where the beads lie at integer coordinates.  Ardila et
al.~\cite{ardila2005necklace} studied the problem for binary strings.  There
also exist results for different distance measures between two sets of points on
the real line in which not every points needs to be
matched~\cite{colannino2006n}, as well as for computing the similarity of two
melodies when they are represented as closed orthogonal chains on a
cylinder~\cite{aloupis2006algorithms}. Bremner et al.~\cite{bremner14necklaces}
showed that $\ell_2$-necklace alignment can be solved in time $O(n\lg n)$, where
$n$ is the number of beads, using FFT. They also showed that
$\ell_\infty$-necklace alignment can be solved using a constant number of
$(\min,+)$-operations and obtained subquadratic-time algorithms for $\ell_1$-
and $\ell_\infty$-necklace alignment.

A common subroutine that is employed when solving DPs is $(\min,+)$-convolution;
note that this subroutine is also of high importance in all of our algorithms.
The complexity of $(\min,+)$ convolution has received significant
attention in the literature
\cite{axiotis2018capacitated,backurs17better,bateni2018fast,bremner14necklaces,bringmann19approximating,bussieck1994fast,chan15clustered,cygan19problems,jansen2018integer,indyk2017fine,laber2014lower,mucha2019subquadratic,chi22faster}.
It was shown that naive algorithm with running time $O(n^{2})$ can be improved to
time $n^2/2^{\Omega(\sqrt{\lg n})}$~\cite{bremner14necklaces,williams18faster}
by a reduction to All Pairs Shortest Path~\cite{bremner14necklaces}
using Williams’ algorithm for the latter~\cite{bremner14necklaces}.
However, so far, no $O(n^{2-\varepsilon})$-time algorithm was found,
which led to the \emph{MinConv} hardness conjecture in fine-grained complexity theory \cite{cygan19problems,indyk2017fine}.
The conjecture is particularly appealing because it implies other conjectures
such as the 3-SUM and the All-Pairs Shortest Paths conjectures, and
dozens of lower bounds that follow from them
(see~\cite{williams2018some,cygan19problems}).
There further exist many conditional lower bounds from the MinConv conjecture and several MinConv-equivalent problems are known, e.g., related to the knapsack problem or to subadditive sequences \cite{cygan19problems,indyk2017fine}, 
among others
\cite{backurs17better,cygan19problems,jansen2018integer,indyk2017fine,laber2014lower,mucha2019subquadratic,abboud2020new,eiben2019bisection}.
There have also been improvements for efficiently approximating the
$(\min,+)$-convolution in the case of large
weights~\cite{bringmann19approximating} for the exact $(\min,+)$-matrix product
with bounded differences~\cite{bringmann19truly}.

\section{Preliminaries}
\label{sec:preliminaries}

We introduce some preliminaries that we will use in the rest of the paper.  For
the sake of better readability, we present some of the proofs in
Appendix~\ref{sec:omitted-proofs}.
We write $[m]$ to denote the set $\{0,1,\dots,m\}$.

Throughout the paper, we will consider input graphs $G=(V_G,E_G,\capac_G)$ with
$n$ vertices and $m$ edges, where $\capac_G:E_G\rightarrow\Winfty\cup\{\infty\}$ is
a weight function that for an edge $e\in E_G$ describes the capacity of the
edge. To simplify notation we extend $\capac_G$ to all vertex pairs and define
\begin{equation*}
\capac_G(x,y)=\left\{
\begin{array}{ll}
\capac_G(\{x,y\}) & \{x,y\}\in E_G\\
0 & \text{otherwise.}
\end{array}
\right.
\enspace.
\end{equation*}
Additionally, for disjoint sets $A,B\subseteq V_G$, we set
$\capac_G(A,B):=\sum_{(a,b)\in A\times B}\capac_G(a,b)$ and
$\capac_G(A):=\capac_G(A,V\setminus A)$. We drop the subscript $G$ of the
capacity function $\capac$ whenever the graph is clear from the context.

Let $(V_T,E_T,r)$ be a rooted tree. For a vertex $v\in V_T$ we use $T_v$ to denote
the subtree rooted at $v$ and we say that the \emph{degree} of $v$ is its number
of children. The \emph{height $h$ of $T$} is the length of the longest path from
the root to a leaf.

\subsection{Räcke Tree}
\label{sec:racke-tree}
A \emph{Räcke tree}~\cite{racke02minimizing} (or \emph{tree cut sparsifier})
$T=(V_T,E_T)$ for an undirected graph $G=(V_G,E_G)$ is a weighted, rooted tree
in which the leaf nodes correspond to vertices of $G$. 
For a vertex $v\in V_T$, we write $V_v\subseteq V_G$ to denote
the set of leaf vertices in $T_v$.
Naturally, an edge $e=(u,v)$ of $T$
corresponds to a cut in $G$, namely to the cut formed by the set $V_u\cap V_v$
in $G$. The capacity $\capac_T$ of the tree edge $(u,v)$ is set to the capacity
of this cut, i.e., to $\capac_G(V_u\cap V_v)$.

For a graph $H=(V_H,E_H)$ and two disjoint subsets $A,B\subseteq V_H$, we write 
\begin{equation*}
\mincut_H(A,B):=
  \min_{S\subseteq V_H: A\subseteq S, B\subseteq \bar{S}}\capac_H(S)
\end{equation*}
to denote the minimum capacity of a cut that separates $A$ and $B$.
By definition of the edge capacities in $T$ we have
$\mincut_T(A,B)\ge \mincut_G(A,B)$ for any two disjoint subsets $A,B\in V_G$.
For the sake of completeness, we prove
this property in Appendix~\ref{sec:property-racke-tree}.

The goal of a Räcke tree $T$ is to approximate the cut-structure of $G$, i.e.,
to guarantee that for all disjoint sets of vertices $A,B\subseteq V_G$, 
\begin{equation*}
\mincut_G(A,B)\le\mincut_T(A,B)\le q\cdot\mincut_G(A,B)\enspace,
\end{equation*}
for a small value $q\ge 1$. The parameter $q$ is called the \emph{quality} of
the Räcke tree.

In the static setting, Räcke trees with polylogarithmic quality guarantees can
be computed in nearly linear time~\cite{racke14computing,peng16approximate}.
When larger running times are allowed, better qualities can be
achieved~\cite{bienkowski03practical,harrelson03polynomial,racke14improved}.
\begin{theorem}[Peng~\cite{peng16approximate}]
\label{thm:racke-tree}
	Let $G$ be a connected undirected graph with $n$~vertices and $m$~edges.
	Then there exist an algorithm that computes a Räcke tree of height $O(\lg n)$ for
	$G$ with quality $O(\log^4n)$ in time $\tilde{O}(m)$.
\end{theorem}

Furthermore, there has recently been interest in maintaining Räcke trees
dynamically~\cite{goranci21expander,hua22maintaining}. Here, we will use a
result by Goranci, Räcke, Saranurak and Tan who showed that one can maintain
Räcke trees for unweighted graphs dynamically with subpolynomial update time.
\begin{theorem}[Goranci, Räcke, Saranurak and Tan~\cite{goranci21expander}]
\label{thm:dynamic-racke-tree}
	Let $G$ be an undirected, unweighted graph with $n$ vertices that is
	undergoing edge insertions and deletions.  There exists a deterministic
	algorithm with amortized update time $n^{o(1)}$ that maintains a Räcke tree
	for $G$ with quality $n^{o(1)}$ and height $O(\lg^{1/6}n)$.
\end{theorem}

\subsection{Okay-Behaved DPs}
\label{sec:dps-on-trees}
We introduce a more general DP condition compared to the one in
Definition~\ref{def:well-behaved} which, however, will not allow us to obtain
results like Theorems~\ref{thm:well-behaved-static}
or~\ref{thm:well-behaved-dynamic}. We will consider the same type of DP tables
as in Section~\ref{sec:approx-conv}.

\begin{definition}
\label{def:okay-behaved}
A DP is \emph{okay-behaved} if it fulfills the sensitivity condition
of well-behaved DPs:
	Suppose $\beta>1$ and for all $i'\in\In(i)$, we obtain a
	$\beta$-approximation $\ADP(i',\cdot)$ of $\DP(i',\cdot)$ (as
	per Equation~\eqref{eq:approximation}).
	Then applying $\Procedure_i$ on the $\ADP(i',\cdot)$ yields
	a $\beta$-approximation of $\DP(i,\cdot)$, i.e.,
	\begin{align*}
		\DP(i,\cdot)
			\le \Procedure_i(\{ \ADP(i',\cdot) \colon i'\in\In(i) \} )
			\le \beta\cdot\DP(i,\cdot).
	\end{align*}
\end{definition}

We also use routines $\tilde{\Procedure}_i$ to compute the DP rows
$\ADP(i,\cdot)$. Again, if for all~$i$ it holds that 
$\tilde{\Procedure}_i(\{ \ADP(i',\cdot) \colon i'\in\In(i) \} )$ is an
$\alpha$-approximation of $\Procedure_i(\{ \ADP(i',\cdot) \colon i'\in\In(i) \})$,
we say that $\ADP(1,\cdot),\dots,\ADP(n,\cdot)$ is an
\emph{$\alpha$-approximate DP solution}.

In the dependency graph, we call a vertex without any incoming edges a
\emph{leaf}. The \emph{level} of a vertex~$u$ is the length of the longest path
from a leaf to~$u$.  Similar to the proof of
Theorem~\ref{thm:well-behaved-static} we can show the following approximation
guarantee for the approximate solutions $\ADP(i,\cdot)$ and the exact solutions
$\DP(i,\cdot)$.
\begin{lemma}
\label{lem:approx-dp}
	Let $i$ be a vertex of the dependency graph with level~$\ell$. Then the
	entry $\tDP(i,\cdot)$ in the $\alpha$-approximate ADP-solution
	for a okay-behaved DP problem fulfills
	\begin{equation*}
		\DP(i,\cdot)\le\tDP(i,\cdot)\le\alpha^{\ell+1}\cdot\DP(i,\cdot).
	\end{equation*}
\end{lemma}

Next, suppose the dependency graph of the DP
that we consider is derived from a tree as follows.
Let $T = (V_T,E_T,r)$ be a rooted tree with root $r$ and
height~$h$.  We assume that the children of a vertex are ordered from left to
right.  The dependency graph that we associate with $T$ is simply a directed
copy of $T$ in which we direct each edge towards the root. More precisely, the
dependency graph contains copies of all vertices in $V_T$ and for each vertex
$v$ (except for $r$) an edge to its parent $p$. Clearly, this set of edges
induces a DAG in which the longest path has at most $h$ edges.  The following
lemma summarizes the properties of approximate DP solutions when using this
approach.

\begin{lemma}
\label{lem:tree-dp}
	Consider a rooted tree $T=(V_T,E_T,r)$ with height~$h$.
	Consider an okay-behaved DP and the ADP-solution $\tDP(i,\cdot)$
	corresponding to the dependency graph described above. Assume that each
	$\tilde{\Procedure}_i$ is an $\alpha$-approximation of $\Procedure_i$ and can be computed in
	time at most $t$.  Then $\tDP(r,\cdot)$ is an $\alpha^{h+1}$-approximation of
	$\DP(r,\cdot)$ and can be computed in time $O(\abs{V_T} \cdot t)$.
\end{lemma}

The main difference of this lemma together with the definition of okay-behaved
DPs and Theorem~\ref{thm:well-behaved-static} with well-behaved DPs is as
follows.  When applying Theorem~\ref{thm:well-behaved-static}, we only have to
consider how many pieces our functions have and we do not have to bother about
deriving running times bound for computing the operations on our functions
(because the additional conditions from the well-behaved DPs imply good running
 time bounds).  Here, we have to check less conditions for okay-behaved DPs
(in particular, we do not have to bound the number of pieces or operations) but we
have to provide our own running time analysis.

Later, when we consider dynamic algorithms, we will have to consider the
scenario when the underlying tree $T$ changes due to edge insertions and
deletions (and therefore might become a forest). In that case, the dependency
graph and the DP solutions $\DP(i,\cdot)$
and $\ADP(i,\cdot)$ change over time as well. The following lemma asserts that when
a vertex $i$ is affected by an edge insertion or deletion,
we only have to recompute the solutions $\DP(j,\cdot)$
and $\ADP(j,\cdot)$ for vertices $j$ that are reachable from $i$ in the
dependency graph and that there are at most $h$ such vertices.
\begin{lemma}
\label{lem:tree-dp-dynamic}
	Consider a rooted tree $T=(V_T,E_T,r)$ with height~$h$
	that is undergoing edge insertions and deletions.  Then after
	each insertion or deletion, we can recompute an ADP-solution with the same
	guarantees as in Lemma~\ref{lem:tree-dp} in time $O(h \cdot t)$, where $t$
	is the time it takes to compute the functions $\tilde{\Procedure}_i$.
\end{lemma}

Lemma~\ref{lem:operations} already provided a way to compute the minimum of two
monotone piecewise constant functions. When more than two functions are involved
in the minimum computation, the following version gives improved guarantees.
\begin{lemma}
\label{lem:multimin}
	Let $f_i: [0,t]\to\Winfty$, $i\in\{1,\dots,k\}$ be piecewise constant functions
	that are either all monotonically increasing or all monotonically decreasing.
	Then $f_{\min}(x):=\min_i\{f_i(x)\}$ can be computed in time $O(\sum_i
	p_i\cdot\log(\sum_i p_i))$, where $p_i$ denotes the number of pieces of
	function $f_i$.
\end{lemma}

We also note the following well-known lemma for sake of completeness.
\begin{lemma}
\label{lem:monotone-convolution-is-monotone}
	Let $f_1,f_2 \colon [0,t]\to\Winfty$ and suppose that one of $f_1$
	and $f_2$ is monotonically decreasing. Then $f=f_1\minconv f_2$ is
	monotonically decreasing.
\end{lemma}

\section{Balanced Graph Partitioning}
\label{sec:balanced-partitioning}
In this section, we provide an algorithm for the \emph{$k$-balanced graph partitioning}
problem.  In this problem, the input consists of a graph $G=(V,E,\capac)$, where
$\capac: E \to \Winfty$ is a weight function on the edges, and an integer $k$.
The goal is to find a partition $V_1,\dots,V_k$ of the vertices such that $\abs{V_i}
\leq \lceil \abs{V}/k\rceil$ for all $i$ and the weight of the edges which are cut
by the partition is minimized. More formally, we want to minimize
$\cut(V_1,\dots,V_k) := \sum_i\capac(V_i)$, where
$\capac(V_i) = \sum_{\{u,v\} \in E \cap (V_i, V\setminus V_i)} \capac(u,v)$.

Since the above problem is \NP-hard to approximate within any factor
$n^{1-\varepsilon}$ for any $\varepsilon$ even on
trees~\cite{feldmann15balanced}, we consider bicriteria approximation
algorithms. Given a weighted graph $G=(V,E,\capac)$, we say that a partition
$V_1,\dots,V_k$ of $V$ is an \emph{$\emph(\alpha,\beta)$-approximate} solution
if $\abs{V_i}\leq\beta\lceil n/k\rceil$ for all $i$ and $\cut(V_1,\dots,V_k) \leq
\alpha\cdot \cut(\OPT)$, where $\OPT=(V_1^*,\dots,V_k^*)$ is the optimal
solution with $\abs{V_i^*}\leq\lceil n/k\rceil$ for all $i$.

Our first main result in this section is summarized in the following theorem.
We use the notation $O'(\cdot)$ to suppress
factors in $\poly(\lg n, k, \lg(1/\varepsilon), \lg\lg(W))$.
\partitioningstatic*

Furthermore, we can also extend our results to the dynamic setting in which the
graph $G$ is undergoing edge insertions and deletions.  Our second main result
in this section is summarized in the following theorem.
\partitioningdynamic*

Our DP approach is inspired by the DP of Feldmann and
Foschini~\cite{feldmann15balanced}.  However, the DP cells in the algorithm of
Feldmann and Foschini are not monotone and, therefore, their DP cannot directly
be sped up by the fast convolution of monotone functions approach.  Hence, we
first simplify and generalize their DP to make it monotone such that we can
apply the fast convolution of monotone functions approach.

We note that in our static and dynamic algorithms, we can output the
corresponding solutions similarly to what we descriped after
Proposition~\ref{prop:dynamic-knapsack-chan} for knapsack.

To obtain these results, we will first describe an exact DP in
Section~\ref{sec:balanced-partitioning-exact} for the special case of binary
trees. Then we will show how to compute
the DP more efficiently by introducing approximation in
Section~\ref{sec:balanced-partitioning-approximate}. In
Section~\ref{sec:computing-result} we show how to return a solution based on
our DP table. Sections~\ref{sec:extension-non-binary}
and~\ref{sec:extension-dynamic} provide extensions from binary trees to more
general graphs and to the dynamic setting, respectively.

\subsection{The Exact DP}
\label{sec:balanced-partitioning-exact}
When describing the DP, we will make two assumptions. First, we assume that the
input graph $T=(V,E)$ is a \emph{binary} tree (we show in
Section~\ref{sec:extension-non-binary} how to remove this assumption).
Second, we consider a slight generalization of the $k$-balanced partition
problem on trees; we note that we did not mention this generalization in
Section~\ref{sec:technical}. In this generalization, we suppose that each vertex
is assigned a weight by a weight function $w \colon V \to \{0,1\}$.\footnote{We
	note that our proofs and algorithms also work for
	more general weight functions $w \colon V \to \mathbb{R}_+$. However, in
	that case the functions $\DP(v,g,\cut,\cdot)$ that we will introduce later
	will become more complicated to compute and, therefore, we stick with the
	simpler case of vertex weights in $\{0,1\}$.}
For convenience we set $w(U)=\sum_{u\in U} w(u)$ for all
$U\subseteq V$ and refer to $w(U)$ as \emph{the weight of the vertices in $U$}.
Now our goal will be to find a partition $V_1,\dots,V_k$ of $V$ such that
$w(V_i) \leq (1+\varepsilon)\lceil w(V)/k \rceil$ for all $i$ and we will
compare against $\OPT=(V_1^*,\dots,V_k^*)$, where $\OPT$ is the optimal solution
with $w(V_i^*)\leq\lceil w(V)/k\rceil$ for all $i$. Note that by setting $w(v)=1$
for all $v\in V$, we obtain the standard $k$-balanced partition problem and,
therefore, our variant is a strict generalization.

The reason for considering the above generalization is that later we want to use
our algorithm to find a balanced partitioning of general graphs $G=(V',E')$
using a Räcke tree~$T=(V,E)$ (see Section~\ref{sec:racke-tree}). However, the vertices
$V'$ of $G$ are just a subset of the vertices $V$ of the Räcke tree~$T$
(since the vertices of $G$ correspond to leaves in $T$ and the internal nodes of
 $T$ do not correspond to any vertices in $G$).  Thus, if we assigned weight
$w(v)=1$ to all vertices in $T$ and computed a balanced partitioning of $T$,
this would not necessarily correspond to a balanced partitioning of $G$.
Instead, later we will consider the weight function which assigns weight~1 to
all leaves in $T$ (corresponding to the vertices in $G$) and weight~0 to all
internal nodes of $T$ (which can be ignored when deriving a partitioning of
$G$). Then each set $V_i$ in $T$ will correspond to a set $V_i'$ in $G$ with
$w(V_i) = \abs{V_i'}$.  In particular, if $w(V_i)\leq (1+\varepsilon)\lceil
w(V)/k\rceil$ then we will obtain that
$\abs{V_i'} \leq (1+\varepsilon)\lceil \abs{V'}/k\rceil$ and, therefore, the
sets $V_1,\dots,V_k$ imply a balanced partition $V_1',\dots,V_k'$ of $G$.

\textbf{High-Level Description of the DP.}
We start by giving a high-level description of the DP. The DP is computed
bottom-up starting at the leaves of the tree~$G$ and then moving up.  For
each vertex $v$, we will compute a DP solution of minimum cost that encodes
whether the edge to the parent~$p$ of $v$ is cut and which edges shall be cut inside
the subtree $T_v$ that is rooted at $v$. Note that the removal of the cut edges
in our solution will decompose the tree into disjoint connected components and
exactly one of them contains $v$'s parent~$p$.
Additionally, we store information about the weight of the vertices that
are still connected to the parent $p$ (and, therefore, to the outside of
$T_v$) after the cut edges are removed.  We will assume that when we compute the
DP cell for a vertex $v$, we have access to the solutions for both of its
children.

More concretely, when we have computed a solution for a subtree $T_v$, i.e., we
know which edges incident to nodes in this subtree we are going to remove (note
that the edge leading to the parent of $v$ is incident to $T_v$ and thus we
consider it as part of this solution), we store the following information in the DP
table.  First, we store its cost, i.e., the total capacity of all edges that are
incident to vertices in $T_v$ and that are cut.  As described above, we would
also like to store the weight of the vertices that are connected to the parent
of $v$ and the sizes of connected components inside $T_v$. However, there are
two difficulties: (1)~We cannot store the weight of the vertices that are
connected to the root exactly because this would result in a too large DP table.
Instead, we store the cheapest solution in which vertices of \emph{at most} some
given weight are still connected to the parent of $v$. As we will see, this
approach gives rise to \emph{monotonically decreasing} functions and allows for a
very efficient computation of the DP table.  (2)~We store implicitly the
\emph{size} of all connected components that are created after the cut edges are
removed and that lie completely inside $T_v$. As before, storing these sizes
exactly would result in a very large DP table and, therefore, we store them
concisely using the concept of a \emph{signature}. The signatures will help us
to characterize the sizes of the components inside $T_v$ very efficiently.

\textbf{Signatures.} We call a connected component in $T_v$ \emph{large} if it
contains vertices of total weight at least $\varepsilon \lceil w(V)/k\rceil$ and
otherwise we call it \emph{small}.  Let $t=\lceil
\log_{1+\varepsilon}(1/\varepsilon) \rceil+1$, and let $M = \lceil k/\varepsilon
\rceil+1$. A \emph{signature} is a vector $g=(g_0,\dots,g_{t-1}) \in [M-1]^t$.
Observe that each $g_i$ is an integer between $0$ and $M-1$ and hence there are
$M^t = (k/\varepsilon)^{O(\varepsilon^{-1} \log(1/\varepsilon))}$ different
signatures. Intuitively, an entry $g_i$ in $g$ tells us roughly how many
components of weight $(1+\varepsilon)^i\cdot\varepsilon\lceil w(V)/k\rceil$
there are in the DP solutions that we consider. The precise definition is as
follows.

Let $\mathcal{S}=\{S_1,\dots,S_r\}$ be a set of connected components inside
$T_v$ (e.g., think of $\mathcal{S}$ as the components that are created after
removing the cut edges in the DP solution for vertex $v$).
We say that a signature vector $g=(g_0,\dots,g_{t-1}) \in [M-1]^t$ is
\emph{consistent} for $\mathcal{S}$ if we can \emph{match} the connected
components in $\mathcal{S}$ to entries in $g$ as follows.  For each large
component $S_j$ we let 
$\ell(S_j) = \arg\min\{ i\in[t] \colon
	w(S_j)\le (1+\varepsilon)^i \cdot\varepsilon\lceil w(V)/k\rceil\}$,
i.e., $\ell(S_j)$ is the smallest number $i$ such that $S_j$ has weight at most
$(1+\varepsilon)^i \cdot\varepsilon\lceil w(V)/k\rceil$. Let $s_i\in[M-1]$
denote the number of times the value $i\in[t]$ has been chosen in this process,
i.e., $s_i=\abs{\{ j \colon \ell(S_j) = i \}}$, and let $s=(s_0,\dots,s_{t-1})$
denote the resulting vector. We say that $g$ is \emph{consistent} with the set
of components $\cal S$ if $g=s$.  Thus, the above matching process can be viewed
as rounding up the component sizes and counting the number of components of each
size.

For $x\in\mathbb{N}$, we let $e(x) \in [M-1]^t$ denote the signature of a single
component with total weight $x$. More precisely, we set $e(x)$ to the vector that
has $e(x)_{j}=1$ for
$j=\arg\min\{j\in\mathbb{N} \colon x\le (1+\epsilon)^j\cdot\epsilon\lceil w(V)/k\rceil\}$ and
$e(x)_j=0$, otherwise. If $x<\epsilon\lceil w(V)/k\rceil$, we define
$e(x)=\smash{\vec{0}}$.

\subsubsection{DP Definition}
Now we describe the DP formally.  An entry $\DP(v,g,\cut,x) \in \Winfty$ in the
DP table for a vertex $v$ is indexed by a signature $g$, a
Boolean value $\cut$ and $x\in[n]$. We will consider the tuples $(v,g,\cut)$ as the rows~$\I$
of the DP table and $x$ as the columns; we associate each such row with a
function $\DP(v,g,\cut,\cdot) \colon [n] \to \Winfty$. Note that our DP
has $\abs{V} \cdot M^t \cdot 2
	= (k/\varepsilon)^{O(\varepsilon^{-1} \log(1/\varepsilon))} \cdot n$
rows. Also, note that it has columns $n$; later, even though $x$ 
only takes discrete values, we will allow $x$ to take values in~$[0,\infty)$.

It describes the optimum cost of cutting edges incident on
the subtree $T_v$ (including the cost of maybe cutting the edge to the parent
of~$v$). We will refer to the set of vertices in $T_v$ that are still connected
to the parent of $v$ after the cut edges are removed as the \emph{root component}.
We impose the following conditions on $\DP(v,g,\cut,x)$:
\begin{itemize}
	\item Once the cut edges are removed, the root component $U\subseteq T_v$
	has total weight \emph{at most} $x$, i.e., $w(U)\leq x$.
	\item If $\cut$ is set to true then the edge between $v$ and its parent
	is cut, otherwise it is kept.
	\item The vertices inside $T_v$ that (once the cut edges are removed) are
	\emph{not} connected to the parent of~$v$ form connected components that are
	consistent with the signature $g$.
\end{itemize}

We observe that if we fix a vertex~$v$, a signature~$g$ and a value for $\cut$,
then the resulting function $\DP(v,g,\cut,\cdot)$ is monotonically decreasing in
$x$.  This will be the crucial property for the rest of the section.
\begin{observation}
\label{obs:monotone}
	Let $v\in V$, $g\in[M-1]^t$ be a signature and $\cut\in\{\iscut,\notcut\}$.
	Then the function $\DP(v,g,\cut,\cdot) : [0,\infty) \to \mathbb{R}_+$ is
	monotonically decreasing.
\end{observation}
\begin{proof}
	By definition, $\DP(v,g,\cut,x)$ stores the cost of the optimum solution in
	which the vertices in the root component have weight \emph{at most} $x$. Now
	observe that for $x\leq x'$, the solution $\DP(v,g,\cut,x)$ is also a
	feasible solution for $\DP(v,g,\cut,x')$.  Therefore, $\DP(v,g,\cut,\cdot)$
	must be monotonically decreasing.
\end{proof}

Since the DP cells are monotonically decreasing in $x$, we will use the shorthand
notation $\DP(v,g,\cut,\infty)$ to denote the solution $\min_x \DP(v,g,\cut,x)$.
Note that this minimum is obtained for the largest $x$-value at which
$\DP(v,g,\cut,\cdot)$ changes.

\subsubsection{Computing the DP}
In the following, we describe how to compute $\DP(v,\cdot,\cdot,\cdot)$ exactly.
For computing $\DP(v,\cdot,\cdot,\cdot)$ we simply iterate over all possible
choices of $x$, $g$ and $\cut$. Note that since each vertex has weight in
$\{0,1\}$, the function $\DP(v,g,\cut,\cdot)$ only changes for
$x\in[n+1]$ (i.e., when $x$ is an integer). Thus, we only need to consider $n+1$
choices for $x$. We conclude that to compute $\DP(v,\cdot,\cdot,\cdot)$ for a
fixed vertex $v$, there are $O(M^t \cdot n)$ parameter choices that we need to
iterate over.

In our descriptions we use $p$ to denote the parent of $v$, and $v_l$ and $v_r$
to denote $v$'s left and right child, respectively, if these exist.

\def\bold #1{{\bfseries\mathversion{bold}#1\mathversion{normal}}}

\medskip\noindent
\bold{Case 1: $v$ is a leaf.} If we cut the edge to the parent of $v$, then the
cost is $\capac(v,p)$, there are no vertices in the root component and $v$ forms
its own connected component with signature $e(w(v))$. Thus, we set
$\DP(v,e(w(v)),\iscut,x) = \capac(v,p)$ for all $x\in[0,\infty)$ and we set
$\DP(v,g,\iscut,x) = \infty$ for all $x\in[0,\infty)$ and for all
signatures~$g\neq e(w(v))$.

Now suppose we do not cut the edge~$(v,p)$ to the parent of $v$. Then we do not
have to pay any cost since we are not cutting any edge, the weight of vertices
in the root component is $w(v)$ and the signature is $g=0$ since there are no
connected components in $T_v$ that are not connected to $p$.
Therefore, for all $x\in[0,w(v))$ we set $\DP(v,0,\notcut,x) = \infty$ and
for all $x\in[w(v),\infty)$ we set $\DP(v,0,\notcut,x) = 0$.
For all signatures $g\neq0$ and all $x\in[0,\infty)$, we set
$\DP(v,g,\notcut,x)=\infty$.

\medskip\noindent
\bold{Case 2: $v$ is not a leaf.}
If $v$ is not a leaf then we assume that it has exactly two children $v_l$ and
$v_r$ (if it has only one child, we can add a second child $v'$ with $w(v')=0$,
$\capac(v,v')=0$ and then $v'$ has no impact on the solution).
We assume that for both $v_l$ and $v_r$, we have already computed the solutions
$\DP(v_l,g,\cut,x)$ and $\DP(v_r,g,\cut,x)$ for all possible values of $x$, $g$
and $\cut$.

Let $e_l=(v,v_l)$ and $e_r=(v,v_r)$ denote the edges to the respective child and
let $e_p=(p,v)$ denote the edge to the parent $p$ of $v$.
In the following we distinguish four cases (A, B, C, D) depending on which of
these edges we decide to cut. For each case, we compute
$\DP_{\case}(v,g,\cut,x)$-values, $\case\in\{A,B,C,D\}$, which are the optimum
values under the condition that we cut $e_l$ and $e_r$ according to the case.
The final entry $\DP(v,g,\cut,x)$ is then obtained by minimizing over all cases,
i.e., by setting
$$\DP(v,g,\cut,x) = \min_{\case\in\{A,B,C,D\}} \DP_{\case}(v,g,\cut,x)$$ for all
$x$, $g$, $\cut$.

\medskip\noindent
\bold{Case A: cut $e_l$ and $e_r$.}
Suppose we cut $e_l$ and $e_r$. Then, given $x$ and $g$, we have to select
subsolutions for the left and right sub-tree such that the weight of vertices
that can reach $p$ is at most $x$ and the connected components inside are
consistent with $g$.

First, assume we cut the edge $e_p$. Then the cost for cutting this edge is
$\capac(v,p)$. Furthermore, the weight of vertices inside $T_v$ that can reach
$p$ is zero and, hence, the value of $x$ is irrelevant by the monotonicity of
$\DP(v,g,\cut,\cdot)$. Next, if we have a solution with signatures
$g_l$ and $g_r$ in the left and right subtree, respectively, we can combine
these solutions as long as $g_l+g_r+e(w(v))=g$ (as the vertex $v$ forms a single
		component of weight $w(v)$ since we cut both edges
$e_l$ and $e_r$). Note that in the subsolution for the child~$v_l$, the value of
$x$ does not play a role for the feasibility of the solution $\DP(v,g,\cut,x)$
since the size of the root component in $T_{v_l}$ is already encoded in $g_l$.
Therefore, to obtain minimum cost we consider $\DP(v_l,g_l,\iscut,\infty)$; by
symmetry, the same holds for $v_r$.  Therefore, we set for all $x\in[0,\infty)$,
\begin{equation}
\label{Acut}
\DP_A(v,g,\iscut,x)
  = \capac(v,p)
	+ \min_{g_l+g_r=g-e(w(v))}\{\DP(v_l,g_l,\iscut,\infty)
	+ \DP(v_r,g_r,\iscut,\infty)\}.
\end{equation}

Second, assume we do not cut the edge to the parent $p$. Then there will be at least one vertex
(namely $v$) that can reach $p$. Hence, $\DP_A(v,g,\notcut,x) = \infty$ for all
signatures $g$ and $x\in[0,w(v))$. For $x\in[w(v),\infty)$, we can combine the
solutions as above and we set
\begin{equation}
\label{Anotcut}
\DP_A(v,g,\notcut,x)
  = \min_{g_l+g_r=g}\{\DP(v_l,g_l,\iscut,\infty) + \DP(v_r,g_r,\iscut,\infty)\}.
\end{equation}

\medskip
\noindent
\bold{Case B: cut neither $e_l$ nor $e_r$.} 
Next, suppose we cut neither $e_l$ nor $e_r$. In this case we have to select
subsolutions for $T_{v_l}$ and $T_{v_r}$, where each subsolution is
characterized by the upper bound $x_l$ (resp.\ $x_r$) and its signature $g_l$
(resp. $g_r$).

First, suppose that we cut the edge $e_p$. If we let $x_l$ and $x_r$ denote the
exact weight of the root components for the subsolutions, then the vertex $v$
will be included in a component of size $x_l+x_r+w(v)$ afterwards. Hence, we can
combine the subsolutions to a solution for signature $g$ as long as
$g_l+g_r+e(x_l+x_r+w(v))=g$.  Consequently we set for every $x\in[0,\infty)$,
\begin{align*}
	&\DP_B(v,g,\iscut,x) = \\
	&\capac(v,p)
		+ \min_{x_l,x_r,g_l+g_r=g-e(x_l+x_r+w(v))}
				\DP(v_l,g_l,\notcut,x_l)+\DP(v_r,g_r,\notcut,x_r).
\end{align*}

Second, suppose that we do not cut $e_p$. Then again we have to set
$\DP_B(v,g,\notcut,x) = \infty$ for all signatures $g$ and all $x\in[0,w(v))$,
because the vertex $v$ of weight $w(v)$ can reach $p$.  For $x\ge w(v)$ we have
to select $x_l$ and $x_r$  such that they sum to $x-w(v)$ as this guarantees
that vertices of weight at most $x$ can reach the parent $p$. Consequently, we
set for all $x\in[w(v),\infty)$
\begin{equation*}
	\DP_B(v,g,\notcut,x) =
	\min_{g_l+g_r=g,x_l+x_r=x-w(v)}
		\DP(v_l,g_l,\notcut,x_l)+\DP(v_r,g_r,\notcut,x_r).
\end{equation*}

\medskip
\noindent
\bold{Case C: cut $e_l$ but not $e_r$.}
Now suppose we cut the edge to the left child $v_l$ but we do not cut the edge
to the right child $v_r$.  In this case, $v$ stays connected to the root
component of $v_r$ and we need to choose a subsolution with parameters $x_r$
and $g_r$ for $T_{v_r}$ and a subsolution with parameter $g_l$ for $T_{v_l}$.
Note that since we cut $e_l$, the upper bound on the weight of the root
component of $v_l$ is irrelevant as this is implicitly encoded in $g_l$.

First, suppose we cut $e_p$. If we let $x_r$ denote the exact weight of the root
component for the subsolution in $T_{v_r}$ then $v$ will be included in a
component of size $x_r+w(v)$ afterwards. Hence, we can combine the subsolutions
to a solution for signature $g$ as long as $g_l+g_r+e(x_r+w(v))=g$.
Consequently, for every $x\in[0,\infty)$ we set
\begin{equation}
\label{Ccut}
	\DP_C(v,g,\iscut,x) =
	\capac(v,p)
	+ \min_{x_r,g_l+g_r=g-e(x_r+w(v))}
		\DP(v_l,g_l,\iscut,\infty)+\DP(v_r,g_r,\notcut,x_r).
\end{equation}

Second, suppose we do not cut $e_p$. Then we have to set $\DP_C(v,g,\notcut,x) =
\infty$ for all signatures $g$ and all $x\in[0,w(v))$, because vertex~$v$ with
weight $w(v)$ can reach $p$.  For $x\in[w(v),\infty)$, we have to select $x_r\le
x-w(v)$ as this guarantees that vertices of total weight at most $x$ can reach
the parent $p$. Due to the monotonicity of $\DP(v_r,g_r,\notcut,\cdot)$ we can
just choose $x_r=x-w(v)$.  Consequently, for all $x\in[w(v),\infty)$ we set
\begin{equation}
\label{Cnotcut}
	\DP_C(v,g,\notcut,x) =
	\min_{g_l+g_r=g,x_r=x-w(v)}
		\DP(v_l,g_l,\iscut,\infty)+\DP(v_r,g_r,\notcut,x_r).
\end{equation}

\medskip\noindent
\bold{Case D: cut $e_r$ but not $e_l$.} Symmetric to Case~C.

\bigskip
Next, we argue that this DP is okay-behaved, i.e., it satisfies
Definition~\ref{def:okay-behaved}. In particular, we note that this DP is not
well-behaved because it does not satisfy Property~(4b) of
Definition~\ref{def:well-behaved} since in Case 2, Step~B below we will have to
perform too many $\min$-operations (see Equation~\eqref{eq:too-many-min-app}). 
We will also show that the DP's dependency graph is exactly the input tree and
hence the conditions of Lemma~\ref{lem:tree-dp} are satisfied.  Furthermore, all
entries for a DP cell $\DP(v,\cdot,\cdot,\cdot)$ can be computed in time
$O(M^{2t} n^3)$ by simply enumerating all choices in the different
$\min$-operations above.
\begin{lemma}
\label{lem:exact-dp-time}
	The DP is okay-behaved and the dependency tree and the input tree $T$
	are identical. Furthermore, given a vertex $v$, we can compute all entries
	in $\DP(v,\cdot,\cdot,\cdot)$ in time $O(M^{2t} n^3)$.
\end{lemma}
\begin{proof}
	First, note that in the DP each cell $\DP(v,\cdot,\cdot,\cdot)$ only depends
	on the solutions of its two children. Note that these are
	exactly the edges which are present in the dependency graph and also in $T$.
	Therefore, the dependency graph and $T$ are identical. Furthermore, when the
	input for a child solution is a $\beta$-approximation, the output of the DP
	will also be an $\beta$-approximation because we perform all computations
	exactly. Thus, the DP is also okay-behaved.

	Second, let us consider the running time. Recall that for fixed $x$, $g$ and
	$\cut$, we set $\DP(v,g,\cut,x) = \min_{\case\in\{A,B,C,D\}}
	\DP_{\case}(v,g,\cut,x)$ and this quantity can be computed in time $O(1)$ by
	a simple table lookup. Thus, we only have to consider the time it takes to
	compute $\DP_{\case}(v,g,\cut,x)$ for each $\case\in\{A,B,C,D\}$ and for
	fixed $x$, $g$ and $\cut$.

	For Case~A, observe the $\min$-operations can be computed by iterating over
	all $M^t$ choices of $g_l$ and setting $g_r=g-e(w(v))-g_l$ as long as $g_r$
	is a non-negative vector. Then the expressions inside the $\min$-term can be
	computed by table lookup in constant time. Thus, the time is $O(M^t)$.
	For Case~B, in case $\cut=\iscut$ note that we can iterate over all choices
	of $x_l$, $x_r$ and iterate over $g_l$ as described above. This takes time
	$O(M^t n^2)$. In the case $\cut=\notcut$ we can again iterate over the $g_l$
	as above and we can iterate over all $x_l\in[n+1]$ and set $x_r=x-w(v)-x_l$
	as long as $x_r\geq 0$; thus, the case can be solved in time $O(M^t n)$.
	For Cases~C and~D, we can iterate over all choices of $x_r$ and then iterate
	over the $g_l$ as above. This gives a total running time of $O(M^t n)$.

	We conclude that for fixed $x$, $g$ and $\cut$, the time to compute
	$\DP_{\case}(v,g,\cut,x)$ for all $\case\in\{A,B,C,D\}$ is $O(M^t n^2)$.
	Since there are $O(n)$ choices of $x$, $M^t$ choices for $g$ and two choices
	for $\cut$, we conclude that the total running time to compute
	$\DP(v,\cdot,\cdot,\cdot)$ is $O(M^{2t} n^3)$.
\end{proof}

\subsection{The Approximate DP}
\label{sec:balanced-partitioning-approximate}
In this section we show how to construct the approximate DP table in an
efficient manner. For this we essentially perform the same computations as
above, but instead of computing the exact solution $\DP(v,\cdot,\cdot,\cdot)$
by computing exact solutions to the cases $\DP_{\case}(v,\cdot,\cdot,\cdot)$, we
compute an approximate solution $\ADP(v,\cdot,\cdot,\cdot)$ which will be the
minimum of approximate solutions $\ADP_{\case}(v,\cdot,\cdot,\cdot)$, where
$\case\in\{A,B,C,D\}$.

However, there are a few crucial differences.  First, for fixed $v$, $g$, $\cut$
and $\case\in\{A,B,C,D\}$, we interpret $\tDP_{\case}(v,g,\cut,\cdot)$ as a
piecewise constant function which is stored in an efficient list representation
(as per Section~\ref{sec:approx-conv}). After we computed the solutions
$\tDP_{\case}(v,g,\cut,\cdot)$, we compute the function
\begin{align}
\label{eq:rounding-step}
\begin{split}
	&\tDP(v,g,\cut,\cdot) := \\
	&\lceil \min\{
		\tDP_A(v,g,\cut,\cdot),
		\tDP_B(v,g,\cut,\cdot),
		\tDP_C(v,g,\cut,\cdot),
		\tDP_D(v,g,\cut,\cdot),
	\} \rceil_{1+\delta},
\end{split}
\end{align}
i.e., instead of just taking the minimum over the different cases, we also
perform a rounding step to multiples of $1+\delta$. This rounding step
introduces an approximation error of $\alpha=1+\delta$ but reduces the number of
pieces within the piecewise constant function $\DP(v,g,\cut,\cdot)$ to
$p:=O(\log_{1+\delta}(W))$ according to Lemma~\ref{lem:operations} (for this to
work we need to guarantee that the function to be rounded is monotone and
therefore we will show that $\tDP_{\case}(v,g,\cut,\cdot)$ is monotone for each
$\case\in\{A,B,C,D\}$). The second crucial difference is, of course, that we
perform the above computations with values that already have been rounded, i.e.,
with entries from $\tDP$ instead of entries from $\DP$.  We note that
Equation~\eqref{eq:rounding-step} is the \emph{only} place in the approximate DP
which is not exact; all other computations are done precisely (without any
rounding) and, therefore, the approximate DP only loses a factor $1+\delta$.

In order to guarantee a highly efficient implementation we rely on the
following invariants for entries in the approximate DP:
\begin{enumerate}
	\item For all $v$, $g$, and $\cut$, the function $\tDP(v,g,\cut,\cdot)$ is
		monotonically decreasing.
	\item For all $v$, $g$, and $\cut$, the function $\tDP(v,g,\cut,\cdot)$ is
		piecewise constant with at most $p:=O(\lg_{1+\delta}(W))$ pieces.
\end{enumerate}
Note that the first property resembles the fact that for the exact DP,
$\DP(v,g,\cut,\cdot)$ is monotonically decreasing as per
Observation~\ref{obs:monotone}. However, here we state this property as an
invariant because there could exist approximations of $\DP(v,g,\cut,\cdot)$
which are non-monotone and, therefore, we need to prove that each of our
functions $\tDP(v,g,\cut,\cdot)$ is indeed monotone.
Note the second property follows immediately from the monotonicity and the
rounding step in Equation~\eqref{eq:rounding-step} and thus we will not need to
prove it in the following.

Similar to the description of the exact DP, we will now go through each of the
cases and, given $v$, describe how to compute $\ADP(v,g,\cut,\cdot)$ in time
$\tO(1)$ for all $g$ and $\cut$. The cases are exactly the same as for the exact
DP and thus for the sake of brevity we do not repeat the correctness
argument.

\medskip\noindent
\bold{Case 1: $v$ is a leaf.}
Then, we do the same in the exact case. We set $\tDP(v,e(w(v)),\iscut,x) = \capac(v,p)$ for all
$x\in[0,\infty)$ and we set $\tDP(v,g,\iscut,x) = \infty$ for all
$x\in[0,\infty)$ and all signatures~$g\neq e(w(v))$. Furthermore, we set
$\tDP(v,0,\notcut,x) = \infty$ for all $x\in[0,w(v))$ and
$\tDP(v,0,\notcut,x) = 0$ for all $x\in[w(v),\infty)$. 
For all signatures~$g\neq0$ and all $x\in[0,w(v))$, we set
$\tDP(v,g,\notcut,x)=\infty$.
Note that in all cases, the corresponding functions $\tDP(v,g,\cut,\cdot)$ are
monotonically decreasing and have $O(1)$ pieces.

\medskip\noindent
\bold{Case 2: $v$ is not a leaf.}
We distinguish the same four cases as for the exact DP. Again, we will assume
that $v$ has exactly two children $v_l$ and $v_r$ and we let $e_l=(v,v_l)$,
$e_r=(v,v_r)$ and $e_p=(p,v)$, where $p$ is the parent of~$v$.

\medskip\noindent
\bold{Case A: cut $e_l$ and $e_r$.} First, suppose we cut $e_l$ and $e_r$. Then,
as in the exact DP, if we cut the edge to the parent of $v$, we wish to set
\begin{align*}
&\tDP_A(v,g,\iscut,x) = \\
&\capac(v,p)
		+\min_{g_l+g_r=g-e(w(v))}
			\{\tDP(v_l,g_l,\iscut,\infty) + \tDP(v_r,g_r,\iscut,\infty)\}.
\end{align*}
for all $x\in[0,\infty)$.  Note that in the equation above, the quantities
$\capac(v,p)$, $\tDP(v_l,g_l,\iscut,\infty)$ and $\tDP(v_r,g_r,\iscut,\infty)$
are simply numbers and can be viewed as a piecewise constant function with a
single piece.  Thus, $\tDP_A(v,g,\iscut,\cdot)$ is a piecewise constant function
with a single piece and, therefore, it is also monotonically decreasing.  Hence,
 the invariants are satisfied for $\tDP_A(v,g,\iscut,\cdot)$.  Furthermore,
 $\tDP_A(v,g,\iscut,\cdot)$ can be computed via a sum and a minimum over
 monotonically decreasing piecewise functions via Lemma~\ref{lem:operations}.
 Note that the minimum takes $O(M^{t})$ different values because it is computed
 by iterating over all $g_l\in[M-1]^t$ and setting $g_r=g-e(w(v))-g_l$ as long
 as all entries in $g_r$ are non-negative.  Since each function
 $\tDP(v_l,g_l,\iscut,\cdot)$ has $O(p)$ pieces according to our invariants, we
 can compute the value $\tDP(v_l,g_l,\iscut,\infty)$ in time $O(1)$; the same
 holds for $\tDP(v_r,g_l,\iscut,\infty)$.  Thus, computing
 $\tDP_A(v,g,\iscut,\cdot)$ takes time $O(M^{t})$.

Next, suppose we do not cut the edge to the parent of $v$. Then, as in the exact
DP, we wish to set:
\begin{align*}
\tDP_A(v,g,\notcut,x)
  &= \min_{g_l+g_r=g}\{\tDP(v_l,g_l,\iscut,\infty) + \tDP(v_r,g_r,\iscut,\infty)\}
\end{align*}
for all $x\in[0,\infty)$. Then by the same arguments as above,
$\tDP_A(v,g,\notcut,\cdot)$ is a piecewise constant monotonically decreasing
function with a single piece. 
It can be computed in time $O(M^{t})$ as described above.

\medskip
\noindent
\bold{Case B: cut neither $e_l$ nor $e_r$.} Now suppose we do not cut any edge
to the children. 

If we do not cut the edge to the parent of $v$, we proceed similar to the exact
DP. We start by setting $\tDP_B(v,g,\notcut,x)=\infty$ for all $x\in[0,w(v))$.
Next, for $x\in[w(v),\infty)$ we wish to set
\begin{align}
	\tDP_B(v,g,\notcut,x) &=
	\min_{g_l+g_r=g,x_l+x_r=x-w(v)}
		\tDP(v_l,g_l,\notcut,x_l)+\tDP(v_r,g_r,\notcut,x_r)
		\nonumber \\
	&= 
	\min_{g_l+g_r=g}
	\min_{x_l+x_r=x-w(v)}
		\tDP(v_l,g_l,\notcut,x_l)+\tDP(v_r,g_r,\notcut,x_r).
\label{eq:b-line}
\end{align}
Note that for fixed $g_l$ and $g_r$, the inner $\min$-operation in the second
line describes a $(\min,+)$-convolution due to the constraint $x_l+x_r=x-w(v)$.
Therefore, in the inner $\min$-operation we compute a convolution
$\tDP(v_l,g_l,\notcut,\cdot)\minconv \tDP(v_r,g_r,\notcut,\cdot)$ and shift the
result by $w(v)$ via the shift operation from Lemma~\ref{lem:operations} (where
for $x\in[0,w(v))$ we set $\tDP_B(v,g,\notcut,x) =\infty$). We need time
$O(p^2\log p)$ for computing the convolution according to
Lemma~\ref{lem:convolution}.
To compute the outer minimum in Equation~\eqref{eq:b-line},
we iterate over all $g_l\in[M-1]^t$ and thus perform $O(M^{t})$ minimum
computations over piecewise constant functions with at most $p^2$~pieces. Hence,
we need time $O(M^{t} p^2 \lg(M^{t}p^2))$ according to Lemma~\ref{lem:multimin}.
By Lemma~\ref{lem:monotone-convolution-is-monotone}, $\tDP_B(v,g,\notcut,\cdot)$
is monotonically decreasing since it is the minimum over convolutions of two
monotonically decreasing functions.

If we cut the edge to the parent of $v$, then for all $x\in[0,\infty)$ we would
like to set
\begin{align*}
	&\tDP_B(v,g,\iscut,x) = \\
	&\capac(v,p)
	+\min_{x_l,x_r,g_l+g_r=g-e(x_l+x_r+w(v))}
		\tDP(v_l,g_l,\notcut,x_l)+\tDP(v_r,g_r,\notcut,x_r).
\end{align*}
Note that here we need to be careful as the range of $g_l$ and $g_r$ depends on
the choice of $x_l+x_r$. Since there are $\Omega(n)$ possible values for
$x_l+x_r$, we cannot afford to iterate over all values that $x_l+x_r$ can take.
Instead, we will show that we only need to consider
$O(\lg(k/\varepsilon)/\varepsilon)$ different pairs $(x_l,x_r)$ by exploiting
the monotonicity of $\ADP(v_l,g_l,\notcut,\cdot)$ and
$\ADP(v_r,g_r,\notcut,\cdot)$.

First, observe that we can assume $x_l\le w(T_{v_l})$ and $x_r\le w(T_{v_r})$:
increasing the upper bounds on the weight of the root component further would
mean that the root component contains more weight than \emph{all} vertices
inside the sub-tree, which is impossible. Thus, $x_l+x_r+w(v)\in[1,w(V)]$.

Second, we partition the interval $[1,w(V)]$ into $O(\log(k/\epsilon)/\epsilon)$
intervals. We have intervals $I_j=(\xi_{j-1},\xi_{j}]$ with
$\xi_j = (1+\varepsilon)^j\varepsilon\lceil w(V)/k\rceil$ for all
$j=1,\dots,\log_{1+\epsilon}(k/\epsilon)$. In addition, we add an
\enquote{interval} $I_0:=[\epsilon\lceil w(V)/k\rceil,\epsilon\lceil
w(V)/k\rceil]$ and the interval $I_{-1}:=[1,\epsilon\lceil w(V)/k\rceil)$.
We set $\xi_{0}=\epsilon\lceil w(V)/k\rceil$ and we set $\xi_{-1}$ to the
largest integer that is less than $\epsilon\lceil w(V)/k\rceil$.  Observe that
for all $j\ge -1$ and $x\in I_j$, we have $e(x)=e(\xi_j)$, i.e., the value of
$e(x)$ does not change on in the interval $I_j$.  Below, this property will
allow us to separate the conditions on $x_l+x_r$ and on $g_l+g_r$.

Now we can rewrite the above expression as
\begin{equation*}
\begin{split}
\tDP_B(v,g,\iscut,x) &=\\
\capac(v,p)+
\min_j
&\min_{x_l+x_r+w(v)\in I_j}
\min_{g_l+g_r=g-e(\xi_j)}\tDP(v_l,g_l,\notcut,x_l)+\tDP(v_r,g_r,\notcut,x_r).
\end{split}
\end{equation*}

Third, note that now the two $\min$-operations only depend on the choice of $j$
and, importantly, the minimum over $g_l$ and $g_r$ does not depend on the choice
of $x_l+x_r$ anymore. Therefore, we can swap the order of the two
$\min$-operations.
Furthermore, since $\tDP_B(v,g,\notcut,x)$ is monotonically decreasing with $x$, we can
restrict the choice of $x_l$ and $x_r$ such that $x_l+x_r+w(v)$ is the largest
number in the corresponding interval $I_j$, i.e., $x_l+x_r+w(v)=\xi_j$. Thus,
\begin{align}
\tDP_B(v,g,\iscut,x) &= \capac(v,p)+ \nonumber \\
\min_j
\min_{g_l+g_r=g-e(\xi_j)}\min_{x_l+x_r+w(v)=\xi_j}
&\tDP(v_l,g_l,\notcut,x_l)+\tDP(v_r,g_r,\notcut,\xi_j-x_l-w(v)).
\label{eq:too-many-min-app}
\end{align}

Next, we explain how the above expression can be computed efficiently.
Let us first argue how we can efficiently compute the inner $\min$-operation of
the above expression. We start by observing that this $\min$-operation is
\emph{not} a convolution since in the constraint we sum up to $\xi_i$ which is a
constant (rather than to the variable $x$). Now recall that
$\tDP(v_l,g_l,\notcut,\cdot)$ and $\tDP(v_r,g_r,\notcut,\cdot)$ are piecewise
constant functions with $O(p)$ pieces by our invariants. Since $x_l,x_r\geq 0$
this implies that there are only $O(p^2)$ choices for $x_l$ and $x_r$ such that
$x_l,x_r \in I_j$ and either a new piece starts in $\tDP(v_l,g_l,\notcut,x_l)$
or in $\tDP(v_r,g_r,\notcut,x_r)$. Thus, we can iterate over all these pairs
$(x_l,x_r)$ and evaluate $\tDP(v_l,g_l,\notcut,x_l)+\tDP(v_r,g_r,\notcut,x_r)$,
where $x_r=\xi_j-x_l-w(v)$.  Thus, we can compute the inner $\min$-operation in
time $O(p^2 \lg p)$.

Next, we can compute the outer two $\min$-operations by simply iterating over
$j$ and all choices for $g_l$ and setting $g_r=g-e(\xi_j)-g_l$ as above in
$O(M^{t}\cdot\log(k/\epsilon)/\epsilon)$ iterations.  Hence, we obtain a running
time of $O(M^{t}p^2\log p\cdot\log(k/\epsilon)/\epsilon)$.
We note that this is the step which makes the okay-behaved rather than
well-behaved (since it violates Property~(4b) of
Definition~\ref{def:well-behaved}).

Finally, we note that as $\tDP_B(v,g,\iscut,x)$ is independent of $x$, it is a
constant. Thus, $\tDP_B(v,g,\iscut,x)$ is a piecewise constant function with a
single piece and it is monotonically decreasing.

\medskip
\noindent
\bold{Case C: cut $e_l$ but not $e_r$.} Now suppose we cut the edge to the left
child but not to the right child.

First assume that we cut the edge to the parent of $v$. As in the exact DP, for
all $x\in[0,\infty)$ we want to set
\begin{align*}
&\tDP_C(v,g,\iscut,x) = \\
&\capac(v,p)+\min_{x_r,g_l+g_r=g-e(x_r+w(v))}\tDP(v_l,g_l,\iscut,\infty)+\tDP(v_r,g_r,\notcut,x_r).
\end{align*}
As in the previous case, observe that in the minimum the constraint
$g_l+g_r=g-e(x_r+w(v))$ depends on the choice of $x_r$. Thus, we rewrite the
above equation analogously to the previous case:
\begin{equation*}
\begin{split}
\label{caseC_ADP}
\tDP_C(v,g,\iscut,x) &\\
=\capac(v,p)+&\min_j\min_{x_r+w(v)\in I_j}\min_{g_l+g_r=g-e(\xi_j)}\tDP(v_l,g_l,\iscut,\infty)+\tDP(v_r,g_r,\notcut,x_r)\\
=\capac(v,p)+&\min_j\min_{g_l+g_r=g-e(\xi_j)}\min_{x_r+w(v)\in I_j}\tDP(v_l,g_l,\iscut,\infty)+\tDP(v_r,g_r,\notcut,x_r)\\
=\capac(v,p)+&\min_j\min_{g_l+g_r=g-e(\xi_j)}\tDP(v_l,g_l,\iscut,\infty)+\tDP(v_r,g_r,\notcut,\xi_j-w(v)),
\end{split}
\end{equation*}
where in the last step we used that $\tDP(v_r,g_r,\notcut,\cdot)$ is
monotonically decreasing.  The evaluation of the function values of the two
piecewise constant functions with $O(p)$ pieces can be done in time
$O(\log(p))$. Furthermore, by
exhaustively enumerating all choices for $j$ and proceeding for $g_l$ and $g_r$
as above, we obtain
$O(M^{t}\log(k/\epsilon)/\epsilon)$ iterations giving a total running time of
$O(M^{t} \log p\log(k/\epsilon)/\epsilon)$.
As before, $\tDP_C(v,g,\iscut,\cdot)$ is a constant (since the computation does
not depend on $x$) and therefore it has only a single piece and it is
monotonically decreasing.

Next, suppose we do not cut the edge to the parent of $v$. Then we set
$\DP_C(v,g,\notcut,0) = \infty$ for all signatures $g$ and all $x\in[0,w(v))$.
For all $x\in[w(v),\infty)$, we set
\begin{align*}
	\tDP_C(v,g,\notcut,x) =
	\min_{g_l+g_r=g}\tDP(v_l,g_l,\iscut,\infty)+\tDP(v_r,g_r,\notcut,x-w(v)).
\end{align*}
Note that inside the $\min$-operation, the first term is a constant and the
second term is a piecewise constant function that is shifted by $w(v)$.
Furthermore, the minimum is taken over $O(M^{t})$ piecewise constant functions
(one for each choice of $g_l$ by the same argument as above). We can perform
the addition and shift operation via Lemma~\ref{lem:operations} (time $O(p\lg
p)$ per application). Then
we perform a minimum operation over $M^t$ functions where each function
has just $p$ pieces. This can be done in time $O(M^tp\lg(M^tp))$ by
Lemma~\ref{lem:multimin}. In total we get a running time of
$O(M^tp\lg(M^tp))$.

\bold{Case D: cut $e_r$ but not $e_l$.} Symmetric to Case~C.

\bigskip 
We conlucde this subsection with the following lemma which summarizes the
properties of the approximate DP computation
The lemma follows immediately from
the above discussion.
\begin{lemma}
\label{lem:property-adp}
	The approximate DP computes a $(1+\delta)$-approximate
	DP solution and the dependency tree and the input tree $T$ are identical. 
	Given a vertex $v$, a signature $g$ and value $\cut\in\{\iscut,\notcut\}$,
	we can compute the corresponding approximate DP entry $\ADP(v,g,\cut,\cdot)$
	in time $O(M^{t}p^2\log(M^tp)\log(k/\epsilon)/\epsilon))$.
\end{lemma}
\begin{proof}
	The approximation ratio of the approximate DP is $(1+\delta)$-approximate
	because, as we pointed out earlier, we only use exact computations except in
	the rounding step in Equation~\eqref{eq:rounding-step}. Thus, we only use
	$(1+\delta)$-factor in the computation.

	The claim about the running time follows immediately from the discussion
	above the lemma, where we already analyzed the running times for all steps.
\end{proof}

\subsection{Computing the Result}
\label{sec:computing-result}
In this section we describe how the previously described DPs can be used to
extract the result for the $k$-balanced partition problem. Recall that we
consider the generalized version of the $k$-balanced partition problem, where
each vertex~$v$ has a weight $w(v)\in\{0,1\}$ (see
Section~\ref{sec:balanced-partitioning-exact} for the definition).

We focus on the \emph{value version} of the problem in which we only need to
output an approximation of the \emph{value} of the optimal cut $\OPT$ but we do
not have to return the actual partition $V_1,\dots,V_k$ that obtains this cut
value.  We note, however, that by analyzing the DP solution from top to bottom,
we could also construct a concrete partition $V_1,\dots,V_k$ in time $\tO(n)$
that achieves the cut value which is returned by the value version.

\textbf{Feasible Signatures.} Before we describe our algorithm, we first need to
introduce the notion of \emph{feasible} signatures. More concretely, recall that
in Section~\ref{sec:balanced-partitioning-exact} we introduced signatures as a
succinct way of storing the sizes of connected components in a solution. Now,
feasible signatures will refer to signatures in which the connected components
can be partitioned such that we obtain a nearly $k$-balanced partitioning of the
vertices. We make this intuition more formal below.

For every signature $g=(g_0,\dots,g_{t-1}) \in [M-1]^t$, we say that its
associated machine scheduling instance\footnote{
	Recall that in the makespan minimization problem with identical machines,
	the input consists of a set of $N$ jobs of sizes $s_1,\dots,s_N$ and an
	integer $k$.  The goal is to find an assignment of the jobs to $k$ machines
	such that the makespan is minimized. Here, the makespan refers to maximum
	load of all $k$ machines.}
$I(g)$ is the instance which contains exactly $g_i$ jobs of size
$(1+\varepsilon)^{i} \cdot \varepsilon \lceil w(V)/k\rceil$ of all $i$.  We say
that $g$ is a \emph{feasible} signature if the jobs in $I(g)$ can be scheduled
on $k$ machines with makespan at most $(1+\epsilon)\lceil w(V)/k\rceil$.  Later,
we will identify the machines of the scheduling problems with partitions in the
$k$-balanced partitioning solution and the jobs with connected components. In
this way, we will be able to ensure the balance constraints of the $k$-balanced
partitioning solution.

\textbf{Algorithm.}
We now describe our two static algorithms for binary trees. The only difference
between the algorithms is whether to use the exact DP from
Section~\ref{sec:balanced-partitioning-exact} or the approximate DP from
Section~\ref{sec:balanced-partitioning-approximate}; we will refer to these
algorithms as the \emph{exact} and the \emph{approximation} algorithm,
respectively.  We assume that the input is an error parameter $\varepsilon>0$
and a rooted, weighted tree $T=(V,E,\capac)$ with root $r$ and vertex weights
$w(v)\in\{0,1\}$ for which we wish to solve the $k$-balanced partitioning
problem.

First, our algorithm augments $T$ by adding a \emph{fake root} $r'$. We make
$r'$ the parent of $r$ and set $w(r')=0$ and $\capac(r,r')=0$. Then we compute
the DP bottom-up as described in Section~\ref{sec:dps-on-trees}, where we
interpret $T$ as its own dependency graph. In the exact algorithm, we use the DP
from Section~\ref{sec:balanced-partitioning-exact}, and in the approximation
algorithm, we use the DP from
Section~\ref{sec:balanced-partitioning-approximate}.

Second, we compute the set of all \emph{nearly} feasible signatures. To obtain
this set, we enumerate all $M^t$ signatures and for each of them, we check
whether it is nearly feasible or not. We do this as follows.  For each
signature~$g$, we construct the machine scheduling instance $I(g)$ and run the
PTAS by Hochbaum and Shmoys~\cite{hochbaum87using} for this problem with
approximation ratio $1+\PTASeps$ and running time
$(N/\PTASeps)^{O({1/\PTASeps^2)}}$, where $N$ denotes the total number of jobs
in $I(g)$ and we will see later that $N$ is a constant if $k$, $\varepsilon$ and
$\PTASeps$ are constants.  We add a signature~$g$ to the set of nearly feasible
signatures if the returned makespan for $I(g)$ is at most
$(1+\PTASeps)(1+\varepsilon)\lceil w(V)/k\rceil$.  We note that by using the
PTAS, the set that we compute can potentially contain some signatures which are
infeasible but they still do not violate the balance constraint too much.

Third, we consider the entries in the DP table at the (true) root~$r$ of the
tree for the case that the edge to its (artificial) parent is not cut (recall
that we added an edge of weight $0$ from the true root~$r$ to the fake root~$r'$
and so cutting it does not incur any cost), i.e., we consider the DP entries
$\DP(r,\cdot,\iscut,w(V))$ or $\ADP(r,\cdot,\iscut,w(V))$ depending on whether
we are in the approximate or in the exact case. We iterate over all
\emph{feasible} signature vectors $g$ and then take the minimum value that we
have seen.

We conclude the algorithms' guarantees in the following proposition. We
note that for constant $k$, $\varepsilon$, $\PTASeps$ and $W$, the running time
of the exact algorithm is $\tO(n^4)$ and the running time of the approximation
algorithm simplifies to $\tO(n\cdot h^2)$, where $h$ is the height of the input
tree.  Thus, for trees of height $\tO(1)$, the approximation algorithm is very
efficient and runs in time $\tO(n)$.
\begin{proposition}
\label{prop:binary-trees}
	Let $\varepsilon, \PTASeps > 0$ and $k\in\mathbb{N}$.
	Let $T=(V,E,\capac)$ be a rooted binary tree that has edge weights
	$\capac(e)$ and vertex weights $w(v)\in\{0,1\}$. Then:
	\begin{itemize}
		\item The exact algorithm obtains a bicriteria
		$(1,(1+\PTASeps)(1+\epsilon))$-approximation for the $k$-balanced
		partitioning problem on $T$ in time $O(M^{2t} n^4)$.
		\item The approximation algorithm obtains a bicriteria
		$(1+\varepsilon,(1+\PTASeps)(1+\epsilon))$-approximation for the $k$-balanced
		partitioning problem on $T$ in time
                $O\big(nh^2\cdot M^{2t}\log^2(W)
          \log(k/\varepsilon)\log(M^th
          \log(W)/\varepsilon)/\varepsilon^3\big)
			+M^t(k/(\epsilon\PTASeps))^{O(1/\PTASeps^2)}$,
		where $h$ denotes the height of~$T$.
	\end{itemize}
\end{proposition}
\begin{proof}
	To prove the proposition, we need to argue about the approximation ratios of
	the algorithms and we also need to prove that the partitioning does not
	violate the balance constraints. We will also need to analyze the running
	times.

	We start by analyzing the balance constraints. We show that in the solution
	returned by the algorithm, the connected components $V_1,\dots,V_k$ can be
	partitioned such that $w(V_i) \leq (1+\PTASeps)(1+\varepsilon)\lceil
	w(V)/k\rceil$ for all $i=1,\dots,k$.

	Consider the DP entry $\DP(r,g,\iscut,w(V))$ for the (true) root~$r$, where
	the edge to the parent is cut and any signature vector $g$ that is in the
	set of nearly feasible signatures that we computed.  Then this corresponds
	to the cost of some partition of $T=T_r$ where, after removing the cut
	edges, the large connected components $\S$ in $T_r$ can be matched to
	entries in $g$ such that:
	\begin{itemize}
		\item a component $S\in \S$ is matched to entry $g_i$ with 
			$\abs{S}\le (1+\epsilon)^i\epsilon\lceil w(V)/k\rceil$ and
		\item exactly $g_i$ components are matched to $g_i$.
	\end{itemize}
	Hence, we can obtain a partitioning $V_1,\dots,V_k$ as follows. First, we
	compute the $(1+\PTASeps)$-approximate solution of $I(g)$ in which (by
	assumption on $g$) the makespan is at most
	$(1+\PTASeps)(1+\epsilon)\lceil w(V)/k\rceil$.  This gives us an assignment
	of jobs to machines. Now we identify components with jobs and the sets $V_i$
	with machines and obtain an assignment of the large components to the
	$V_i$.  In particular, each $V_i$ receives large components for which the
	(rounded) weights sum to at most $(1+\PTASeps)(1+\epsilon)\lceil w(V)/k\rceil$.
	Now we need to assign the small components in the algorithm's solution.
	These can be assigned greedily by always assigning a small component (of
	weight less than $\epsilon\lceil n/k\rceil$) to set $V_i$ of (currently)
	smallest weight. In the end, all $V_i$ will have weight at most
	$(1+\PTASeps)(1+\epsilon)\lceil w(V)/k\rceil$ (this follows from the
	standard argument that, when considering exact component weights, on average
	each server has makespan at most $w(V)/k$ and thus there will always be a
	server of makespan at most $w(V)/k$ to which the current small component can
	be assigned without violating the capacity constraint).  This means if the
	algorithm returns an objective function value then there is a partition
	$V_1,\dots,V_k$ with the same objective function value that is nearly
	feasible, i.e., that satisfies $w(V_i) \leq
	(1+\PTASeps)(1+\varepsilon)\lceil w(V)/k\rceil$ for all $i=1,\dots,k$.

	Next, let us consider the approximation ratios of the algorithms. Consider
	the optimum partition $\OPT=(V_1^*,\dots,V_k^*)$ that minimizes
	$\cut(V_1^*,\dots,V_k^*)$ such that $w(V_i^*)\leq \lceil w(V)/k\rceil$ for all
	$i$. We first argue that $\OPT$ gives rise to a DP entry with a feasible
	signature and cost $\OPT$ in the exact DP.  To see this, take the optimum
	partition $V_1^*,\dots,V_k^*$ and round up the weight of every large
	connected component to the next value of the form
	$(1+\epsilon)^i\cdot\epsilon\lceil n/k\rceil$. Let
	$g=(g_0,\dots,g_{t-1})\in[M-1]^t$ be the signature where $g_i$ denotes the
	number of large components in $\OPT$ whose rounded weight is
	$(1+\epsilon)^i\cdot\epsilon\lceil n/k\rceil$.  Note that since
	$w(V_i^*)\leq\lceil w(V)/k\rceil$ for all $i$, the total \emph{rounded}
	weight of components in $V_i^*$ is at most $(1+\epsilon)\cdot\lceil w(V)/k\rceil$
	as component weights are increased at most by a $(1+\epsilon)$-factor.
	Hence, the constructed signature vector $g$ is feasible because the
	partition $V_1^*,\dots,V_k^*$ gives rise to a feasible solution for $I(g)$.
	Furthermore, the rounding did not have any effect on the objective function
	value and, thus, $\OPT$ gives rise to a DP entry with a feasible signature
	and cost $\OPT$ in the exact DP. This implies that the optimum value for the
	exact DP is at most $\cut(V_1^*,\dots,V_k^*)$.  Together with the above
	claim that the DPs approximately satisfy the balance constraint, we obtain
	that the exact algorithm computes a bicriteria
	$(1,(1+\PTASeps)(1+\epsilon))$-approximation.

	Now let us turn to the approximation ratio of the approximation algorithm
	from Section~\ref{sec:balanced-partitioning-approximate}. Recall that by
	Lemma~\ref{lem:exact-dp-time} the exact DP is okay-behaved and in
	Lemma~\ref{lem:property-adp} we show that in each step the approximation
	algorithm loses a factor of at most $1+\delta$ at every level of the
	tree~$T$. Now, we can apply Lemma~\ref{lem:tree-dp} to obtain that the
	approximation in the root is $(1+\delta)^{h+1}$, where $h$ is
	the height of the tree~$T$.  Thus, the approximation ratio of the
	approximate DP is $1+\varepsilon$ if we set
	$\delta=\ln(1+\varepsilon)/(h+1)$ since then
	$(1+\delta)^{h+1} \leq \exp(\delta (h+1)) = 1+\varepsilon$. Since the
	notion of approximation from Lemma~\ref{lem:tree-dp} holds for all functions
	of the form $\ADP(r,g,\iscut,\cdot)$ and all possible values of $x$, we
	obtain that the approximation algorithm computes a bicriteria
	$(1+\varepsilon,(1+\PTASeps)(1+\epsilon))$-approximation.

	We conclude the proof of the proposition by considering the running times of
	the algorithms. Note that w.r.t.\ running time, both algorithms only differ
	by how long it takes to fill the DP cells and the time for computing the
	solution is the same.

	Let us first consider the time for computing the solution as per
	Section~\ref{sec:computing-result}. First, let us consider the time for
	solving the PTAS which is $(N/\PTASeps)^{O({1/\PTASeps^2)}}$, where $N$
	denotes the total number of jobs.
	Note that in our case there are at most $N\leq k(1+1/\epsilon)$ jobs:
	each job has size at least $\epsilon\lceil n/k\rceil$ and therefore a
	machine can take at most $1+1/\epsilon$ jobs in an optimum solution. Hence, if
	we have more than $k(1+1/\epsilon)$ jobs, a PTAS can directly reject the instance
	and declare it infeasible.  Thus, the time for running the PTAS a single
	time is $(k/(\epsilon\PTASeps))^{O(1/\PTASeps^2)}$. Since we have to run the
	PTAS for each of the $M^t$ signatures, the total time for finding the nearly
	feasible configurations is $M^t (k/(\epsilon\PTASeps))^{O(1/\PTASeps^2)}$.

	Finally, let us consider the time for filling the DP cells. For the
        exact DP, Lemma~\ref{lem:exact-dp-time} states that filling a cell
        $\DP(v,\cdot,\cdot,\cdot)$ takes time $O(M^{2t}n^3)$. Then, by applying
        Lemma~\ref{lem:tree-dp}, the total time to compute all DP cells is
        $O(M^{2t}n^4)$. For the approximate DP, it takes time
        $O(M^{t}p^2\log(M^tp)\log(k/\epsilon)/\epsilon))$ to fill a single
        DP cell $\ADP(v,g,\cut,\cdot)$ by Lemma~\ref{lem:property-adp}. Since
        there are $M^t$ choices for $g$ and by again applying
        Lemma~\ref{lem:tree-dp}, we obtain that the total running time for
        filling the approximate DP table is
        $O(nM^{2t}p^2\log(M^tp)\log(k/\epsilon)/\epsilon))$.
		Since in
        Section~\ref{sec:balanced-partitioning-approximate} we picked the
        number of pieces to be $p=O(\lg_{1+\delta}(W))$ and above we picked
        $\delta=O(\varepsilon/h)$, the running time is upper bounded by
        $O\big(nM^{2t} \cdot \big(\nicefrac{1}{\varepsilon}\cdot h \log
          W\big)^2 \cdot \log(k/\varepsilon)/\varepsilon\cdot\log(M^th
          \log(W)/\varepsilon)\big)=O\big(nh^2\cdot M^{2t}\log^2(W)
          \log(k/\varepsilon)\log(M^th
          \log(W)/\varepsilon)/\varepsilon^3\big)$.
\end{proof}

\subsection{Extension to General Graphs}
\label{sec:extension-balanced}
Now we generalize the results of Proposition~\ref{prop:binary-trees} from binary
trees to general graphs.

We start with the generalization to general graphs in which we will make use of
Räcke trees (see Section~\ref{sec:racke-tree}). Since Räcke trees might be
non-binary, we now introduce the notion of \emph{binarized Räcke trees} which
essentially describe a way of turning a non-binary Räcke tree into a binary tree
that is very similar to a Räcke tree. Later, the binarized Räcke trees will
allow us to apply Proposition~\ref{prop:binary-trees} on them.
\label{sec:extension-non-binary}
\begin{definition}[Binarized Räcke Tree]
	Let $G=(V_G,E_G,\capac_G)$ be a weighted graph. We say that a weighted,
	rooted tree $T=(V_T,E_T,\capac_T)$ is a \emph{binarized Räcke tree for $G$}
	if the following properties hold:
	\begin{itemize}
		\item $T$ is a rooted binary tree.
		\item $V_G\subseteq V_T$.
		\item All edges in $T$ have weights in $\Winfty$.
		\item Let $T'$ be the tree that is obtained by contracting all edges with
			weight $\infty$ in $T$. Then $T'$ is a Räcke tree for $G$. 
	\end{itemize}
	We call the tree $T'$ from the last bullet point the \emph{corresponding
	(non-binarized) Räcke tree of $T$}.  We say that $T$ has quality $q$ if the
	corresponding Räcke tree $T'$ has quality $q$.
\end{definition}
Next, we observe that each cut in $T$ \emph{of finite cost} corresponds to a cut
in the corresponding Räcke tree $T'$ and vice versa. Therefore, cuts of finite
cost in $T'$ approximate the cut structure of the initial graph $G$. We
make this more formal in following observation.
\begin{observation}
\label{obs:binarized-racke-tree}
	Let $G=(V_G,E_G,\capac_G)$ be a weighted graph and let
	$T=(V_T,E_T,\capac_T)$ be a binarized Räcke tree for $G$ with quality $q$.
	Then for all disjoint subsets $A,B\subseteq V_G$ it holds that
	$\mincut_G(A,B)\le\mincut_T(A,B)\le q\cdot\mincut_G(A,B)$.
\end{observation}
\begin{proof}
	Let $T'$ be the corresponding (non-binarized) Räcke tree of $T$ and consider
	two disjoint subsets of vertices $A,B\subseteq V_G$. We show that
	$\min_T(A,B)=\mincut_{T'}(A,B)$. Then the observation follows immediately
	since $T$ has quality~$q$ (by assumption) and, therefore, $T'$ is a Räcke
	tree for $G$ with quality~$q$ which satisfies the property from the
	observation.

	Since $T'$ can be obtained from $T$ only by contracting edges, we have
	$\mincut_{T'}(A,B)\geq\mincut_{T}(A,B)$. Next, let us argue that
	$\mincut_T(A,B)\geq\mincut_{T'}(A,B)$. First, note that
	$\mincut_{T'}(A,B)\leq q\cdot \mincut(A,B) < \infty$. Since we contract only
	edges with weight~$\infty$ to go from $T$ to $T'$, $T$ does not contain any
	cut with finite cost that is not contained in $T'$. Therefore,
	$\mincut_T(A,B)\geq\mincut_{T'}(A,B)$.
\end{proof}

Additionally, we show that we can compute a binarized Räcke tree of good quality
in nearly-linear time.
\begin{lemma}
\label{lem:binary-racke-tree}
	Let $G=(V_G,E_G,\capac_G)$ be a weighted graph with $n$~vertices and
	$m$~edges.  We can compute a binarized Räcke tree~$T=(V_T,E_T,\capac_T)$
	with $O(n)$ vertices, height $O(\log^2 n)$ and quality $O(\lg^4 n)$ in time
	$\tO(m)$.
\end{lemma}
\begin{proof}
	Let $T'=(V_{T'},E_{T'},\capac_{T'})$ be the Räcke tree for $G$
	from Theorem~\ref{thm:racke-tree} that can be computed in time $\tO(m)$.
	First, note that $T'$ has
	$n_{T'}:=O(n)$ vertices and height $O(\log n)$.  Second, note that $T'$
	can have unbounded degree. Therefore, we will show how to compute a binarized
	Räcke tree $T$ that has $T'$ as its corresponding (non-binarized) Räcke
	tree. We do so replacing in $T'$ each vertex $u$ by a
	balanced binary tree $\tau_u$ with $\deg(u)$ leaves, where $\deg(u)$ denotes
	the number of children of $u$. The internal edges of
	$\tau_u$ will have weight $\infty$ and the edges connecting subtrees
	$\tau_u$ and $\tau_v$, $u\neq v$, in $T$ will correspond to the edges in
	$T'$ and will have the same (finite) weight as in~$T'$. We will see that by
	contracting all edges with weight $\infty$ in $T$, we will obtain $T'$. We
	now elaborate on this process.

	We construct $T$ as follows. First, we compute $T'$ as per the algorithm
	from Theorem~\ref{thm:racke-tree}. Now we construct $T$ as follows. For each
	vertex $u\in V_{T'}$, we add a balanced rooted binary tree $\tau_u$ with
	$\deg(u)$ leaves. We refer to the root of $\tau_u$ as $r_u$.  We identify
	each leaf of $\tau_u$ with a child of $u$ and denote the leaf of $\tau_u$
	that corresponds to the child $v$ by $c_{u,v}$. We set the weight of edges
	inside $\tau_u$ to $\infty$. Note that for each vertex $u$, the tree
	$\tau_u$ has $O(\deg(u))$ vertices, and, therefore, $T$ has $O(n_{T'}) =
	O(n)$ vertices.  Next, for each edge $(u,v)\in E_{T'}$ (where we assume that
	$v$ is a child of $u$), we insert the edge $(c_{u,v},r_v)$ in $T$
	and set $\capac_{T}(c_{u,v},r_u)=\capac_{T'}(u,v)$.  Finally, if $u$ is the
	root of $T'$ then we set $r_u$ to the root of $T$.

	It is left to show that $T$ is a binarized Räcke tree of height
	$O(\log^2 n)$ and quality $O(\log^4 n)$.  Clearly, $T$ is a binary tree
	since all vertices inside each subtree $\tau_u$ have at most two child nodes
	and, additionally, each vertex $c_{u,v}$ has at most one child node (namely
	$r_v$). Next, $T$ has height $O(\lg^2 n)$ since $T'$ has height $O(\lg n)$
	and the subtrees $\tau_u$ have height $O(\lg n)$. Finally, let $T''$ be the
	tree obtained from $T$ by contracting all edges with weight $\infty$. We
	argue that $T'=T''$. Indeed, consider any vertex $u\in V_{T'}$ and its
	subtree $\tau_u$ in $T$. Then after contracting the edges in $\tau_u$, we
	are left with a subtree that only contains $r_u$. Furthermore, all edges
	between vertices of different subtrees $\tau_u$ and $\tau_v$, $u\neq v$,
	have finite weight.  Therefore, $T'=T''$. This implies that $T$ is binarized
	Räcke tree for $G$.  Since $T'$ has quality $O(\log^4 n)$, the quality of
	$T$ is also $O(\log^4 n)$.
\end{proof}

We conclude the subsection by proving Theorem~\ref{thm:partitioning-static}.

\begin{proof}[Proof of Theorem~\ref{thm:partitioning-static}]
	We can obtain the proof for the claim about general graphs as follows. Let
	$G=(V_G,E_G,\capac_G)$ be a weighted graph with $n$ vertices. We compute a
	binarized Räcke tree~$T=(V_T,E_T,\capac_T)$ with $O(n)$ vertices as per
	Lemma~\ref{lem:binary-racke-tree} in time $\tO(n)$. In $T$, we assign weight
	$w(v)=1$ to all vertices $v\in V_G\cap V_T$ (i.e., to the leaves in $T$ that
	correspond to vertices in $G$) and weight $w(v)=0$ to all vertices $v\in
	V_T\setminus V_G$ (i.e., to the internal nodes of $T$ that do not correspond
	to any vertex in $G$). Now observe that $w(V)=n$ and thus a balanced
	partitioning $V_1,\dots,V_k$ of $T$ with
	$w(V_i)\leq(1+\varepsilon)\lceil w(V)/k\rceil$ for all $i$ corresponds to a
	balanced partitioning $V_1',\dots,V_k'$ of $G$ with
	$\abs{V_i'}\leq(1+\varepsilon)\lceil n/k\rceil$ for all $i$, where
	$V_i'=\{v \in V_i \colon w(v)=1\}$.  Now by combining
	Observation~\ref{obs:binarized-racke-tree},
	Proposition~\ref{prop:binary-trees} and the fact that $T$ has quality
	$O(\lg^4 n)$, we obtain the claim.

	To obtain the result about general trees~$T'$ (with unbounded degrees), we
	proceed similarly. We construct a binarized tree~$T$ exactly as in the proof
	of Lemma~\ref{lem:binary-racke-tree}. Now, in $T$ we set $w(r_v)=1$ for all
	root vertices of the subtrees $\tau_v$ and we set $w(v)=0$ for all other
	vertices of the subtrees $\tau_v$.  Similar to before, observe that
	$w(V_T)=n$ and thus a balanced partitioning $V_1,\dots,V_k$ of $T$ with
	$w(V_i)\leq(1+\varepsilon)\lceil w(V)/k\rceil$ for all $i$ corresponds to a
	balanced partitioning $V_1',\dots,V_k'$ of $T'$ with
	$\abs{V_i'}\leq(1+\varepsilon)\lceil n/k\rceil$ for all $i$, where
	$V_i'=\{v \in V_{T'} \colon r_v \in V_i\}$.  Then by
	Proposition~\ref{prop:binary-trees}, this implies the proof for trees with
	unbounded degrees.
\end{proof}

\subsection{Extension to the Dynamic Setting}
\label{sec:extension-dynamic}
Next, we provide new dynamic algorithms in which edges are inserted and deleted
from the graph. We give new algorithms for trees and for general graphs.

\textbf{Extension to Dynamic Trees.}
Let us start with the case when $T$ is a binary tree that is undergoing edge
insertions and deletions. We will use Lemma~\ref{lem:tree-dp-dynamic} to make
the result from Proposition~\ref{prop:binary-trees} dynamic. However, there is a
slight technical difficulty: due to edge deletions, $T$ will become a forest and
fall apart into several connected components. This becomes an issue, when an
edge $(u,v)$ is inserted for which both $u$ and $v$ already have parents in
their respective components. In that case, we cannot immediately make $u$ the
root of $v$ (or vice versa).  Therefore, we need to find an efficient way of
re-rooting the tree containing $v$, i.e., we need to make $v$ the root of
its component and we need to ensure that we do not have to recompute the DP
solution for all vertices in the component of $v$. We now describe our dynamic
algorithm in more detail.

First, suppose that an edge $(u,v)$ is removed from $T$ and assume that (before
the edge deletion) $u$ is closer to the root of $T$ than $v$. Then $T$ becomes a
forest with multiple connected components. In that case, we make $v$ the root of
its component and recompute the DP solution for $v$ (since $v$ does not have a
parent, we only have to recompute the DP cell for $v$).  Furthermore, for $u$
and all of its ancestors we recompute the DP solution as per
Lemma~\ref{lem:tree-dp-dynamic}.

Next, suppose an edge $(u,v)$ is inserted, where $u$ and $v$ are in different
connected components. Further suppose that after the edge insertion, $u$ is the
parent of $v$. Then we distinguish two cases whether $v$ is the root of its
component or not.

First, suppose that $v$ is the root of its component. Then we simply insert the
edge $(u,v)$ into $T$ and recompute the solution for $v$ and all of its
ancestors (including $u$) as per Lemma~\ref{lem:tree-dp-dynamic}.

Second, suppose that neither $u$ nor $v$ is the root of its component.  Now, we
first have to re-root the component containing $v$ such that it has $v$ as its
root and such that all DP solution are valid. We do this as follows. Let
$v=v_1,\dots,v_\ell$ denote the vertices on the path from $v$ to the
root~$v_{\ell}$ of its component (before the edge insertion). Then we first
remove all edges $(v_{\ell},v_{\ell-1}),\dots,(v_2,v_1)$ from $T$ (in this
order) as per the edge deletion routine described above. Note that after the
deletions, none of the $v_i$ has a parent and, therefore, each $v_i$ is the root
of its own component.  Furthermore, by how we picked the order of the edge
deletions, after the $i$'th deletion we only have to recompute the DP cells for
the vertices $v_{\ell-i}$ and $v_{\ell-i-1}$. Now we insert the edges again but
with flipped direction, i.e., we insert the edges
$(v_{\ell-1},v_{\ell}),\dots,(v_1,v_2)$ (in this order). Thus, $v=v_1$ becomes
the root of the component. To insert the edges, we use the subroutine from the
paragraph above, where we exploit that each $v_i$ is the parent of its own
component, which implies that the DP solutions can be updated efficiently: by
how we picked the order of the edge insertions, after the $i$'th edge insertion
we only need to recompute the DP cells for vertices $v_{\ell-i-1}$ and
$v_{\ell-i}$. After the rebalancing of the component containing $v$ is done, $v$
has become the parent of its component and, therefore, we can use the routine
from above to insert the edge $(u,v)$. This concludes the edge insertion
procedure.

Next, when we want to output the value of the DP solution, we simply use the
subroutine described in Section~\ref{sec:computing-result}.

We summarize the guarantees of our dynamic algorithm in the following
proposition. Note that when the parameters $\varepsilon$, $\PTASeps$,
$k$ and $W$ are constants, the update time becomes $\tO(h^3)$ and the query time
is just $O(1)$. Therefore, the
algorithm is very efficient for trees that have polylogarithmic or subpolynomial
height in the number of vertices.
\begin{proposition}
\label{prop:dynamic-approx-dp}
	Let $\varepsilon, \PTASeps > 0$ and $k\in\mathbb{N}$.  Let $T=(V,E,\capac)$
	be a rooted binary tree with edge weights $\capac(e)\in\Winfty$ and vertex
	weights $w(v)\in\{0,1\}$, that is undergoing edge insertions and deletions.
	Let $h$ be an upper bound on the height of the tree $T$ at all times. Then
	there exists a fully dynamic algorithm that maintains a bicriteria
	$(1+\varepsilon,(1+\PTASeps)(1+\epsilon))$-approximation for the
	$k$-balanced partition problem on $T$ with update time
        $O\big(h^3\cdot M^{2t}\log^2(W)
          \log(k/\varepsilon)\log(M^th
          \log(W)/\varepsilon)/\varepsilon^3\big)$
	and query time 
	$M^t(k/(\epsilon\PTASeps))^{O(1/\PTASeps^2)}$.
\end{proposition}
\begin{proof}
	The fact that the algorithm maintains a bicriteria
	$(1+\varepsilon,(1+\PTASeps)(1+\epsilon))$-approximation follows immediately
	from Lemma~\ref{lem:tree-dp-dynamic} and the same arguments as in the proof
	of Proposition~\ref{prop:binary-trees}, where we argued that the approximate
	DP satisfies the conditions of Lemma~\ref{lem:tree-dp}.

	It is left to analyze the update and query times.
	For the query times, note that all we do is run the subroutine from
	Section~\ref{sec:computing-result}. This subroutine runs in time
	$M^t(k/(\epsilon\PTASeps))^{O(1/\PTASeps^2)}$ as we argued in the proof of
	Proposition~\ref{prop:binary-trees}. This proves the claim about the query
	time.

	For the update times, let us first consider edge deletions $(u,v)$. In this
	case, we need to update the DP cell for $v$ and the DP solutions for $u$ and
	all of its ancestors. By Lemma~\ref{lem:tree-dp-dynamic},
	Lemma~\ref{lem:property-adp} and by our choice of $p=O(h\log W/\varepsilon)$, this
	can be done in time $O\big(h^3\cdot M^{2t}\log^2(W)
          \log(k/\varepsilon)\log(M^th
          \log(W)/\varepsilon)/\varepsilon^3\big)$.
          
	Next, consider the case in which $(u,v)$ is inserted and $v$ is the root of
	its component. Then we need to recompute the DP solutions for $v$ and all of
	its ancestors (including $u$) which, by Lemma~\ref{lem:tree-dp-dynamic},
	Lemma~\ref{lem:property-adp} and by our choice of $p=O(h\log W/\varepsilon)$, 
	can be done in the time claimed in the lemma. In the case that we need to
	re-root the component of $v$, note that we have to recompute the solutions
	for all ancestors of $v$. Since the height of $T$ is bounded by $h$, there
	are at most $h$ such ancestors. Furthermore, we have picked the order of
	edge deletions such that whenever we delete or insert an edge in the
	re-rooting process then we only need to recompute two DP cells. Hence, in
	total we only need to recompute the solutions for $O(h)$ DP cells in the
	re-rooting process and thus by Lemma~\ref{lem:property-adp}, the total time
	for this process is $O\big(h^3\cdot M^{2t}\log^2(W)
          \log(k/\varepsilon)\log(M^th
          \log(W)/\varepsilon)/\varepsilon^3\big)$.
\end{proof}

\textbf{Extension to Dynamic General Graphs and Non-Binary Trees.}
Now suppose that our input is a dynamic (general) graph $G$ that is undergoing edge
insertions and deletions. Essentially we will solve this problem by
maintaining a dynamic Räcke tree and running the algorithm from
Proposition~\ref{prop:dynamic-approx-dp} on top of it. However, the dynamic Räcke tree
from Theorem~\ref{thm:dynamic-racke-tree} is non-binary and, therefore, we
start by arguing that we can maintain a binarized Räcke tree dynamically in the
following lemma.

\begin{lemma}
\label{lem:dynamic-binary-racke-tree}
	Let $G=(V_G,E_G)$ be a dynamic unweighted graph with $n$ vertices
	that is undergoing edge insertions and deletions.  We can maintain a
	binarized Räcke tree~$T=(V_T,E_T,\capac_T)$ with $O(n^2)$ vertices, height
	$O(\log^{7/6} n)$ and quality $n^{o(1)}$ in amortized update time
	$n^{o(1)}$. The preprocessing time is $O(n^2)$.
\end{lemma}
\begin{proof}
	Let $T'=(V_{T'},E_{T'},\capac_{T'})$ be the fully dynamic Räcke tree for $G$
	from Theorem~\ref{thm:dynamic-racke-tree}. First, note that $T'$ has
	$n_{T'}:=O(n)$ vertices and height $O(\log^{1/6} n)$.  Second, note that $T'$
	can have unbounded degree. Therefore, similar to the proof of
	Lemma~\ref{lem:binary-racke-tree}, we will show how to maintain a binarized
	Räcke tree $T$ that has $T'$ as its corresponding (non-binarized) Räcke
	tree. We do so by taking $T'$ and replacing each vertex $u$ in $T'$ by a
	balanced binary tree $\tau_u$ with $n_{T'}$ leaves; the internal edges of
	$\tau_u$ will have weight $\infty$ and the edges connecting subtrees
	$\tau_u$ and $\tau_v$, $u\neq v$, in $T$ will correspond to the edges in
	$T'$ and will have the same (finite) weight as in~$T'$. We will see that by
	contracting all edges with weight $\infty$ in $T$, we will obtain again
	$T'$. We now elaborate on this process.

	During the preprocessing, we first build $T'$. Note that this takes time
	$O(n^2)$. Now we construct $T$ as follows. For each vertex $u\in V_{T'}$, we
	add a balanced rooted binary tree $\tau_u$ with $n_{T'}$ leaves. We refer to
	the root of $\tau_u$ as $r_u$.  We identify each leaf of $\tau_u$ with a
	vertex $v\in V_{T'}$ and denote the leaf of $\tau_u$ that corresponds to $v$
	by $c_{u,v}$. We set the weight of the edges inside $\tau_u$ to $\infty$. Note
	that $T$ has $O(n_{T'}^2) = O(n^2)$ vertices.  Next, for each edge $(u,v)\in
	E_{T'}$ (where we assume that $v$ is a child of $u$), we insert the edge
	$(c_{u,v},r_v)$ in $T$ and set $\capac_{T}(c_{u,v},r_u)=\capac_{T'}(u,v)$.
	Finally, if $u$ is the root of $T'$ then we set $r_u$ to the root of $T$.

	Next, suppose that $G$ is changed due to an edge insertion or deletion. Then
	we first update the tree $T'$ via the algorithm from
	Theorem~\ref{thm:dynamic-racke-tree}. Now, whenever an edge $(u,v)$ is
	inserted (deleted) in $T'$, we insert (delete) the edge $(c_{u,v},r_{v})$
	into (from) $T$. Each of these insertions and deletions in $T$ can be done
	in time $O(1)$. Since it takes amortized time $n^{o(1)}$ to update $T'$ (via
	Theorem~\ref{thm:dynamic-racke-tree}), the total update time is $n^{o(1)}$.

	It is left to show that $T$ is a binarized Räcke tree of height
	$O(\log^{7/6}n)$ and quality $n^{o(1)}$.
	Clearly, $T$ is a binary tree since all vertices inside each subtree
	$\tau_u$ have at most two child nodes and, additionally, each vertex
	$c_{u,v}$ has at most one child node (namely $r_v$). Next, $T$ has height
	$O(\lg^{7/6} n)$ since $T'$ has height $O(\lg^{1/6}n)$ and the subtrees
	$\tau_u$ have height $O(\lg n)$. Finally, let $T''$ be the tree obtained
	from $T$ by contracting all edges with weight $\infty$. We argue that
	$T'=T''$. Indeed, consider any vertex $u\in V_{T'}$ and its subtree $\tau_u$
	in $T$. Then after contracting the edges in $\tau_u$, we are left with a
	subtree that only contains $r_u$. Furthermore, all edges between vertices of
	different subtrees $\tau_u$ and $\tau_v$, $u\neq v$, have finite weight.
	Therefore, $T'=T''$. This implies that $T$ is binarized Räcke tree for $G$.
	Since $T'$ has quality $n^{o(1)}$, the quality of $T$ is also $n^{o(1)}$.
\end{proof}

Given the lemma above, our dynamic algorithm for dynamic general graphs $G$
works as follows. We maintain the dynamic binarized Räcke tree~$T$ as per
Lemma~\ref{lem:dynamic-binary-racke-tree} on our input graph $G$, i.e., whenever
an edge is inserted or deleted in $G$, we update the data structure from the
lemma as well. Note that this causes edge insertions and deletions in $T$ as
well. As before, we set the vertex weights in $T$ such that $w(v)=1$ if $v$
corresponds to a vertex in $G$ and $w(v)=0$ if $v$ is an internal node of $T$
that does not correspond to any vertex in $G$. Furthermore, we run our dynamic
algorithm from Proposition~\ref{prop:dynamic-approx-dp} for binary trees on $T$.
In particular, whenever $T$ gets updated, we also update the DP solution as
per Proposition~\ref{prop:dynamic-approx-dp}. We also use the same query
procedure as in the proposition.

We conclude the subsection by proving Theorem~\ref{thm:partitioning-dynamic}.

\begin{proof}[Proof of Theorem~\ref{thm:partitioning-dynamic}]
	We prove the result for general graphs first.  Since the dynamic binarized
	Räcke tree~$T$ that we maintain has quality $n^{o(1)}$, the same
	argumentation as in the proof of Theorem~\ref{thm:partitioning-static}
	implies that we maintain a bicriteria
	$(n^{o(1)},(1+\PTASeps)(1+\epsilon))$-approximation for $G$.  Since by
	Lemma~\ref{lem:dynamic-binary-racke-tree} we can maintain $T$ with amortized
	update time $n^{o(1)}$, the amortized number of edge insertions and
	deletions into $T$ is $n^{o(1)}$ per update operation. Since $T$ has height
	$O(\lg^{7/6}n)$ and by Proposition~\ref{prop:dynamic-approx-dp}, the total
	amortized update time $n^{o(1)}$.  This implies the claim about dynamic
	general graphs.

	To obtain our result for non-binary trees, we can proceed similar to above.
	Consider a non-binary $T'$ that is undergoing edge insertions and deletions.
	We can maintain the same data structure as in the proof of
	Lemma~\ref{lem:dynamic-binary-racke-tree} to obtain a binary tree~$T$ with
	$O(n^2)$ vertices with worst-case update time $O(1)$. Now we assign weight
	$w(r_v)=1$ to all vertices~$r_v$ that are roots of the subtrees $\tau_v$ in
	$T$ and weight $w(v)=0$ to all other vertices of the subtrees $\tau_v$.  By
	the same arguments as in the proof of Theorem~\ref{thm:partitioning-static},
	we obtain the claim.
\end{proof}

\section{Simultaneous Source Location}
\label{sec:simultaneous-source-location}
In this section, we provide efficient algorithms for the simultaneous source
location problem as studied by Andreev et al.~\cite{andreev09simultaneous}.
Recall that in
this problem, the input consists of a graph $G=(V,E,\capac,d)$ with a capacity
function $\capac \colon E \to \Winfty$ on the edges and a demand function
$d \colon V \to \Winfty$ on the vertices of the graph. The goal is to select a
minimum set $S\subseteq V$ of \emph{sources} that can simultaneously supply all
vertex demands. More concretely, a set of sources $S$ is \emph{feasible} if
there exists a flow from the vertices in $S$ that supplies demand $d(v)$ to all
vertices $v\in V$ and that does not violate the capacity constraints on the
edges. Here, we assume that each source vertex can potentially send an infinite
amount of flow that is only constrained by the edge capacities.
The objective is to find a feasible set of sources of minimum size.

Next, we summarize our main results for the simultaneous source location
problem. First, we introduce our notion of bicriteria approximation.  Let $S^*$
be the optimal solution for the simultaneous source location problem.  Then we
say that  a solution $S$ is a bicriteria $(\alpha,\beta)$-approximate solution if
$\abs{S}\leq\alpha\abs{S^*}$ and if $S$ is a feasible set of sources after all
edge capacities are increased by a factor $\beta$.

The following theorem summarizes our main result for static algorithms.
\sslstatic*

Next, we turn to our dynamic algorithms which support the following update
operations:
\begin{itemize}
\item \emph{SetDemand}($v$, $d$): updates the demand of vertex $v$ to $d(v)=d$,
\item \emph{SetCapacity}($(u,v)$, $c$): updates the capacity of the edge $(u,v)$ to
								  $\capac(u,v)=c$,
\item \emph{Remove}($u,v$): removes the edge $(u,v)$ from the graph,
\item \emph{Insert}($(u,v)$, $c$): inserts the edge $(u,v)$ into the graph with
								capacity $\capac(u,v)=c$.
\end{itemize}

The next theorem summarizes our main results for dynamic algorithms.
\ssldynamic*

We note that in our static and dynamic algorithms, we can output the
corresponding solutions similarly to what we descriped after
Proposition~\ref{prop:dynamic-knapsack-chan} for knapsack.

We start by presenting an exact DP for the special case of binary trees in
Section~\ref{sec:ssl-dp-trees} and then present an approximate DP in
Section~\ref{sec:ssl-dp-approx}. After that, we generalize the result from
binary trees to general graphs in Section~\ref{sec:ssl-general} and then also
to the fully dynamic setting in Section~\ref{sec:ssl-dynamic}.

\subsection{The Exact DP}
\label{sec:ssl-dp-trees}

We consider the special case of the simultaneous source location problem on
\emph{binary} trees and provide a DP that solves this problem exactly. We let
$T=(V,E,\capac,d)$ denote the rooted binary tree with root $r$ that we obtain as
input. Additionally, we assume that for each vertex $v\in V$ we obtain as input
whether we are allowed to make $v$ a source or not; note that this only
generalizes the problem (as in the original problem all vertices can be made
sources). Later in Section~\ref{sec:ssl-general}, this generalization will be
helpful when we apply Räcke trees because then we only want to allow leaves to
act as sources.

\subsubsection{DP Definition}
We now define our exact DP. We will also discuss its relationship with the DP by
Andreev el al.~\cite{andreev09simultaneous} and why we did not use the DP of
Andreev et al.
Given a vertex $v$ and a value $x\in\mathbb{R}$, we denote by $\DP(v,x)$ the
minimum number of sources to place in $T_v$ such that when $v$ receives
flow at most $x$ from its parent then all demands in $T_v$ can be satisfied.  We
note that $x$ can take positive and negative values: for $x\geq0$ this
corresponds to the setting in which flow is sent from the parent of $v$ into
$T_v$ and for $x<0$ this corresponds to the setting in which flow is sent from
$T_v$ towards the parent of $v$. We further follow the convention that when the
demands in $T_v$ cannot be satisfied when $v$ receives flow $x$ from its
parent, then we set $\DP(v,x)=\infty$.

Observe that this DP has rows $\I=V$ and columns $\J=\mathbb{R}$. We will store
the rows $\DP(v,\cdot)$ using our data structure from
Section~\ref{sec:approx-conv} using monotone piecewise constant functions. 
Next, we observe that each $\DP(v,\cdot)$ is monotonically decreasing.
Hence, the DP satisfies Property~(1) of Definition~\ref{def:well-behaved}.
\begin{observation}
\label{obs:ssl-monotone}
	The function $\DP(v,\cdot) \colon \mathbb{R} \to [n+1] \cup \{\infty \}$ is monotonically
	decreasing.
\end{observation}
\begin{proof}
	This follows immediately from the definition of $\DP(v,x)$: Consider
	$x,x'\in\mathbb{R}$ with $x\leq x'$. Then any solution in which $T_v$
	receives flow at most $x$ from the parent of $v$ is also feasible when $T_v$
	receives flow at most $x'$ from the parent of $v$. Therefore,
	$\DP(v,x)\geq\DP(v,x')$, which finishes the proof.
\end{proof}

Observe that the global solution for the simultaneous source location problem on
$T$ can be obtained by evaluating $\DP(r,0)$, where $r$ is the root of $T$:
First, $r$ has no parent and, therefore, it must be a source itself or have its
demand satisfied by its children; this explains the choice of $x=0$.
Furthermore, (by definition) $\DP(r,0)$ is the minimum number of sources that we
need to satisfy all demands in $T_r = T$ and, thus, the flow that we obtain is
feasible. We conclude that $\DP(r,0)$ gives the global optimum solution.

\textbf{Relationship to the approach by Andreev et al.~\cite{andreev09simultaneous}.}
Next, let us elaborate on the relationship of our DP and the function $f$ used by Andreev
et al.~\cite{andreev09simultaneous}.
In~\cite{andreev09simultaneous}, the function $f$ computed by a dynamic program  is defined as follows. Given a
vertex $v$ and an integer $i\in\mathbb{N}$,  Andreev et al.~define a function $f(v,i)$ that
denotes the \emph{minimum amount of flow} that $v$ needs to receive from its
parent if all demands in $T_v$ need to be satisfied and if we can place
$i$~sources in the subtree $T_v$.  Similar to above, $f(v,i)$ can take positive
and negative values: if the demand in $T_v$ can only be satisfied by receiving
flow from the parent, then $f(v,i)$ is positive; if the demand in $T_v$ is
already satisfied by the sources in the subtree~$T_v$, then it is possible that
$v$ can send flow to its parent and $f(v,i)$ is negative.  It is not hard to see
that the function $f(v,i)$ is monotonically decreasing in $i$.\footnote{This
	follows immediately from the definition of $f$ and the fact that by adding
	more sources to a subtree~$T_v$, the amount of flow that $T_v$ needs to
	receive from the parent of $v$ only decreases.}

Now consider $f(v,\cdot) \colon \mathbb{N} \to \mathbb{R}$ as a function and
consider its ``inverse''\footnote{We note that, formally, $f(v,\cdot)$ has no
	inverse since it is possible that multiple values map to the same number,
	i.e., $f(v,i)=f(v,i')$ for $i\neq i'$. Thus, formally, we set
	$f^{-1}(v,x) = \min\{ i \colon f(v,i)\leq x \}$, where we follow the
	convention $\min\{\emptyset\} = \infty$. Then we interpret $f^{-1}(v,\cdot)$
	as a piecewise constant function from $\mathbb{R}$ to $[n+1]$.}
function $f^{-1}(v,\cdot) \colon \mathbb{R} \to \mathbb{N}$, where $f^{-1}$ is
defined on the whole set of real numbers (including negative numbers). That is,
$f^{-1}(v,x)$ denotes the minimum number of sources that we need to place in
$T_v$ such that the demand that $v$ requires from its parent is at most $x$.
But this was exactly the definition of $\DP(v,x)$. Thus, $\DP(v,x)=f^{-1}(v,x)$
for all $v$ and $x$.  

\emph{Why Did We Not Use $f$?}
In~\cite{andreev09simultaneous} it is shown how the function $f$ can be
computed in polynomial time by a bottom-up dynamic program using just a few case
distinctions and a $(\min,+)$-convolution in each DP cell. Thus, one might
wonder why we picked $\DP(v,\cdot)=f^{-1}(v,\cdot)$ and not $f$ for our DP?
Indeed, it seems quite natural to interpret the function $f$ as a monotone
piecewise constant function and to use it for our dynamic program. While for the case of
exact computations this is possible, we now sketch why this appears unhandy for
the approximate case later.

Suppose that we used the function $f$ in our approximate computations.  To
obtain efficient approximation algorithms, we will have to ensure that $f$ has
only few pieces and our main way to achieve this is by rounding~$f$ as per
Lemma~\ref{lem:operations}. However, this becomes tricky because the function
values of $f$ are positive \emph{and negative}. In the following, it will be
illustrative to think of positive function values for $f$ as vertex demands that
need to be satisfied and of negative values for $f$ as available edge
capacities.
The main issue is that
since the function values of $f$ are positive and negative, it is not clear
how we should perform the rounding: if we rounded positive and negative
values up (towards $+\infty$) then this would correspond to \emph{increasing}
the vertex demands while at the same time \emph{decreasing} the edge capacities;
however, this could render some feasible solutions (in the exact computation)
infeasible (in the rounded computation). On the other hand, it is conceivable
that by always rounding $f$ \emph{down} (towards $-\infty$), we would
essentially decrease the vertex demands while increasing the edge capacities.
Potentially, this approach could work when we are allowed to violate the edge
capacities by a $(1+\varepsilon)$-factor.  However, even if we did that, we
would have another issue: to only use a small number of pieces for representing
$f$, we would have to use different rounding mechanisms for those function
values in $[-1,1]$ and those in $[-W,W]\setminus[-1,1]$, where $W$ is the
largest edge capacity.  Indeed, if we rounded the values of $f$ to powers of
$(1+\delta)^j$ then there are only $O(\lg_{1+\delta}(W))$ function values in
$[-W,W]\setminus[-1,1]$ but there are infinitely many function values in
$[-1,1]$.  Similarly, if we rounded to multiples of $\delta$ then there are only
$O(1/\delta)$ function values in $[-1,1]$ but this would lead to $O(W/\delta)$
function values in $[-W,W]\setminus[-1,1]$. In both cases, our functions would
have too many pieces and, thus, one would have to pick a rounding function which
provides a tradeoff between these two cases.  Furthermore, we would have to find
an analysis that shows that this ``more involved'' rounding function does not
introduce too much error.

Note that in the above discussion, all of the issues come from the fact that
$f(v,\cdot)$ can also take negative values. On the other hand, our DP
(which is $f^{-1}(v,\cdot)$) only takes non-negative function values
and, therefore, we avoid all of the above complications because we can use the
standard rounding function $\lceil \cdot \rceil_{1+\delta}$ that rounds to
powers of $1+\delta$. Thus, we bypass all of the issues above.

Our approach also has the positive side effects that instead of getting factors of
$\polylog(W)$ in our running times, we only get factors of $\polylog(n)$ because
the codomain of our monotone piecewise constant functions became $[n+1]$ rather
than some potentially large interval $[-W,W]$.

\subsubsection{Computing the DP}
Now we describe the exact computation of our DP. This will reveal the
procedures $\Procedure_i$ from Definition~\ref{def:well-behaved}.

Let $v\in V$ be any vertex in $T$.  We describe how to compute
$\DP(v,\cdot)$ efficiently assuming that we have already computed the solutions
for the children of $v$ (if they exist). Recall that for each vertex~$v$ we also
obtain as input, whether $v$ can be used as a source or not. In our following
case distinctions, whenever we consider the case that $v$ is used as a source, we
will implicitly condition on the fact that it is also possible to use $v$ as
source; if $v$ cannot be used as a source, we simply skip this case.

In the construction for each DP cell $\DP(v,\cdot)$ for a vertex $v$ with parent
$p$, we will additionally ensure that we do not violate the capacity of the edge
$(p,v)$ when $x$ is very small or very large. More concretely, we will ensure
that $\DP(v,\cdot)$ satisfies the additional property that $\DP(v,x)=\infty$ for
$x<-\capac(p,v)$ and $\DP(v,x)=\DP(v,\capac(p,v))$ for all $x>\capac(p,v)$.
We will denote this property as the \emph{feasible capacity property}.

\bold{Case 1: $v$ is a leaf.}
Suppose that $v$ is a leaf. We initialize $\DP(v,\cdot)$ as the function which
takes value $\infty$ on all of $\mathbb{R}$. In the following, we add at most
two pieces to $\DP(v,\cdot)$ depending on whether $v$ can be used as a source or
not.

First, suppose $v$ can be used as a source. Then we can send flow up to
$\capac(p,v)$ to the parent $p$ of~$v$.  Furthermore, since $v$ is a leaf, there
is exactly one source in $T_v$. Thus, we update $\DP(v,\cdot)$ and set
$\DP(v,x)=1$ for all $x\geq-\capac(p,v)$. This adds one piece to $\DP(v,\cdot)$. 

Second, suppose  $v$ is not a source. Then if $x\geq d(v)$ and $\capac(p,v)\geq
d(v)$, $v$ can receive all of its demand~$d(v)$ from its parent and the flow is
feasible because we do not exceed the capacity of the edge $(p,v)$.  Therefore,
if $\capac(p,v)\geq d(v)$, then we update $\DP(v,\cdot)$ again and add the piece
with $\DP(v,x)=0$ for all $x\geq d(v)$.  If $x\geq d(v)$ but $d(v)>\capac(p,v)$
then we do nothing because the parent of $v$ cannot satisfy the demand of $v$.

Observe that $\DP(v,\cdot)$ is a monotonically decreasing function with at most
three pieces. Furthermore, it clearly satisfies the feasible capacity property
and Property~(3) of Definition~\ref{def:well-behaved}.

\bold{Case 2: $v$ is not a leaf.}
Suppose that $v$ is not a leaf and that $v$ has children $v_1$ and $v_2$, as
well as a parent $p$. Recall that we assume that we have already computed the DP
entries $\DP(v_1,\cdot)$ and $\DP(v_2,\cdot)$ for both children of $v$. We now
show how to compute two DP solutions $\DP_A(v,\cdot)$ and $\DP_B(v,\cdot)$
depending on whether $v$ is a source (in Case~A) or not (in Case~B). Then, if
$v$ can be used as a source, we set
$$\DP(v,\cdot)=\min\{\DP_A(v,\cdot), \, \DP_B(v,\cdot) \},$$
where we compute the $\min$-operation via Lemma~\ref{lem:operations}. If $v$
cannot be used as a source, we set $\DP(v,\cdot)=\DP_B(v,\cdot)$.

\emph{Case~A:} Suppose $v$ is used as a source. We initialize
$\DP_A(v,\cdot)$ as the function which takes value $\infty$ on all of
$\mathbb{R}$. Now, since $v$ can be used as a source, $v$ can send flow
$\capac(p,v)$ to its parent and flow $\capac(v,v_1)$ and $\capac(v,v_2)$ to its
children.  Therefore, for $x\geq-\capac(p,v)$, the number of sources in
$\DP_A(v,x)$ is 1 (since $v$ is a source) \emph{plus} the number of sources that
we require in $T_{v_1}$ when $v_1$ can receive flow $\capac(v,v_1)$ from its
parent~$v$ \emph{plus} the same quantity for $T_{v_2}$.  Thus, it suffices to
set
$$\DP_A(v,x)=1+\DP(v_1,\capac(v,v_1))+\DP(v_2,\capac(v,v_2))$$
for all
$x\geq-\capac(p,v)$. Note that here we exploited that the functions
$\DP(v_1,\cdot)$ and $\DP(v_2,\cdot)$ are monotonically decreasing and that both
of them satisfy the feasible capacity property.  We conclude that in this case
$\DP_A(v,\cdot)$ is a monotonically decreasing function with two pieces.

\emph{Case~B:} Suppose that $v$ is not used as a source. We initialize $\DP_B(v,\cdot)$ as the
function which takes value $\infty$ on all of~$\mathbb{R}$. To compute the value
of $\DP_B(v,x)$, we need to obtain the minimum number of sources such that $v$
receives flow at most $x$ from its parent and such that all demands in $T_v$ are
satisfied.  Since $v$ is not a source, its demand~$d(v)$ must be satisfied
either by its parent~$p$ or by its children (or a combination of them).
Therefore, to obtain that we have to pick the children solutions $\DP(v_1,x_1)$
and $\DP(v_2,x_2)$ such that $d(v)\leq x-x_1-x_2$. 

Since we did not make $v$ a source, the number of sources in $\DP_B(v,x)$ is
the number of sources that we need to place in the subtrees $T_{v_1}$ and
$T_{v_2}$. Thus, we get
$$\DP_B(v,x)=\min_{x_1\in\mathbb{R}} \{ \DP(v_1,x_1)+\DP(v_2,x-x_1-d(v)) \},$$
where we used that $x_2\leq x-x_1-d(v)$ and by monotonicity of $\DP(v_2,\cdot)$
we minimize the number of sources in $T_{v_2}$ if we consider $x_2=x-x_1-d(v)$.
Here, the flows that we computed for $\DP_B(v,x)$ are set feasible because the
solutions $\DP(v_i,\cdot)$ satisfy the feasible capacity property and therefore
we do not violate the edge constraints to the children.

Since the above equality holds for all values of~$x$, $\DP_B(v,\cdot)$
corresponds to a shifted $(\min,+)$-convolution of two monotonically decreasing
functions.  More concretely, via Lemma~\ref{lem:operations} we can first compute
the shifted function $\DP(v_2,\cdot-d(v))$ and then we can set
$$\DP_B(v,\cdot)=\DP(v_1,\cdot)\minconv\DP(v_2,\cdot-d(v)),$$
which we compute via Lemma~\ref{lem:convolution}.

Finally, as a postprocessing step, we set
$\DP(v,\cdot)=\min\{\DP_A(v,\cdot),\,\DP_B(v,\cdot)\}$ if $v$ can be used as a source
and $\DP(v,\cdot)=\DP_B(v,\cdot)$ otherwise, as we already mentioned
above. But we also need to ensure that $\DP(v,\cdot)$ satisfies the feasible
capacity property.  Therefore, we set $\DP(v,x)=\infty$ for $x<-\capac(p,v)$ and
we set $\DP(v,\cdot)=\DP(v,\capac(p,v))$ for $x>\capac(p,v)$.  Observe that
these changes to $\DP(v,\cdot)$ can be done in time linear in the number of
pieces of $\DP(v,\cdot)$.

\textbf{Properties of the DP.}
Observe that in the DP above for each vertex~$v$ we only required the DP
solutions for its children $v_1$ and $v_2$.  Hence, our dependency graph is
given by our input tree~$T$ where all edges are directed towards the root. This
implies that every node in the dependency graph can only reach those nodes on a
path to the root and thus Property~(2) of Definition~\ref{def:well-behaved} is
satisfied with~$h$ being the height of~$T$.  Additionally, one can verify that
above all operations also satisfy Property~(3) of
Definition~\ref{def:well-behaved}. Finally, observe that in each step we only
used a constant number of operations from Lemma~\ref{lem:operations} and at most
one $(\min,+)$-convolution from Lemma~\ref{lem:convolution}.

\subsection{The Approximate DP}
\label{sec:ssl-dp-approx}
Now we explain how we solve the above DP more efficiently by computing
approximate solutions $\ADP(v,\cdot)$. This will reveal the procedures
$\tilde{\Procedure}_i$ from Definition~\ref{def:well-behaved}.

In our approximation algorithm, we do everything exactly as above except that we
replace each exact solution $\DP(v,\cdot)$ with the approximate solution
$\ADP(v,\cdot)$. Then we add a postprocessing step in which we round
$\ADP(v,\cdot)$, i.e., we set
\begin{align}
\label{eq:ssl-rounding}
	\ADP(v,\cdot)=\lceil \ADP(v,\cdot) \rceil_{1+\delta}
\end{align}
for a parameter $\delta>0$ that we will set later.

Note that all of our operations are exact except the rounding step which loses a
factor of $\alpha=1+\delta$. Thus, Property~(4a) of
Definition~\ref{def:well-behaved} is satisfied.
Additionally, observe that in each step we only used a constant number of
operations from Lemma~\ref{lem:operations} and at most one
$(\min,+)$-convolution from Lemma~\ref{lem:convolution}. This implies that
Property~(4b) is satisfied.  Furthermore, all functions we consider are monotone
and our rounding step ensures that each row $\ADP(v,\cdot)$ has at most
$p=O(\lg_{1+\delta}n)$ pieces. Hence, Property~(4c) is also satisfied.

This implies that the DP is $(h,1+\delta,O(\log_{1+\delta}(n)))$-well-behaved.
By applying Theorem~\ref{thm:well-behaved-static} with
$\delta=\ln(1+\varepsilon)/(h+1)$, we obtain the following proposition
which shows that on binary trees, the approximation algorithm computes a
bicriteria $(1+\varepsilon,1)$-approximate solution. We note that for constant
$\varepsilon$, the running time essentially becomes
$\tO(n\cdot h^2)$, where $h$ is the height of the tree. Thus, for trees of
height $\tO(1)$, we obtain a near-linear running time.
\begin{proposition}
\label{prop:ssl-binary-approximate}
	Let $\varepsilon>0$. The approximation algorithm computes a bicriteria
	$(1+\varepsilon,1)$-approximate solution for the simultaneous source
	location problem on binary trees in time 
	$O( n\cdot (h \log(n)/\varepsilon)^2 \log( h \log(n)/\varepsilon) )$, where $h$ is
	the height of the tree.
\end{proposition}

\subsection{Extension to General Graphs (Proof of Theorem~\ref{thm:ssl-static})}
\label{sec:ssl-general}
We prove Theorem~\ref{thm:ssl-static} by giving reductions to the binary
setting.

First, suppose that $G$ is a tree with potentially unbounded degree. Then we
turn $G$ into a binary tree~$T$ using the same construction as in the proof of
Lemma~\ref{lem:binary-racke-tree}. That is, we replace each vertex $u$ in $G$ by
a balanced binary tree $\tau_u$ with $\deg(u)$ leaves
$c_{u,v_1},\dots,c_{u,v_{\deg(u)}}$, where the $v_i$ are the children of $u$ in
$G$; the internal edges of $\tau_u$ have capacity $\infty$ and we denote the root
of each $\tau_u$ by $r_u$. Furthermore, for each
edge $(u,v)$ in $G$, we insert the edge $(c_{u,v},r_v)$ into $T$ with capacity
$\capac(c_{u,v},r_v) = \capac(u,v)$. By the same arguments as in the proof of
Lemma~\ref{lem:binary-racke-tree}, $T$ has $O(n)$ vertices and height
$O(h \log n)$, where $h$ is the height of $G$.  It is straight-forward to see
that with this construction, there exists a flow from $u$ to $v$ in $G$
\emph{if and only if} there exists a flow from $r_u$ to $r_v$ in $T$.  Now, in
$T$ we have already set the edge capacities and it remains to set the vertex
demands.  For each vertex $r_u$ in $T$, we set $d(r_u)=d(u)$, and for all other
vertices~$v$ in $T$, we set $d(v)=0$. Furthermore, in our instance of the
simultaneous source location problem we set that each vertex $r_u$ can be picked
as a source and none of the other vertices in $T$ can be picked as a source.
Note that there exists a one-to-one correspondence between sources in $G$ and
sources in $T$.  Together with our observation for flows above, this means that
solving the simultaneous source location problem on $T$ gives a solution for
$G$.

To obtain the bicriteria $(1+\varepsilon,1)$-approximation result for trees, we
apply the approximation algorithm from
Proposition~\ref{prop:ssl-binary-approximate} on $T$.

Finally, to obtain the $(1+\varepsilon,O(\lg^4 n))$-approximate solution for a general graph $G$, we
proceed as follows. We build the binarized Räcke tree~$T$ for $G$ as per
Lemma~\ref{lem:binary-racke-tree} and recall that $T$ has quality $q=O(\lg^4 n)$
and height $\tO(1)$.  In $T$, we set the bits to indicate that all leaves can be used as sources but
none of the other vertices might be used as a source. We apply the approximation
algorithm from Proposition~\ref{prop:ssl-binary-approximate} on $T$ to obtain a
$(1+\varepsilon,1)$-approximate solution on $T$ in time $\tO(n)$. Now let us
point out that the Räcke tree from Theorem~\ref{thm:racke-tree} (and, therefore,
also the binarized Räcke tree from Lemma~\ref{lem:binary-racke-tree}) is also a
\emph{tree flow sparsifier}. That is, if there exists a feasible flow $F$ in $G$,
then there exists a flow of the same value between the corresponding leaves in
$T$. Additionally, for any feasible flow~$F$ with value $v$ between leaves in
$T$, there exists a feasible flow with value $\frac{1}{q} v$ between the
corresponding vertices in $G$.  Therefore, if we are allowed to exceed the edge
capacities in $G$ by a factor of $q=O(\lg^4 n)$, the flow that we compute in $T$
is feasible in $G$. This gives that we can compute a
$(1+\varepsilon,O(\lg^4 n))$-approximate solution in time $\tO(n)$.

\subsection{Extension to the Dynamic Setting (Proof of Theorem~\ref{thm:ssl-dynamic})}
\label{sec:ssl-dynamic}
To prove Theorem~\ref{thm:ssl-dynamic}, we first consider the special case of
dynamic \emph{binary} trees (which is not mentioned in the theorem). We show
that for binary trees we can maintain a bicriteria
$(1+\varepsilon,1)$-approximate solution with worst-case update time
$\tO(h^3/\varepsilon^2)$,
where $h$ is an upper bound on the height of the tree. Then we show that the
results of the theorem can be derived from this result.

Consider a dynamic \emph{binary} tree on which we maintain the approximate DP
from Section~\ref{sec:ssl-dp-approx}. We will exploit that $T$ and the
dependency tree of our DP coincide. Hence, an update in $T$ will trigger the
same update in the dependency tree.
Observe that the update operation
\emph{SetDemand}($v$,~$d$) triggers a change to $\ADP(v,\cdot)$. Then we can
recompute the global approximate DP table using
Theorem~\ref{thm:well-behaved-dynamic}.  Since the DP is well-behaved, the tree
has height at most $h$ and since we set $\delta=O(h/\varepsilon)$, the
theorem implies that we need time $\tO(h^3/\varepsilon^2)$ to recompute the ADP
solution.  Similarly, for \emph{SetCapacity}($(u,v)$,~$c$) we
can again update the rows $\ADP(u,\cdot)$ and $\ADP(v,\cdot)$ and we update the
entire DP table using Theorem~\ref{thm:well-behaved-dynamic}. By the same
arguments as above, this takes time $\tO(h^3/\varepsilon^2)$. For
\emph{Remove}($u,v$), we remove the edge $(u,v)$ from $T$ and by the same
reasoning as before we get update time $\tO(h^3/\varepsilon^2)$. Finally,
consider \emph{Insert}($(u,v)$,~$c$), where we assume that $v$ becomes the child
of $u$.  Then we might have the issue that before the update, $v$ is not the
root of its connected component.  To mitigate this issue, we run the same
re-rooting procedure as described in Section~\ref{sec:extension-dynamic}. As
described in Section~\ref{sec:extension-dynamic}, this will only recompute the
solutions of $O(h)$ DP cells and thus we again have a total update time of
$\tO(h^3/\varepsilon^2)$.

Next, we prove the results from Theorem~\ref{thm:ssl-dynamic}. First, consider
the case in which all edge capacities are set to~$1$ and where we want to
obtain a bicriteria $(1+\varepsilon,n^{o(1)})$-approximate solution with
amortized update time $n^{o(1)}/\varepsilon^2$ and preprocessing time $O(n^2)$.  Let $G$ be
the dynamic input graph. We maintain the dynamic binarized Räcke tree~$T$ for
$G$ as per Lemma~\ref{lem:dynamic-binary-racke-tree} and remark that the
dynamic Räcke tree from Theorem~\ref{thm:dynamic-racke-tree} is also a tree flow
sparsifier. We note that any update to
$G$ triggers an update operation on $T$ that requires amortized update time
$n^{o(1)}$. On $T$, we allow the leaves to act as sources but no other
vertices. Furthermore, we set the demands of the leaves
in $T$ to the demands of the corresponding vertices in $G$; all other vertices
have demand $0$. Now we use the data structure for binary trees from the
previous paragraph to maintain a dynamic bicriteria
$(1+\varepsilon,1)$-approximate solution on $T$. That is, when a vertex demand
changes in $G$, we update the corresponding vertex demand in $T$. When an edge
is inserted or deleted in $T$ due to the subroutine from
Lemma~\ref{lem:dynamic-binary-racke-tree}, then we update the data structure
from the previous paragraph that maintains the DP solution on $T$. By the same
argumentation as in Section~\ref{sec:ssl-general}, we obtain that since $T$ has
quality $n^{o(1)}$, if we can exceed the edge capacities in $G$ by a $n^{o(1)}$
factor then any feasible flow in $T$ is also feasible in $G$. This implies the
result claimed in the theorem.

If $G$ is a tree of height $h$ but (potentially) with unbounded degrees, we can
maintain a bicriteria $(1+\varepsilon,1)$-approximate solution with worst-case
update time $\tO(h^3/\varepsilon^2)$ and preprocessing time
$O(n^2/\varepsilon^2)$ similar to above. That is, we transform $G$ into a binary
tree using the same procedure that we use in the proof of
Lemma~\ref{lem:dynamic-binary-racke-tree}, where we replace each vertex $u$ by a
subtree $\tau_u$ with root $r_u$. Similar to what we argued in
Section~\ref{sec:ssl-general}, we only allow the vertices $r_u$ as roots in $T$
and obtain any flow in $T$ corresponds to a flow in $G$. Then by applying the
dynamic data structure for binary trees on $T$, we obtain the result.

Finally, let us consider the case in which we wish to obtain bicriteria
approximation algorithms when we only allow the update operations
\emph{SetDemand}($v$,~$d$). In this case, we observe that
the underlying graph is static, since only the vertex
demands change. Therefore, for our input graph $G$, we can build the
Räcke tree from Theorem~\ref{thm:racke-tree} which is also a tree flow
sparsifier and we consider its binarized
version~$T$ as per Lemma~\ref{lem:dynamic-binary-racke-tree}. Note that building
this tree with quality $\tO(\lg^4 n)$ takes time $\tO(m)$. Given such a static Räcke tree,
we can use our dynamic data structure for binary trees from above to support the
operations \emph{SetDemand}($v$,~$d$) on
$T$. Since $T$ has height $\tO(1)$, we obtain the result with the bicriteria
$(1+\varepsilon,O(\lg^4 n))$-approximation. To obtain the bicriteria
$(1+\varepsilon,O(\lg^2 n\lg \lg n))$-approximation we do exactly the same as
above, but instead of using the Räcke tree from Theorem~\ref{thm:racke-tree}, we
use the Räcke tree from Harrelson, Hildrum and Rao~\cite{harrelson03polynomial}
which can be used as a tree flow sparsifier. As it 
has quality $O(\lg^2 n \lg \lg n)$ and can be constructed in time
$\poly(n)$, we obtain the result.

\section{Recourse Bounds}
\label{sec:recourse}
In this section discuss the recourse bounds we derive. To motivate these lower
bounds, let us note that ``classic'' dynamic algorithms with polylogarithmic
update time maintain a single explicit solution in memory; this is desirable in
many practical scenarios. However, some dynamic algorithms (like our DP
algorithms above) only return the \emph{value} of an approximate solution in
polylogarithmic time, which is sometimes referred to as \emph{implicit}.  To
understand whether for our problems implicit solutions are necessary, we
consider algorithms which maintain multiple explicit solutions, of which only
one has to be feasible. We believe that this is an interesting setting to look
at, as it essentially interpolates between the two scenarios above.  If even
algorithms with multiple solutions must have high recourse, this suggests that
implicit solutions are somehow inevitable. We show below that for fully dynamic
knapsack and fully dynamic $k$-balanced partitioning the latter is the case.

Here, we consider dynamic algorithms over inputs that are undergoing insertions
and deletions via an update operation. The algorithms are allowed to maintain
multiple explicit solutions and must ensure that after every time step, there
exists a solution with certain guarantees while minimizing the recourse for
updating the solutions.  More concretely, we consider algorithms which
explicitly maintain $s$~solutions $S_1^{(t)},\dots,S_s^{(t)}$ for each time
step~$t$. Here, we assume that after each time step a single update operation is
performed, after which an algorithm can make changes to its solutions.  We say
that an algorithm maintains an \emph{$\alpha$-approximate solution} if for each
time step~$t$, there exists an index $i=i(t)$ such that $S_{i}^{(t)}$ is a
feasible and $\alpha$-approximate solution for the problem we study. 

Observe that in this setting, the algorithm might have much lower recourse,
since for each time~$t$ it may pick a different solution. Thus, it may not have
to update any of the solutions significantly after the update operations.
Further note that this notion of ensuring that at each time step there exists a
feasible solution is somewhat reminiscent of list decoding in coding theory,
where the decoder can output a list of messages and only has to ensure that the
correct messages is contained in that list.

Measuring the recourse will be problem-specific, based on how the solutions for
the problems are stored. In general, given two solutions from consecutive time
steps, we let $d(S_i^{(t)},S_i^{(t+1)})$ denote the (problem-specific)
\emph{recourse} incurred by
the $i$'th solution at time step~$t$ (see below for how to set $d(\cdot,\cdot)$
for the problems we study). The \emph{total recourse} of an algorithm is given
by $$\sum_t \sum_{i=1}^s d(S_i^{(t)},S_i^{(t+1)}).$$

Next, we will present the concrete recourse lower bounds that we derive.

\textbf{Recourse Bounds for Knapsack.}
In knapsack, the solutions $S_i^{(t)}$ simply correspond to subsets of items
which are contained in the knapsack. To measure the recourse, we set 
$d(S_i^{(t)},S_i^{(t+1)}) = \abs{S_i^{(t)} \triangle S_i^{(t+1)}}$, i.e., 
we consider the cardinality of the symmetric difference of the $i$'th solution
at time steps~$t$ and $t+1$.

Our main result shows that for a fixed accuracy~$\varepsilon$, any dynamic
$(1-\varepsilon)$-approximation algorithm must maintain
$\Omega(1/\varepsilon)$~solutions or it must have recourse
$\Omega(\frac{n}{\varepsilon})$, even when only a single item is
inserted.

\begin{theorem}
\label{thm:recourse-knapsack}
	Let $\varepsilon\in(0,1/2)$.
	Assume $s < \frac{1}{8\varepsilon(1+2\varepsilon)}$ and
	$n\in\mathbb{N}$ is a sufficiently large multiple of~$s$.
	Then any dynamic randomized $(1-\varepsilon)$-approximation algorithm
	for knapsack
	with $s$~solutions must have recourse~$\Omega(\frac{n}{s})$.
	This holds even for a single item insertion.
\end{theorem}

\textbf{Recourse Bounds for $k$-Balanced Partitioning.}
In $k$-balanced partitioning, each solution $S_i^{(t)}$ consists of $k$~clusters
$V^{(i,t)}_1,\dots,V^{(i,t)}_k$ that partition the set of vertices~$V$. To
measure the recourse, we set 
$d(S_i^{(t)},S_i^{(t+1)}) = \sum_{j=1}^k \abs{V^{(i,t)}_j \triangle V^{(i,t+1)}_j}$,
i.e., we consider the total number of vertices that change their set
$V^{(i,\cdot)}_j$ from time~$t$ to $t+1$.

Our main result shows that for any~$C$ and fixed~$\varepsilon$, any algorithm
that maintains a $(C,1+\varepsilon)$-approximate solution must use
$\Omega(1/\varepsilon)$~solutions or it must have \emph{amortized} recourse
$\Omega(\varepsilon^2 \frac{n}{k})$, even when only $O(1/\varepsilon)$~edges are
inserted. Here, the amortized recourse refers to the total recourse divided by
the total number of update operations.
\begin{theorem}
\label{thm:recourse-partitioning}
	Let $C>0$ be arbitrary and $\varepsilon\in(0,1/2)$.
	Assume $k\geq 4$ and $s < \frac{1}{4\varepsilon}$.
	Then any dynamic randomized $(C,1+\varepsilon)$-approximation algorithm for
	$k$-balanced partitioning with $s$~solutions must have amortized
	recourse~$\Omega(\varepsilon^2 \frac{n}{k})$.
	This holds even for $O(1/\varepsilon)$~edge insertions.
\end{theorem}

\subsection{Proof of Theorem~\ref{thm:recourse-knapsack}}
\label{sec:proof-recourse-knapsack}
We prove Theorem~\ref{thm:recourse-knapsack}.
We use Yao's principle~\cite{yao77probabilistic}, i.e., we consider a
deterministic algorithm and give a distribution over inputs, showing that in
expectation the algorithm will have recourse $\Omega(\frac{n}{s})$.

We consider an instance in which initially we have $n$~items, and each item~$i$
has weight $w_i=1$ and price $p_i=1$. We refer to these items as \emph{small}
items. We set the budget of our knapsack to $B=n$. Note that in this instance,
$\OPT=n$ because all small items fit into the knapsack.

Now we sample an integer~$j$ uniformly at random from $[2s-1]=\{0,\dots,2s-1\}$,
and we set $k = i\cdot\frac{n}{2s} + \frac{n}{4s}$. We insert a single
\emph{heavy} item with $p = n-k+2\varepsilon n$ and $w=n-k$. Note that after
inserting the heavy item, we have that $\OPT = n+2\varepsilon n$ since the
optimal solution consists of the heavy item and $k$~small items.

We let $S_1,\dots,S_s$ denote the solutions maintained by the algorithm
\emph{before} the heavy item was inserted and we let $S_1',\dots,S_s'$ denote
the solutions \emph{after} the heavy item was inserted. We write
$\numSmall(S_i)$ to denote the number of small items in solution~$S_i$. 

In the following, we will show that any $(1-\varepsilon)$-approximate
solution~$S_i'$ must contain the heavy item and ``almost'' $k$~small items.
However, we will also show that with constant probability all solutions~$S_i$
had ``much less'' or ``much more'' than $k$~small items initially. This then
gives that obtaining any $(1-\varepsilon)$-approximate solution must encur high
recourse.

We follow this proof strategy in reverse order. We start by showing that with
constant probability, all $S_i$ have ``much less'' or ''much more'' than
$k$~small items.
\begin{lemma}
\label{lem:recourse-knapsack}
	With probability at least $1/2$, it holds that
	$\abs{k-\numSmall(S_i)} \geq \frac{n}{4s}$ for all $i\in[s]$.
\end{lemma}
\begin{proof}
	Suppose that we partition the set $\{1,\dots,n\}$ into $2s$~consecutive
	intervals, each of length $\frac{n}{2s}$. Note that $k$ is the middle point
	of one of these intervals. Furthermore, as there are only $s$~solutions and
	$2s$~intervals, at least half of the intervals do not contain a number from
	the set $\{ \numSmall(S_i) \colon i=1,\dots,s \}$; we call these intervals
	\emph{empty}. Thus, with probability at least $1/2$, $k$ is the middle point
	of an empty interval. If the interval containing $k$ is empty, then $k$ has
	distance at least $\frac{n}{4s}$ to $\numSmall(S_i)$ for all $i=1,\dots,s$.
\end{proof}

Next, recall that $S_1',\dots,S_n'$ are the solutions maintained by the
algorithm after the insertion of the heavy item. We show that any
$(1-\varepsilon)$-approximate solution must contain the heavy item and some
small items.
\begin{lemma}
\label{lem:recourse-knapsack-both}
	Suppose that $S_i'$ is a $(1-\varepsilon)$-approximate solution. Then $S_i'$
	contains the heavy item and a positive number of small items.
\end{lemma}
\begin{proof}
	First, suppose that a solution~$S_i'$ only contains small items. Then its total
	price is at most~$n$. However, we have that
	\begin{align*}
		(1-\varepsilon) \OPT
		= (1-\varepsilon) (1+2\varepsilon) n
		> n,
	\end{align*}
	where we used that $\varepsilon<1/2$. Hence, $S_i'$ is not
	a $(1-\varepsilon)$-approximate solution.

	Second, suppose that $S_i'$ only contains the heavy item. Then we have that 
	\begin{align*}
		(1-\varepsilon) \OPT
		&= (1-\varepsilon) (p+k) \\
		&= p + k - \varepsilon(p+k) \\
		&\geq p + \frac{n}{4s} - \varepsilon(1+2\varepsilon)n \\
		&> p + 2\varepsilon(1+2\varepsilon)n - \varepsilon(1+2\varepsilon)n \\
		&> p,
	\end{align*}
	where we used that $k\geq \frac{n}{4s}$, $p+k=n+2\varepsilon n$ and
	$s < \frac{1}{8\varepsilon(1+2\varepsilon)}$.
	Therefore, a solution containing only the heavy item is not
	$(1-\varepsilon)$-approximate.
\end{proof}

Next, we show that any $(1-\varepsilon)$-approximate solution must contain
``almost'' $k$~small items.
\begin{lemma}
\label{lem:recourse-knapsack-numSmall}
	Suppose $S_i'$ is a $(1-\varepsilon)$-approximate solution. Then
	\begin{align*}
		k - \frac{n}{8s} \leq \numSmall(S_i') \leq k.
	\end{align*}
\end{lemma}
\begin{proof}
	The upper bound follows from the fact $S_i'$ must contain the heavy item (by
	Lemma~\ref{lem:recourse-knapsack-both}) of weight $n-k$ and then it can only
	include $k$~small items since the budget constraint is set to $B=n$.

	To prove the lower bound, note that since $S_i'$ is a
	$(1-\varepsilon)$-approximate solution, we have that its solution has value
	\begin{align*}
		p + \numSmall(S_i')
		\geq (1-\varepsilon)\OPT
		= (1-\varepsilon) (p+k).
	\end{align*}
	Hence, we get that
	\begin{align*}
		\numSmall(S_i')
		&\geq k - \varepsilon(p+k) \\
		&= k - \varepsilon(1+2\varepsilon)n \\
		&> k - \frac{n}{8s},
	\end{align*}
	where we used that 
	$s < \frac{1}{8\varepsilon(1+2\varepsilon)}$.
\end{proof}

To finish the proof of the theorem, we condition on the event from
Lemma~\ref{lem:recourse-knapsack}, i.e., we
have that $\abs{k-\numSmall(S_i)} \geq \frac{n}{4s}$ for all $i\in[s]$.

Now consider any solution $S_i'$ after the insertion of the heavy item. If
$S_i'$ is not $(1-\varepsilon)$-approximate, we can ignore $S_i'$. If $S_i'$ is 
$(1-\varepsilon)$-approximate then it satisfies
$k - \frac{n}{8s} \leq \numSmall(S_i') \leq k$ by
Lemma~\ref{lem:recourse-knapsack-numSmall}. Since we are assuming the event from
Lemma~\ref{lem:recourse-knapsack}, the algorithm had to insert/delete at least
$\frac{n}{8s}$ small items into/from $S_i$ to obtain $S_i'$.

Since the event from Lemma~\ref{lem:recourse-knapsack} occurs with probability
at least $1/2$ and the above argument holds for all
$(1-\varepsilon)$-approximate solutions $S_i'$, we have that the expected
recourse is $\Omega(\frac{n}{s})$.

\subsection{Proof of Theorem~\ref{thm:recourse-partitioning}}
\label{sec:proof-recourse-partitioning}
We prove Theorem~\ref{thm:recourse-partitioning}.  Again, we apply Yao's
principle~\cite{yao77probabilistic}, i.e., we consider a deterministic algorithm
and give a distribution over inputs, showing that in expectation the algorithm
will have amortized recourse $\Omega(\varepsilon^2 \frac{n}{k})$.

We consider a graph with $n$~vertices. Our initial instance consists of
$\frac{k}{2\varepsilon}$~star graphs, each of which contains
$2\varepsilon \frac{n}{k}$~vertices. Note that here an optimal solution places
the vertices from exactly $\frac{1}{2\varepsilon}$~star graphs into each
partition $V_j$; there are no edges between vertices from different $V_j$ and
hence the optimal cut-value is zero. Hence, the solution of any
$(C,1+\varepsilon)$-approximate solution must also have cut-value zero.

In the update phase, we sample $s$~edges between the central nodes of the
star graphs uniformly at random and insert them into the graph. Note that after
the insertion of the edges, we connected at most $s$~star graphs and the largest
connected component has size at most
$s\cdot 2\varepsilon \frac{n}{k} \leq \frac{1}{2}\frac{n}{k}$, where we used
that $s\leq\frac{1}{4\varepsilon}$. Hence, the optimal solution still has
cut-value
zero and thus any $(C,1+\varepsilon)$-approximate solution must have cut-value zero.

Next, let us analyze the recourse of an algorithm which starts with initial
solutions $S^{(0)}_1,\dots,S^{(0)}_s$. In a first step, we show that solutions
which at time~$0$ splits one of the star graphs up ``too much'' must entail high
recourse. In a second step, we consider all other solutions and show that our
insertions still trigger high recourse in expectation.

We say that a solution
$S^{(0)}_i=\{V^{(i,0)}_1,\dots,V^{(i,0)}_k\}$ is \emph{useful} if for all star
graphs~$H$, it holds that there exists an
index~$j$ such that $V^{(i,0)}_j$ contains at least $\varepsilon \frac{n}{k}$
vertices from $H$ and at most $\varepsilon \frac{n}{k}$~vertices are placed
in $\bigcup_{j'\neq j}V^{(i,0)}_{j'}$. Given a star graph $H$ and a solution
$S^{(0)}_i$, we write $j(H,i)$ to the denote the index~$j$ such that
$V^{(i,0)}_j$ contains at least at least $\varepsilon \frac{n}{k}$
vertices from $H$.
If a solution is not useful, we call it \emph{useless}.

First, consider solutions which are useless.  There exist two cases.
\emph{Case~A}: Suppose there exists a star graph~$H$ such that for all indices~$j$
it holds that $\bigcup_{j'\neq j}V^{(i,0)}_{j'}$ contains more than $\varepsilon
\frac{n}{k}$ vertices from $H$. Observe that if the algorithm wants to use this
solution after the edge insertions finished, it must ensure that the cut-value
is zero.  Thus it must move at least $\varepsilon \frac{n}{k}$~vertices to one
of the $V^{(i,0)}_j$ which requires
$\abs{\bigcup_{j'\neq j}V^{(i,0)}_{j'}}\geq\varepsilon\frac{n}{k}$ vertex moves.
\emph{Case~B}: Suppose there exists a star graph~$H$ such that for all $j$~it
holds that $V^{(i,0)}_j$ contains less than $\varepsilon \frac{n}{k}$ vertices
from $H$. Using that $\varepsilon\in(0,\frac{1}{2})$, also in this case the
algorithm must move at least
$(1-\varepsilon) \frac{n}{k}\geq\varepsilon\frac{n}{k}$~vertices such that
eventually all of $H$ is contained in the same set $V^{(i,s)}_j$ when the
updates finished.  We conclude that for useless solutions our theorem holds
after amortizing over $s\leq\frac{1}{\varepsilon}$~insertions.

Second, for the remainder of the proof consider only solutions
$S^{(0)}_i=\{V^{(i,0)}_1,\dots,V^{(i,0)}_k\}$ which are
useful.
Observe that when we insert an edge between two star graphs $H_1$ and $H_2$,
then if $j(H_1,i)\neq j(H_2,i)$ the algorithm must move at least
$\varepsilon\frac{n}{k}$~vertices to ensure that after the $s$~insertions
finished, all vertices from $H_1$ and $H_2$ are placed in the same set
$V^{(i,s)}_j$ for some~$j$. We call such an insertion \emph{expensive} for
solution~$i$.

Observe that if our edge insertions are such that they contain an expensive
insertion \emph{for all} solutions, then updating \emph{any} solution
$S^{(0)}_i$ such that $S^{(s)}_i$ is $(C,1+\varepsilon)$-approximate will incur
recourse at least $\varepsilon\frac{n}{k}$.  The rest of our proof is devoted to
showing that with constant probability this event occurs. This will prove the
theorem.

We start by considering a fixed solution $S^{(0)}_i$ and a single random edge
insertion between randomly picked star graphs $H_1$ and $H_2$.
Recall that there are $\frac{k}{2\varepsilon}$ star graphs in total.
Furthermore, 
we have that $\abs{V^{(i,0)}_j}\leq (1+\varepsilon)\frac{n}{k}$ for
all~$j$ and thus for each $j$ there can be at most
$\frac{(1+\varepsilon)}{\varepsilon}$ star graphs $H$ with
$j = j(H,i)$. Hence, for the probability that the edge insertion is
expensive we get that
\begin{align*}
	\Prob{j(H_1,i) \neq j(H_2,i)}
	&= 1 - \Prob{j(H_1,i) = j(H_2,i)} \\
	&\geq 1 - \frac{(1+\varepsilon)/\varepsilon}{k/(2\varepsilon)} \\
	&= 1 - \frac{2(1+\varepsilon)}{k} \\
	&\geq \frac{1}{2},
\end{align*}
where we used that $k\geq 4$.

Next, we consider a fixed solution $S^{(0)}_i$ and $s$~edge insertions between
star graphs which were picked independently and uniformly at random. Then with
probability at least $1-2^{-s}$, at least one of these edge insertions is
expensive for solution~$i$.

Finally, observe that probability that for \emph{all} solutions~$i$ there exists
an expensive edge insertion is at least
\begin{align*}
	\left(1 - 2^{-s}\right)^s
	&= \exp\left( s \ln(1-2^{-s}) \right) \\
	&\geq 1 + s \ln(1-2^{-s}) \\
	&\geq 1 - s 2^{-s} \\
	&\geq \frac{1}{4},
\end{align*}
where we used that $\exp(x)\geq1+x$ for all $x\in\mathbb{R}$, the Taylor
expansion of $\ln(x)$ for $x$ close to~$1$ and the fact that
$s2^{-s}\leq \frac{3}{4}$ for all $s$.

We conclude that with constant probability, for all solutions~$i$ there exists
an expensive edge insertion. In this case, the algorithm has total recourse at
least $\varepsilon \frac{n}{k}$. Hence, the expected total recourse of the
algorithm is $\Omega(\varepsilon \frac{n}{k})$.  Since we only performed
$s$~edge insertions, this gives an \emph{amortized} recourse of
$\Omega(\varepsilon^2 \frac{n}{k})$.

\section{Non-Monotone Functions and $\ell_\infty$-Necklace Alignment}
\label{sec:generalization}

So far we have only considered \emph{monotone} piecewise constant
functions. Now we will generalize some of our results to
\emph{piecewise constant functions with multiple non-monotonicities} and provide
the details in Section~\ref{sec:non-monotone}. We also derive new approximation
algorithms for the $\ell_\infty$-necklace problem in Section~\ref{sec:necklace}.
In particular, for $\ell_\infty$-necklace we present the first approximation
algorithm \emph{with near-linear running time} with additive
error~$\varepsilon$. We also present the first \emph{dynamic} approximation
algorithm for this problem which achieves additive error $\varepsilon$ and has
update time $O((1/\varepsilon)^{2} \lg(1/\varepsilon)$); the algorithm has
preprocessing time $O(1)$ when starting with empty vectors $x$ and $y$ and
requires \emph{sublinear space} $O(1/\varepsilon)$. See
Theorem~\ref{thm:necklace} for the details of our results.

\subsection{Piecewise Constant Functions With Non-Monotonicities}
\label{sec:non-monotone}
We now show that we can perform efficient operations on piecewise constant
functions even when these functions contain non-monotonicities. However,
the running times of our subprocedures will typically have some dependency on
the number of non-monotonicities of the function.

Let us formalize our notion of non-monotonicities.  We say that a function
$f \colon [0,t) \to [0,W]\cup\{-\infty,\infty\}$ has
\emph{$k$~monotone segments} if there exist values $0=x_0<x_1<\dots<x_{k}=t$
such that on each interval $[x_i,x_{i+1})$, $f$ is monotone. Here, we require
that either $f$ is monotonically decreasing on \emph{all} segments or it is
monotonically increasing on \emph{all} segments.  Note that a monotone function
has $k=1$ monotone segments (by setting $x_0=0$ and $x_1=t$) and that the points
$x_1,\dots,x_{k-1}$ can be viewed as the points in which $f$ is non-monotone.

One crucial operations will again be rounding.  However, unlike previously we
will mostly talk about rounding to \emph{multiples} of~$\delta$ instead of
rounding to powers of $1+\delta$.  This will be convenient for our applications
to $\ell_{\infty}$-necklace later.  We will also briefly mention how to extend
our results from this subsection to the setting in which we round to powers of
$1+\delta$.

Next, let $\delta > 0$ and consider a simple rounding function that rounds down
to multiples of $\delta$. More concretely, for $y\in\mathbb{R}$ we set
$\lfloor y\rfloor_\delta^* = \max\{ i\cdot \delta \colon i\cdot \delta \leq y, i\in\mathbb{Z}\}$
and we follow the convention that
$\lfloor -\infty\rfloor_\delta^* = -\infty$ and
$\lfloor \infty\rfloor_\delta^* = \infty$.
We also extend the rounding operation to functions
$f \colon [0,t)\to[0,W]\cup\{-\infty,\infty\}$ by defining
$\lfloor f\rfloor_{\delta}^* \colon [0,t)\to[0,W]\cup\{-\infty,\infty\}$
to be the function with
$\lfloor f \rfloor_{\delta}^*(x) = \lfloor f(x) \rfloor_{\delta}^*$ for all
$x\in[0,t)$. Next, we show that the function $\lfloor f \rfloor_{\delta}^*$ can
be computed efficiently and that it has only few pieces.
\begin{lemma}
\label{lem:rounding-multiples}
	Let $\delta > 0$ and let $f \colon [0,t)\to[0,W]\cup\{-\infty,\infty\}$
	be a piecewise constant function with $p$ pieces and $k$ monotone segments.
	Then we can compute the function $\lfloor f\rfloor_\delta^*$ in time
	$O(p\lg p)$ and $\lfloor f\rfloor_\delta^*$ has $O(k\cdot W/\delta)$ pieces.
\end{lemma}
\begin{proof}
	Let $(x_1,y_1),\dots,(x_p,y_p)$ denote the list representation of $f$.  We
	construct the list representation $(x_1',y_1'),\dots,(x_{p}',y_{p}')$ of
	$\lfloor f\rfloor_\delta^*$.  For all $i=1,\dots,p$, we set $x_i' = x_i$ and
	$y_i' = \lfloor y_i \rfloor^*_\delta$. After that, we merge all consecutive
	pieces that have the same $y_i'$-values; this can be done exactly as in the
	pruning step described in the proof of Lemma~\ref{lem:operations}.  Since
	$f$ takes values in $[0,W]\cup\{-\infty,\infty\}$, there are $O(W/\delta)$ choices for multiples
	of $\delta$ in $[0,W]$. In particular, on each monotone segment of $f$,
	$\lfloor f\rfloor_\delta^*$ has $O(W/\delta)$ pieces.  Since $f$ has $k$
	monotone segments this implies that $\lfloor f\rfloor_\delta^*$ has
	$O(k\cdot W/\delta)$ pieces in total.  Note that all operations from above
	can be performed in linear time and the running time bound stems from the
	fact that we also need to store the pieces in a binary search tree.
\end{proof}

Next, we show that we can compute the $(\min,+)$-convolution of two
piecewise constant functions in time that is quadratic in the number of their
pieces. The lemma generalizes the result from Lemma~\ref{lem:convolution}
because we drop the assumption that one of the functions needs to be
monotone (but this comes at the cost of a more complicated proof). We prove the
lemma in Section~\ref{sec:proof-convolution-general}.
\begin{lemma}
\label{lem:convolution-general}
	Let $f_1,f_2 : [0,t) \to [0,W]\cup\{-\infty,\infty\}$ be piecewise
	constant functions which have at most $p$ pieces. Then we can compute
	$f=f_1\minconv f_2$ in time $O(p^2 \lg p)$ and $f$ has $O(p^2)$ pieces.
\end{lemma}

By combining the two lemmas above, we can show that we can efficiently compute
additive approximations of $(\min,+)$-convolutions even in the case of
non-monotonicities. More concretely, we say that
$f \colon [0,t) \to [0,W]\cup\{-\infty,\infty\}$ is
\emph{an additive $\varepsilon$-approximation} of
$g \colon [0,t) \to [0,W]\cup\{-\infty,\infty\}$ if
$g(x)-\varepsilon \leq f(x) \leq g(x)$ for all $x\in[0,t)$.
Now we obtain the following theorem.
\begin{theorem}
\label{thm:convolution-general}
	Let $f,g \colon [0,t) \to [0,W]\cup\{-\infty,\infty\}$ be two
	functions with $k$ monotone segments and suppose we have already computed
	$\lfloor f\rfloor^*_\delta$ and $\lfloor g\rfloor^*_\delta$. Then 
	the function $(\lfloor f\rfloor^*_{\delta}) \minconv (\lfloor g \rfloor^*_{\delta})$
	is an additive $2\delta$-approximation of $f\minconv g$, has at most
	$O((k\cdot W/\delta)^2)$ pieces and can be computed in time
	$O((k\cdot W/\delta)^2 \lg( (k\cdot W/\delta)^2 ))$.
\end{theorem}
\begin{proof}
	The approximation ratio follows from the triangle inequality.  The claims
	about the number of pieces and the running time follow from combining
	Lemma~\ref{lem:rounding-multiples} and Lemma~\ref{lem:convolution-general}.
\end{proof}

We note that by stating Lemma~\ref{lem:rounding-multiples} for the rounding
operation $\lceil \cdot \rceil_{1+\delta}$ that rounds to powers of $1+\delta$
(see Lemma~\ref{lem:operations}), we can obtain the following version of
Theorem~\ref{thm:convolution-general}.
\begin{theorem}
\label{thm:convolution-general-mult}
	Let $f,g \colon [0,t) \to [0,W]\cup\{-\infty,\infty\}$ be two
	functions with $k$ monotone segments and suppose we have already computed
	$\lceil f\rceil_{1+\delta}$ and $\lceil g\rceil_{1+\delta}$. Then 
	$(\lceil f\rceil_{1+\delta}) \minconv (\lceil g \rceil_{1+\delta})$
	is a $(1+\delta)$-approximation of $f\minconv g$, has at most
	$O((k\cdot \lg_{1+\delta}(W)^2)$ pieces and can be computed in time
	$O((k\cdot \lg_{1+\delta}(W))^2 \lg( (k\cdot \lg_{1+\delta}(W))^2 ))$.
\end{theorem}

This result generalizes our previous method of first rounding a monotone
function via Lemma~\ref{lem:operations} and then applying the efficient
convolution from Lemma~\ref{lem:convolution}.  More concretely, observe that
monotone functions have one monotone segment and, thus, after rounding both
functions, our algorithm from Lemma~\ref{lem:convolution} computes the
$(\min,+)$-convolution in time $O( \lg^2_{1+\delta}(W) \lg\lg_{1+\delta}(W) )$
which is the same running time that we obtain by combining the two lemmas above.
Hence, the algorithm from Theorem~\ref{thm:convolution-general-mult} matches
this result for $k=1$ and it generalizes it when we apply it for $k>1$.

\subsubsection{Proof of Lemma~\ref{lem:convolution-general}}
\label{sec:proof-convolution-general}
We assume that $f_i$ for $i = 1, 2$ is given as a doubly linked list $(x^i_1,y^i_1),\dots,(x^i_p,y^i_p)$ such that
$x^i_j < x^i_{j+1}$ for all $1 \le j < p$.
We will output $f$ in the same representation.

To compute $f$ we will make use of the following \emph{non-overlapping interval data structure (NOI)}.
Let $[a,b]$ and $[a', b']$ be two subsets of the real line.
We call each of them an \emph{interval} and say that they \emph{overlap} if $[a,b] \cap [a', b'] \not= \emptyset$.
We say that an interval $[a,b]$ is \emph{empty} if $a \ge b$.
The NOI data structure stores a set $S$ of non-overlapping,   non-empty intervals $I = [a,b]$  
 and supports the following operations:
\begin{itemize}
	\item \emph{ClosestLargerInterval($z$)}, which given a number $z$ returns the interval $[a,b]$ together with a Boolean value $bool$. If $bool$ is true, then $z \le b$ and there is no interval $[a',b']$ in $S$ with $z \le  b' < b$. Note that it is possible that $z$ belongs to $[a,b]$. If $bool$ is false, then there exists no interval $[a,b]$ with $z \le b$ and the returned values for $a$ and $b$ are undefined.
	\item \emph{InsertInterval($a,b$)}, which  inserts the interval $[a,b]$ into $S$, merging it with any interval that it overlaps with and updating $S$ accordingly.
\end{itemize}

There exists an efficient implementation of such a data structure as stated in the next claim, which we prove at the end of this section.

\begin{claim}
\label{claim:my}
There exists an implementation of the non-overlapping interval data structure such that any sequence  of $q$ operations takes time $O(q \log q)$. 
\end{claim}

We compute $f$ as follows. Note that the function values of $f_1$ and of $f_2$ are 
constant  over each 2-dimensional rectangle whose corners are $(x^1_s, x^2_t)$,
$(x^1_s, x^2_{t+1})$, $(x^1_{s+1}, x^2_t)$, and $(x^1_{s+1}, x^2_{t+1})$ for any $1 \le s \le p$ and
$1 \le t \le p$. We call this rectangle $R_{st}$ and denote by
$[x^1_s + x^2_t, x^1_{s+1} + x^2_{t+1}]$ the
\emph{range} of the rectangle $R_{st}$ and by $y_s^1 +y^2_t$ the
\emph{function value} of the rectangle, where we assume that
$\infty + y$ with $y \in \Winfty$ equals $\infty$. 
There are $K^2$ such rectangles.

Now note that for any value $x$ with $x^1_s + x^2_t \le x \le x^1_{s+1} + x^2_{t+1}$,
i.e., $x$ is in the range of the rectangle $R_{st}$,
the function value $y_s^1 +y^2_t$ is one of the sums that occurs in the computation of
$f(x)=\min_{\bar{x}} \{f_1(\bar{x}) + f_2(x-\bar{x})\}$.
We will compute $f(x)$ (for all values $x$ ``simultaneously'') by comparing the function values of all rectangles $R_{st}$ to whose range $x$ belongs. The main observation that we exploit is the following: \emph{As we will consider the rectangles by decreasing function values, the first rectangle (in this order) to whose range a value $x$ belongs is the rectangle whose function value equals $f(x)$. }

Thus, when processing a rectangle, we need to determine all ranges, i.e, subintervals of $[0,t]$, to which no function value has yet been assigned. To do so, we use the NOI data structure to store the intervals of all values $x$ for which we have already assigned a function value. Furthermore we use a 
balanced binary search tree
$\cal B$ that stores at its leaves every interval to which a function value has already been assigned, together with its (constant) function value. Specifically, we will store these ranges in the leaves of
$\cal B$, ordered by their smaller boundary value $x'$. The difference between the two is that the NOI data structures merges overlapping intervals, no matter what their function value is, while every interval stored as a leaf of {\cal B} has the \emph{same} function value, i.e., has a constant $f$-value.

To be precise we proceed as follows: We first generate all rectangles $R_{st}$
by iterating over the lists of $f_1$ and $f_2$  and sort them by non-decreasing order of their function value. This takes time $O(p^2 \log p)$. Then we process the rectangles in this order. To do so, we first initialize an empty NOI data structure as well as an empty balanced binary search tree
$\cal B$. 
Next we describe how to process the rectangles. Let $R_{st}$ be the next
rectangle to be processed. We execute the following steps for $R_{st}$:
\begin{enumerate}
\item $z = x^1_s + x^2_t$
\item $(a ,b, bool)= $ \emph{ClosestLargerInterval($z$)}
\item \textbf{while} $bool$ is true and $b < x^1_{s+1} + x^2_{t+1}$ \textbf{do}
\begin{enumerate}
\item if $z \not\in [a, b]$ then insert the interval $[z, a]$ together
with the function value of $R_{st}$ into $\cal B$.
\item $z = b$
\item
$(a ,b, bool)= $ \emph{ClosestLargerInterval($z$)}
\end{enumerate}
\item If $bool$ is true then insert the interval $[z, a]$ together
with the function value of $R_{st}$ into $\cal B$; else
insert the interval $[z, x^1_{s+1} + x^2_{t+1}]$ together
with the function value of $R_{st}$ into $\cal B$.
\item \emph{InsertInterval($x^1_s+ x^2_t, x^1_{s+1} + x^2_{t+1}$)}.
\end{enumerate}

Once all rectangles have been processed, we traverse the leaves of $\cal B$ in order  and  connect them by 
a doubly linked list to create an  (ordered) list representation of the function $f$.
As we process the rectangles in increasing order of function value this guarantees that for each value $x$ the smallest function value of any rectangle $R_{st}$ is returned as $f(x)$.

Note that each insertion into $\cal B$ takes time $O(\log p)$ and the number of calls to the NOI data structure is proportional to the number of rectangles plus the number of intervals merged in the NOI data structure. As processing a rectangle creates at most one new interval, and merged intervals are never separated again, the number of interval merges is at most the number of rectangles. Thus, there are at most $p$ interval merges and
at most $2 p^2$  insertions into $\cal B$. Hence, the total running time for the above algorithm is $O(p^2 \log p)$
plus the time for the NOI data structure, which, by Claim~\ref{claim:my}, is also $O(p^2 \log p)$ as $q = O(p^2)$.

We still have to prove Claim~\ref{claim:my}.
\begin{proof}[Proof of Claim~\ref{claim:my}]
We implement the NOI data structure with a balanced binary search tree. The leaves store the non-overlapping intervals, ordered by their upper endpoint.

The \emph{ClosestLargerInterval($z$)} operation searches for the interval $[a,b]$ such that $b$ is the smallest upper endpoint of an interval that is at least $z$. If no such interval exists, $bool$ is set to false, otherwise it is set to true and $[a,b]$ is returned as interval.
Note that finding $[a,b]$ takes time $O(\log q)$, as $q$ is the maximum number of intervals stored in the balanced binary tree.

The \emph{InsertInterval($a,b$)} operation first executes a \emph{ClosestLargerInterval($a$)} operation. Let $(a', b', bool)$ be the result. If $bool$ is false, then the interval $[a,b]$ is inserted as new interval and the procedure terminates.
Otherwise the interval $[a', b']$ is the interval with smallest upper endpoint such that $a \le b'$.
Note that $[a', b']$ might overlap with $[a,b]$ and we test for this next.
If $b < a'$ then  a leaf with range $[a,b]$ is inserted into the balanced search
tree and \emph{InsertInterval($a,b$)} terminates. Otherwise ($b \geq a'$),
let $L$ be the leaf of the balanced search tree that stores $[a', b']$. 
If $b \le b'$, the two intervals are merged by updating $L$ to store the interval $[\min(a, a'),b']$ and \emph{InsertInterval($a,b$)}  terminates.
If, however, $b > b'$,
it is possible that the new interval $[a,b]$ overlaps with even more intervals in
$S$. Thus, we execute the following steps:
\begin{enumerate}
	\item $z = b'$
	\item $(a'', b'', bool)=$ \emph{ClosestLargerInterval($z$)}
	\item \textbf{while} $bool$ is true \textbf{do}
	\begin{enumerate}
	\item If $b < a''$ then the leaf $L$ is updated to store the interval
	$[\min(a,a'),b]$ and \emph{InsertInterval} terminates. 
	\item The leaf storing the interval $[a'', b'']$ is removed from the
		balanced search tree.
	\item If $b \le b''$
	then the leaf $L$ is updated to store the interval $[\min(a, a'),b'']$
	and  \emph{InsertInterval} terminates. 
	\item Otherwise, $z = b''$
	and $(a'', b'', bool)=$ \emph{ClosestLargerInterval($z$)}.
\end{enumerate}
\item  $L$ is
updated to store the interval $[\min(a, a'),b]$. 
\end{enumerate}

\smallskip
Note that this algorithm merges all intervals that overlap with $[a,b]$ into one interval and updates the balanced search tree accordingly.

Let $t$ be the number of iterations executed by \emph{InsertInterval($x,y$)}. The running time is
$O((t+1) \log q)$ as each iteration executes one call to 
\emph{ClosestLargerInterval}, one deletion  of a leaf in the balanced binary
tree, and at most one modification of a label at a leaf. Every such iteration
decreases the number of leaves in the balanced binary tree by 1.
Furthermore, each call to \emph{InsertInterval} that does not execute any
iterations of the above while-loop increases the number of leaves by at most 1
and there is no other operation that modifies the number of leaves.
As there are at most $q$ calls to \emph{InsertInterval}, the while-loop can be executed at most $q$ times 
\emph{over all calls to InsertInterval}, each taking time $O(\log q)$. Thus, the
total runnning time for  $q$ calls to
\emph{InsertInterval} is $O(q \log q)$.
\end{proof}

\subsection{$\ell_\infty$-Necklace Alignment}
\label{sec:necklace}

Using our techniques from above, we present a novel approximation algorithm for the
$\ell_\infty$-necklace alignment
problem~\cite{toussaint2004geometry,bremner14necklaces}.  In this problem, the
input consists of two necklaces represented as two \emph{sorted} vectors of
$n$~real numbers, $x = \langle x_0, x_1, \dots, x_{n-1} \rangle$ and $y =
\langle y_0, y_1, \dots, y_{n-1} \rangle$, where the $x_i,y_i \in [0,1)$ represent
points on the unit-circumference circle. We will sometimes refer to the elements
$x_i$ and $y_j$ as \emph{beads}.

We define the distance between two beads $x_i$ and $y_j$ by the minimum of the
clockwise and counterclockwise distances along the circumference of the
unit-perimeter circular necklaces, i.e., we set
$$
d^\circ (x_i,y_j) = \min \{\abs{x_i-y_j}, 1-\abs{x_i-y_j} \} .
$$

In the \emph{$\ell_\infty$-necklace alignment problem}, we need to find an
\emph{offset} $c \in [0,1)$ and a \emph{shift} $s\in[n+1]$ that minimize 
$$
	\max_{i=0}^{n-1} (d^\circ ((x_i+c)\mod 1, y_{(i+s) \mod n}) ).
$$
In the above definition, the offset~$c$ encodes how much we rotate the first
necklace clockwise relative to the second necklace.  Additionally, the shift~$s$
defines a perfect matching between the beads such that bead $i$ of the first
necklace is matched with bead $(i+s) \mod n$ of the second necklace.

Bremner et al.~\cite{bremner14necklaces} showed that the $\ell_\infty$-necklace
alignment problem can be solved \emph{exactly} in time
$\tO(n^2)$. We complement this by showing that we can compute a solution with
additive error~$\varepsilon$ in time $\tO(n + \varepsilon^{-2})$.

We also consider the dynamic version of the problem in which beads are inserted
and deleted. More concretely, we assume that initially $x$ and $y$ are empty and
we offer the following update operations:
\begin{itemize}
	\item \emph{Insert}($i$, $\alpha$, $\beta$) which inserts $\alpha\in[0,1)$
	into $x$ at the $i$'th position and it further inserts $\beta\in[0,1)$ into
	$y$ at the $i$'th position.  We require that after the insertion, $x$ and
	$y$ are still ordered.
	\item \emph{Delete}($i$) which deletes $x_i$ from $x$ and $y_i$ from $y$.
\end{itemize}
Note that
both of these operations change the number of entries in $x$ and $y$ but they
ensure that $x$ and $y$ always have the same length. We show that we can
maintain a solution with additive error $\varepsilon$ using update time
$O(1/\varepsilon^2 \lg(1/\varepsilon))$. The preprocessing time is $O(1)$ and
\emph{the space usage is only $O(1/\varepsilon)$} which is sublinear in the size
of the vectors $x$ and $y$.

\begin{restatable}{theorem}{necklace}
\label{thm:necklace}
	Let $\varepsilon>0$.
	There exists a static algorithm for the $\ell_\infty$-necklace alignment
	problem that computes a solution with additive error~$\varepsilon$ in time
	$O(n + (1/\varepsilon)^2 \lg(1/\varepsilon))$.
	Furthermore, there exists
	a fully dynamic algorithm for the $\ell_\infty$-necklace alignment problem that
	maintains a solution with additive error~$\varepsilon$ with update time
	$O(1/\varepsilon^2 \lg(1/\varepsilon))$ and preprocessing time $O(1)$; the
	space usage of the algorithm is $O(1/\varepsilon)$.
\end{restatable}
To obtain the result for the dynamic algorithm, we show that for vectors
$A,B\in\mathbb{R}^n$ that are undergoing element insertions and deletions,
we can dynamically maintain an approximation of the $(\min,+)$-convolution
$A\minconv B$.  We expect that this result will have further applications.  The
proof of the theorem follows from Propositions~\ref{prop:necklace-static}
and~\ref{prop:necklace-dynamic} below.

\subsubsection{The Static Algorithm}
\label{sec:necklace-static}

Now we consider our static algorithm and prove the following proposition.
\begin{proposition}
\label{prop:necklace-static}
	There exists a static algorithm for the $\ell_\infty$-necklace alignment
	problem that computes a solution with additive error~$\varepsilon$ in time
	$O(n + (1/\varepsilon)^2 \lg(1/\varepsilon))$.
\end{proposition}

We devote the rest of this subsection to the proof of the proposition.

\textbf{The Algorithm.}
Our algorithm is rather simple and (up to the part in which we perform the
rounding) it is the same as the one used by Bremner et al.~\cite{bremner14necklaces}.
Consider the input $\varepsilon$ (as error parameter),
$x=\langle x_0,x_1,\dots,x_{n-1}\rangle$ and
$y=\langle y_0,y_1,\dots,y_{n-1}\rangle$.
Now we set $\delta=\varepsilon/2$ and perform a single pass over $x$ and
$y$ and apply the rounding function $\lfloor\cdot \rfloor^*_{\delta}$ to each
of the entries. While doing so, we compute the list representations of $x$ and
$y$ (where we interpret $x$ and $y$ as functions from $[0,n)$ to $[0,1)$) which
have at most $O(1/\delta)$ pieces (by applying
Lemma~\ref{lem:rounding-multiples} with $W=1$).
Then we compute the vectors 
\begin{align*}
	x' &= \langle x_0, x_1, \dots, x_{n-1},
				\underbrace{\infty,\dots,\infty}_{n\text{ times}} \rangle, \\
	x'' &= \langle x_0, x_1, \dots, x_{n-1},
				\underbrace{-\infty,\dots,-\infty}_{n\text{ times}} \rangle, \\
	y' &= \langle y_{n-1},y_{n-2},\dots,y_0,y_{n-1},y_{n-2},\dots,y_0 \rangle,
\end{align*}
but we do not store them explicitly. Instead, we only store their list
representations. We note that $x'$ is a monotonically increasing vector,
$x''$ has two monotonically increasing segments and $y'$ has two monotonically
decreasing segments.

Next, we set $a$ to the $(\min,-)$-convolution of $x'$ and $y'$ and we set $b$
to the $(\max,-)$-convolution of $x''$ and $y'$ (we show below in
Lemma~\ref{lem:other-convolutions} that we can compute these functions
efficiently).

Finally, we set $v = \frac{1}{2}(b-a)$ and return $\min\{ v_s \colon s\in[n]\}$
as the solution for our problem. We note that $v$ can be efficiently computed
via the list representations of $a$ and $b$ and we can also quickly find the
minimum over the $v_s$ by iterating over the list representation of $v$.

\textbf{Analysis.} Now we turn to the analysis of the algorithm above.
We adapt the proof of Theorem~6 in Bremner et al.~\cite{bremner14necklaces} for
approximate solutions and argue how to implement it using piecewise constant
functions.

We start by showing that we can compute $(\min,-)$-convolution and
$(\max,-)$-convolution as efficiently as the classic $(\min,+)$-convolution.
\begin{lemma}
\label{lem:other-convolutions}
	Let $f$ and $g$ be two piecewise constant functions with $p$ pieces and
	suppose that $g$ has $k$~monotonically decreasing segments.  Suppose that we
	can compute the $(\min,+)$-convolution of $f'$ and $g'$ in time $t(p,k)$ if
	$f'$ and $g'$ have $k$~monotonically decreasing segments.  Then in time
	$O(t(p,k) + p\lg p)$ we can compute:
	\begin{itemize}
		\item The $(\max,-)$-convolution of $f$ and $g$ if $f$ has $k$
		monotonically increasing segments.
		\item The $(\min,-)$-convolution of $f$ and $g$ if $f$ has
		$k$~monotonically increasing segments.
	\end{itemize}
\end{lemma}
\begin{proof}
	First, suppose that we wish to compute the $(\max,-)$-convolution of two
	functions $f$ and $g$. We show that we can compute the
	$(\max,-)$-convolution of $f$ and $g$ via the $(\min,+)$-convolution of $-f$
	and $g$. Indeed, for all $x$ it holds that:
	\begin{align*}
		\max_{\bar{x}\in[0,x]} \{ f(\bar{x}) - g(x-\bar{x}) \}
		&= \max_{\bar{x}\in[0,x]} \{- (-f(\bar{x}) + g(x-\bar{x})) \} \\
		&= -\min_{\bar{x}\in[0,x]} \{ -f(\bar{x}) + g(x-\bar{x}) \} \\
		&= -(((-f)\minconv g)(x)).
	\end{align*}
	To see that the running time is correct, note that we can compute the list
	representation of $-f$ in time $O(p)$ and it takes takes $O(p \lg p)$ to
	update the binary search tree in which we store the pieces of $-f$.
	Furthermore, $-f$ has $k$~monotonically decreasing segments since $f$ has
	$k$~monotonically increasing segments. Thus, we can apply the efficient
	algorithm for $(\min,+)$-convolution in time $t(p,k)$ on $-f$ and $g$.

	We can prove the result for $(\min,-)$-convolution similarly by computing a
	$(\min,+)$-convolution of $g$ and $-f$. More concretely, for all $x$ it
	holds that
	\begin{align*}
		\min_{\bar{x}\in[0,x]} \{ f(\bar{x}) - g(x-\bar{x}) \}
		=\min_{\bar{x}\in[0,x]} \{ -f(\bar{x}) + g(x-\bar{x}) \}
		= ((-f)\minconv g)(x),
	\end{align*}
	where in the first step we used the symmetry of $(\min,-)$-convolution.  The
	running time analysis is exactly as above.
\end{proof}

In the proof of Proposition~\ref{prop:necklace-static} we need the following
lemma. We will use the lemma to find the optimal offset~$c$ for a given
shift~$s$.
\begin{lemma}[Fact~5 in~\cite{bremner14necklaces}]
\label{fact:bremner}
	Let $z = \langle z_{0}, z_{1}, \dots, z_{n-1} \rangle$.  Then 
	\begin{align*}
		\min_{c\in\mathbb{R}} \max_{i=0}^{n-1}\abs{z_i+c} =
		\frac{1}{2} \left(\max_{i=0}^{n-1} z_i - \min_{i=0}^{n-1} z_i \right)
	\end{align*}
	and the minimizer for this quantity is given by
	$c=-\frac{1}{2} (\min_{i=0}^{n-1} z_i + \max_{i=0}^{n-1} z_i)$.
\end{lemma}

Next, we can prove Theorem~\ref{thm:necklace}.
\begin{proof}[Proof of Theorem~\ref{thm:necklace}]
	We prove the theorem in three steps. In Step~1, we will prove that we compute the correct result
	in the exact case (i.e., when we perform no rounding). This first part is
	essentially the same proof as in in Bremner et al.~\cite{bremner14necklaces}
	but with more details. In Step~2, we argue about approximation guarantee of
	our algorithm. In Step~3, we prove its running time.

	\emph{Step~1: The Exact Case.}
	First, we use Theorem~2 of Bremner et al.~\cite{bremner14necklaces} which
	states that if
	\begin{align*}
		\tilde{y} &= \langle y_{0},y_{1},\dots,y_{n-1},y_{0},y_{1},\dots,y_{n-1} \rangle
	\end{align*}
	then
	\begin{align*}
		\min_{c,s} \max_{i=0}^{n-1} d^\circ((x_i+c)\mod 1,y_{(i+s)\mod n})
		= \min_{c,s} \max_{i=0}^{n-1} d^-(x_i+c,\tilde{y}_{i+s}),
	\end{align*}
	where $d^-(a,b) = \abs{a-b}$ for all $a,b\in\mathbb{R}$. Thus, instead of
	directly optimizing the original objective function
	$\min_{c,s} \max_{i=0}^{n-1} d^\circ((x_i+c)\mod 1,y_{(i+s)\mod n})$, we
	will consider the more convenient objective function
	$\min_{c,s} \max_{i=0}^{n-1} d^-(x_i+c,\tilde{y}_{i+s})$ which involves no
	modulo operations.
	
	Indeed, consider the new objective function and for all $s\in[n]$ we define
	the vector $z(s)\in\mathbb{R}^n$ such that $z(s)_i = x_i - y_{(i+s) \mod n}$.
	Now we obtain that for the new objective function it holds that:
	\begin{align*}
		\min_{c,s} \max_{i=0}^{n-1} d^-(x_i+c, \tilde{y}_{i+s})
		&= \min_{c,s} \max_{i=0}^{n-1} \abs{x_i+c - y_{(i+s)\mod n})} \\
		&= \min_s \min_c \max_i \abs{z(s)_i + c} \\
		&= \min_s \frac{1}{2} \left( \max_i\{z(s)_i\} - \min_i\{z(s)_i\} \right),
	\end{align*}
	where in the first step we used the definition of $d^-(\cdot,\cdot)$ and
	that $\tilde{y}_k = y_{k \mod n}$ for all $k\in[2n]$, in the second step we
	substituted the definition of $z(s)_i$ and in the third step we applied
	Lemma~\ref{fact:bremner}.

	The above implies that we need to compute the quantities $\max_i\{z(s)_i\}$
	and $\min_i\{z(s)_i\}$ efficiently. Even more, consider the vector
	$v\in\mathbb{R}^n$ with entries
	$v_s = \frac{1}{2} \left( \max_i\{z(s)_i\} - \min_i\{z(s)_i\} \right)$
	and observe that the calculation above shows that the optimal objective function
	value is the same as the smallest entry in $v$.  Therefore, in the following
	we show that we can compute $v$ efficiently using the vectors $a$ and
	$b$ that we computed in our algorithm.

	Recall the definitions of the two vectors $x'$ and $y'$:
	\begin{align*}
		x' &= \langle x_0, x_1, \dots, x_{n-1},
					\underbrace{\infty,\dots,\infty}_{n\text{ times}} \rangle, \\
		y' &= \langle y_{n-1},y_{n-2},\dots,y_0,y_{n-1},y_{n-2},\dots,y_0 \rangle.
	\end{align*}
	Now we let $a\in\mathbb{R}^{2n}$ be the vector resulting from the
	$(\min,-)$-convolution of
	$x'$ and $y'$, i.e., $a_k=\min_{i}\{x_i' - y_{k-i}\}$ for all $k\in[2n]$.
	Now we observe that for each entry $a_{n+s'}$ with $s'\in[n]$, it holds that
	\begin{align*}
		a_{n+s'}
		= \min_{i=0}^{n+s'} \{ x_i' - y_{n+s'-i}' \}
		= \min_{i=0}^{n-1} \{ x_i - y_{(i-s'-1)\mod n} \},
	\end{align*}
	where in the second step we used that $x_i'=\infty$ for $i\geq n$ and that
	$y_{n+s'-i}' = y_{((n-1)-(n+s'-i))\mod n} = y_{(i-s'-1)\mod n}$
	since in $y'$ we concatenated the entries of $y$ twice but in reverse order.
	Now observe that if $s'=n-1-s$ then
	\begin{align*}
		a_{2n-s-1}
		= a_{n+s'} 
		= \min_{i=0}^{n-1} \{ x_i - y_{(i-s'-1)\mod n} \}
		= \min_{i=0}^{n-1} \{ x_i - y_{(i+s)\mod n} \}
		= \min_{i=0}^{n-1} \{ z(s)_i \}.
	\end{align*}

	Next, we define the vector $x''$ such that:
	\begin{align*}
		x'' &= \langle x_0, x_1, \dots, x_{n-1},
					\underbrace{-\infty,\dots,-\infty}_{n\text{ times}} \rangle.
	\end{align*}
	We let $b$ denote the vector resulting from the $(\max,-)$-convolution of
	$x''$ and $y'$, i.e., $b_{k} = \max_i\{x_i''-y_{k-i}'\}$ for
	all~$k\in[2n]$. Now a similar argument as above shows that
	$b_{2n-s-1}=\max_{i=0}^{n-1} \{ z(s)_i \}$ for all $s\in[n]$.
	More concretely, for each entry $b_{n+s'}$ with $s'\in[n]$ it holds that
	\begin{align*}
		b_{n+s'}
		= \max_{i=0}^{n+s'} \{ x_i' - y_{n+s'-i}' \}
		= \max_{i=0}^{n-1} \{ x_i - y_{(i-s'-1)\mod n} \},
	\end{align*}
	where we used that $x_i''=-\infty$ for $i\geq n$ and the same argument
	relating the entries of $y'$ and $y$ as above. Thus, if
	$s'=n-1-s$ then
	\begin{align*}
		b_{2n-s-1}
		= b_{n+s'} 
		= \max_{i=0}^{n-1} \{ x_i - y_{(i-s'-1)\mod n} \}
		= \max_{i=0}^{n-1} \{ x_i - y_{(i+s)\mod n} \}
		= \max_{i=0}^{n-1} \{ z(s)_i \}.
	\end{align*}

	Combining the results above we get that $v_s = \frac{1}{2}( b_{2n-s-1} - a_{2n-s-1} )$ for all
	$s\in[n]$.
	Therefore, we get that the optimal objective function value is given by
	$\min_s v_s = \min_s \frac{1}{2}( b_{2n-s-1} - a_{2n-s-1} )$. In other
	words, to compute the optimal objective function value it suffices to compute
	the difference $\frac{1}{2}(b-a)$ and then to return the smallest entry in
	$v$ with index between $n$ and $2n-1$.

	\emph{Step~2: Approximation Guarantees.}
	We argue that the algorithm returns an additive $\varepsilon$-approximation.
	First, observe that in the algorithm all computations are performed exactly
	except for the rounding at the beginning. In the rounding process, we
	decrease each entry by at most $\delta=\varepsilon/2$. Therefore, the
	triangle inequality implies that when we match bead $x_i$ to bead $y_{i+s}$, 
	the error that was introduced by the approximation is at most
	$2\delta=\varepsilon$. Since in the objective function we are only
	interested in the maximum error over all matched beads, this implies that we
	obtain an additive $\varepsilon$-approximation.

	\emph{Step~3: Running Time Analysis.}
	It is left to analyze the running time of our algorithm. Iterating over the
	input vectors $x$ and $y$, rounding the entries and computing the list
	representation of $x$ and $y$ can be done in time $O(n)$.  Recall that $x$
	and $y$ have $O(1/\delta)$ pieces. Therefore, we can also compute the
	vectors $x'$, $x''$ and $y$ in time $O(1/\delta \lg 1/\delta)$.
	Then Lemmas~\ref{lem:other-convolutions} and~\ref{lem:convolution-general}
	imply that we can compute the $(\min,-)$-convolution and the
	$(\min,+)$-convolutions in time $O(1/\delta^2 \lg(1/\delta))$ and the
	resulting vectors have $O(1/\delta^2)$ pieces. Finally, the vector $v$ can
	be computed in time $O(1/\delta^2 \lg(1/\delta))$ and the minimum that we
	return can be found by simply iterating over the pieces of $v$.
	Since previously we have set $\delta=\varepsilon/2$, this finishes the
	proof.
\end{proof}

\subsubsection{The Dynamic Algorithm}
\label{sec:necklace-dynamic}

We now give our extension to the dynamic setting of the $\ell_\infty$-necklace
alignment problem in which there are insertions and deletions from $x$ and $y$.
\begin{proposition}
\label{prop:necklace-dynamic}
	Let $\varepsilon>0$.  There exists a fully dynamic algorithm for the
	$\ell_\infty$-necklace alignment problem that maintains a solution with
	additive error~$\varepsilon$ with update time
	$O(1/\varepsilon^2 \lg(1/\varepsilon))$ and preprocessing time $O(1)$; the
	space usage of the algorithm is $O(1/\varepsilon)$.
\end{proposition}
\begin{proof}
	In the preprocessing, we initialize $x$ and $y$ as empty vectors and store
	them as piecewise constant functions (as per Section~\ref{sec:approx-conv})
	and we \emph{do not} store them explicitly as vectors. Furthermore, we set
	$\delta=\varepsilon/2$. These operations can be done in time $O(1)$.

	Next, consider an operation
	\emph{Insert}($i$, $\alpha$, $\beta$) which asks to insert $\alpha$ into $x$ at the
	$i$'th position and to insert $\beta$ into $y$ at the $i$'th position.
	Since we are in the approximate setting, instead of inserting the exact
	values of $\alpha$ and
	$\beta$, we insert
	$\lfloor\alpha\rfloor^*_\delta$ into the $i$'th position of $x$ and
	$\lfloor\beta\rfloor^*_\delta$ into the $i$'th position of $y$. We perform
	these insertions by manipulating the list representations of $x$ and $y$. We
	only describe how to perform the manipulations for $x$, as for $y$ they are
	essentially the same.
	
	Denote the list representation of $x$ as $(X_0,Y_0),\dots,(X_p,Y_p)$ where
	$p$ is the number of pieces of $x$. Now we iterate over all pieces of $x$
	and check whether there exists a piece with value
	$Y_j=\lfloor\alpha\rfloor^*_\delta$. If no such piece exists, we insert
	$(i,\lfloor\alpha\rfloor^*_\delta)$ into the list representation at the
	appropriate position.  Then we find the smallest integer $j$ such that
	$Y_j>\lfloor\alpha\rfloor^*_\delta$ and for all $k\geq j$, we increment
	$X_k$ by 1. Intuitively, we are moving all pieces that are larger than
	$\lfloor\alpha\rfloor^*_\delta$ one unit to the right in order to make space
	for the element that was just inserted.

	Once we have updated $x$ and $y$ as described above, we simply run the
	static algorithm without the step in which we initialize $x$ and $y$.  Note
	that, since we assume that after each insertion $x$ and $y$ are still
	ordered and since we only insert rounded entries into $x$ and $y$, we get
	that $x$ and $y$ never have have more than $O(1/\delta)$ pieces by
	Lemma~\ref{lem:rounding-multiples}. Now, since above we have set
	$\delta=\varepsilon/2$, the proof of Proposition~\ref{prop:necklace-static}
	implies that we obtain a solution with additive error $\varepsilon$ in time
	$O(1/\varepsilon^2\lg(1/\varepsilon))$.  Furthermore, note that since we do
	not store $x$ and $y$ explicitly (we only store their rounded version
	represented by their list representations), the space usage is
	$O(1/\varepsilon)$.

	Finally, we note that the operation \emph{Delete}($i$) can be implemented
	similar to above by first manipulating the list representations of $x$ and
	$y$ to remove the $i$'th entries from $x$ and $y$ and then running the
	static algorithm.
\end{proof}

We remark that by storing two dynamic vectors $x$ and $y$ that are undergoing
element insertions and deletions as described in the proof of
Proposition~\ref{prop:necklace-dynamic}, we can also efficiently maintain an
approximation of their $(\min,+)$-convolution $x\minconv y$ via
Lemma~\ref{lem:convolution-general}.

\section{Omitted Proofs}
\label{sec:omitted-proofs}

\def\tildevec #1{\tilde{#1}}

\subsection{Proof of Lemma~\ref{lem:operations}}
Denote the list representations of $g$ and $h$ as
$(x_1^g,y_1^g),\dots,(x_{p_g}^g,y_{p_g}^g)$ and
$(x_1^h,y_1^h),\dots,(x_{p_h}^h,y_{p_h}^h)$, respectively. Recall that both list
representation are stored in doubly linked lists and that the pieces of $g$ and
$h$ are stored in a binary search tree such that for all $x\in[0,t]$ we can evaluate
$g(x)$ and $f(x)$ in time $O(\log p_g)$ and $O(\log p_h)$, respectively.

We show how to construct each of the functions $f_{\min}$, $f_{\shift}$,
$f_{\add}$ and $f_{\round}$ by showing how to construct their list
representations.

First, let us consider $f_{\min}$. We construct the list representation
$(x_1^{\min},y_1^{\min}),\dots$ of $f_{\min}$. The intuition of our approach is
that each piece of $f_{\min}$ must start and end at one of the start or end
points of the pieces of $g$ and $h$. Thus, we will evaluate the function
$\min\{g,h\}$ at all points $x_i^g$ and $x_j^h$ and set $f_{\min}$ accordingly;
then if $f_{\min}$ contains multiple pieces with the same $y_i^{\min}$-value, we
will remove these duplicate pieces.  More concretely, we consider the set
$X = \{x_1^g,\dots,x_{p_g}^g,x_1^h,\dots,x_{p_h}^h \}$ and order it from small
to large.  Now we set $x_i^{\min}$ to the $i$'th smallest element in $X$ for all
$i=1,\dots,p_g+p_h$. Observe that on the interval $[x_{i-1}^{\min},x_i^{\min})$,
$f_{\min}$ must take the value $\min\{g(x_{i-1}^{\min}),h(x_{i-1}^{\min})\}$.
Therefore, we set $y_i^{\min}=\min\{ g(x_{i-1}^{\min}), h(x_{i-1}^{\min}) \}$.
This gives an initial list representation of $f^{\min}$.  Then we ``prune'' the
list representation of $f_{\min}$, i.e., we iterate over all pairs
$(x_i^{\min},y_i^{\min})$ in increasing order of $i$ and if $y_{i-1}^{\min} =
y_i^{\min}$ then we remove the pair $(x_{i-1}^{\min},y_{i-1}^{\min})$ from the
list representation of $f_{\min}$.  Observe that at the end of this process, all
values of $y_i^{\min}$ are pairwise disjoint (since the functions $g$ and $h$
are monotone).

To see that $f_{\min}(x) = \min\{g(x),h(x)\}$ for all $x\in[0,t]$, we observe that
for all $x\in X$ (where $X$ is as in the paragraph above) we have set
$f_{\min}(x)$ correctly by construction. Furthermore, on all contiguous
intervals in $[0,t]\setminus X$, $g$ and $h$ are constant and thus $f_{\min}$ is
constant. Therefore, for all $x\in[0,t]\setminus X$, $f_{\min}(x)$ is also set
correctly.

Next, we observe that $f_{\min}$ has at most $p_g + p_h$ pieces because $X$ consisted
of at most $p_g + p_h$ elements and after that we only removed pieces from
$f_{\min}$.  Furthermore, ordering the elements in $X$ can be done in time
$O((p_g+p_h)\log(p_g+p_h))$ and evaluating $\min\{ g(x_i^{\min}), h(x_i^{\min}) \}$
can be done in time $O(\log(p_g)+\log(p_h))$.
After that we only performed a single pass over the list representations of
$f_{\min}$ in time $O(\abs{X}) = O(p_g+p_h)$. Therefore, it took time
$O((p_g + p_h)\log(p_g+p_h))$ to create the list representation of $f_{\min}$.
Finally, note that to store the elements $x_i^{\min}$ in the binary search tree,
we need additional time $O((p_g+p_h)\log(p_g+p_h))$.

Now we observe that $f_{\add}$ can be computed similarly to $f_{\min}$: the
function $f_{\add}$ only changes its functions values at the points in $X$
(where $X$ is as above). Therefore, we let $x_i^{\add}$ be the $i$'th smallest
element in $X$ and set $y_i^{\add} = g(x_{i-1}^{\add}) + h(x_{i-1}^{\add})$,
followed by the same pruning step as above. The rest of the proof goes through
as above.

Next, consider $f_{\shift}$.  We construct the list representation
$(x_1^{\shift},y_1^{\shift}),\dots,(x_{p_g}^{\shift},x_{p_g}^{\shift})$ of
$f_{\shift}$. For all $i=1,\dots,p_g$, we set $x_i^{\shift} = x_i^g + c$ and
$y_i^{\shift} = y_i^g$. The correctness is straightforward and from the
construction it is evident that there are only $p_g$ pieces and that everything
can be done in time $O(p_g \log(p_g))$ (since we still need to construct the
binary tree for the pieces of $f_{\shift}$).

Finally, us consider $f_{\round}$. As before, we construct the list representation
of $f_{\round}$,
$(x_1^{\round},y_1^{\round}),\dots,(x_{p_g}^{\round},x_{p_g}^{\round})$. For all
$i=1,\dots,p_g$, we set $x_i^{\round} = x_i^g$ and
$y_i^{\round} = \lceil y_i^g \rceil_{1+\delta}$. After that, we perform the same
pruning step as in the construction of $f_{\min}$. Since $g$ takes values in
$\Winfty = \{0\} \cup[1,W] \cup\{+\infty\}$ and $g$ is monotone, $f_{\round}$
can take at most $2 + \lceil \log_{1+\delta}(W) \rceil$ different values. Again,
the running time bound stems from the fact that we have to construct the binary
search tree for the pieces of $f_{\round}$.

\subsection{Proof of Lemma~\ref{lem:convolution}}
    Let $(x^s_1,y^s_1),\dots,(x^s_{p_s},y^s_{p_s})$ be the list
    representation of $f_s$ for $s=1,2$, where $p_s \leq p$ is the number of pieces of
	$f_s$. We create pairs $(y^1_i,y^2_j)$ for all
    $(i,j) \in\{1,\dots,p_1\}\times\{1,\dots,p_2\}$, and order them such that $y^1_i+y^2_j$ becomes
    monotonically increasing. We iterate over all pairs in this order, and
    in each iteration we set the function value $f(x)$ for some $x$-values
    to $y:=y^1_i+y^2_j$, where $(y^1_i,y^2_j)$ is the pair considered
    during the iteration. Here, we start with large $x$-values (at which $f$
	takes the smallest values) and keep on decreasing $x$ (and the function
	values increase); in other words, we construct
	$f$ on its domain $[0,t]$ from right to left. More concretely, let $x_{\max}$ denote the
    highest $x$-value for which we did not yet set a function value (or
    $-\infty$ if all function values have been set). Let
    $x'= x^1_{i-1} + x^2_{j-1}$. We set the function values for all
    $x\in[x',x_{\max})$ to $y$ and then set $x_{\max}=x'$. For each such new
	piece of $f$, we store that we combined the indices~$i$ and~$j$ of the
	pieces that we used from $f_1$ and from $f_2$. Then we proceed
	with next iteration until all function values have been set.
	
	The following two statements show
    that this procedure is correct.
	\begin{description}
	\item[1. Each \mathversion{bold}$x$\mathversion{normal} is assigned a function
	value that is at most the correct value \mathversion{bold}$f(x)$\mathversion{normal}.]

	To see this let $x\in[0,t]$ and recall that 
	\begin{equation*}
	f(x)=\min_{\bar{x}\in[0,x]} f_1(\bar{x})+f_2(x-\bar{x}) \enspace.
	\end{equation*}
	Let $\bar{x}^*$ denote the value of $\bar{x}$ that attains the minimum in the
	above expression, and let $i^*$ and $j^*$ denote the indices of the pieces that $\bar{x}^*$
	and $x-\bar{x}^*$ fall into, w.r.t.\ the list representations of $f_1$ and
	$f_2$, respectively. This means $\bar{x}^*\in[x_{i^*-1}^1,x_{i^*}^1)$ and
	$x-\bar{x}^*\in [x_{j^*-1}^2,x_{j^*}^2)$. Hence, $x\ge x'=x^1_{i^*-1}+x^2_{j^*-1}$.
	Therefore, either in the iteration for the pair $(y_{i^*}^1,y_{j^*}^2)$ or
	before, the procedure assigns a function value to~$x$. Because the procedure
	assigns function-values in monotonically increasing fashion we are
	guaranteed that the function value that is assigned is at most the correct
	value.

	\item[\mathversion{bold}2. The function value $y$ that is assigned is at least the correct value $f(x)$\mathversion{normal}.]
	Suppose that during some iteration we assign the function value
	$y=y_i^1+y_j^2$ to $x\in[x',\dots,x_{\max})$, where
	$x'=x_{i-1}^1+x_{j-1}^2$.
	We have
	 \begin{align*}
	   f(x)&=\min_{\bar{x}\in[0,x]} f_1(\bar{x})+f_2(x-\bar{x}) &\text{(definition)}\\ 
		   &\le f_1(x_{i-1}^1)+f_2(x-x^1_{i-1})   & \text{(consider $\bar{x}=x_{i-1}^1$)}\\
		   &\le f_1(x_{i-1}^1)+f_2(x'-x^1_{i-1})  & \text{($x'\le x,\,f_2$ monotonically decreasing)}\\
		   &= f_1(x_{i-1}^1)+f_2(x^2_{j-1})  \\
		   &= y_i^1+y_j^2     &\text{($y_i^1=f_1(x_{i-1}^1)$ and $y_j^2=f(x_{j-1}^2)$)}\\
		   &= y\enspace.
			\end{align*}
	Hence, the assigned value is at least $f(x)$.

	\end{description}
	Observe that we can implement the above procedure in time $O(p^2\log p)$:
	We first sort the at most $p^2$ pairs in time $O(p^2\log p)$. Then every
	iteration can be executed in constant time because setting the function
	values for $x\in[x',x_{\max})$ to $y$ can be performed by adding the pair
	$(x_{\max},y)$ to the list-representation of $f$ and updating $x_{\max}$ to
	$x'$ takes time $O(1)$.

	Finally, suppose we already computed $f$ and, given $x\in[0,t]$, we shall
	return a value $\bar{x}^*\in[0,t]$
	such that $f(x) = f_1(\bar{x}^*) + f_2(x-\bar{x}^*)$. First, let
	$(x_1,y_1),\dots,(x_p,y_p)$ denote the list representation of~$f$. Then we
	can determine the piece~$\ell$ of $f$ such that $x\in[x_\ell,x_{\ell+1})$ in
	time $O(\lg p)$ since we store the pairs $(x_i,y_i)$ of $f$ in a binary
	search tree.  Recall that for each piece of $f$, we stored the indices $i$
	and $j$ of the pieces from $f_1$ and $f_2$ that we combined. Now observe
	that we have $\bar{x}^*\in[x_{i-1}^1,x_i^1)$ and
	$x-\bar{x}^*\in[x_{j-1}^2,x_j^2)$, where $i$ and $j$ are such that these
	pieces from $f_1$ and $f_2$ form the corresponding piece of $f$.  Thus, to
	find $\bar{x}^*$ we can first try to set $\bar{x}^* = x_{i-1}^1$.  If
	$x-\bar{x}^*=x-x_{i-1}^1\in[x_{j-1}^2,x_j^2)$ then we are done. Otherwise,
	we must have that $x-x_{i-1}^1\geq x_j^2$.  Thus, we have to increase the
	value of $\bar{x}^*$ from $x_{i-1}^1$ until it is large enough such that
	$x-\bar{x}^*\in[x_{j-1}^2,x_j^2)$. This can be achieved by setting
	$\Delta=(x-x_{x-1}^1) - x_j^2$ and
	$\bar{x}^* = x_{i-1}^1 + \Delta
		+ \frac{1}{2} \min\{ x_i^1 - (x_{i-1}^1 + \Delta), x_j^2 - x_{j-1}^2 \}$.
	Note that this value of $\bar{x}^*$ can be computed in time $O(1)$.
	Thus, the total time to return $\bar{x}^*$ is $O(\lg p)$.

\subsection{Proof of Theorem~\ref{thm:well-behaved-static}}
Recall that the dependency graph is a DAG. We call a vertex without any incoming
edges a \emph{leaf}. The \emph{level} of a vertex~$u$ is the length of the
longest path from a leaf to~$u$. Note that since each node can only reach
$h$~other nodes, every vertex has level at most~$h$.

We compute the DP bottom-up, starting at the leaves of the DAG and
then recursively computing the solutions for rows~$i$ for which the solutions of
$\In(i)$ have already been computed. We store the approximate solutions
$\ADP(i,\cdot)$ using monotone piecewise constant functions. 

We prove the theorem by induction over the level of~$i$ in the dependency graph.
We show the stronger statement that for every DP row~$i$ of level~$\ell$,
$\ADP(i,\cdot)$ is an $\alpha^{\ell+1}$-approximation of $\DP(i,\cdot)$.

We start with leaf vertices (i.e., vertices of level $0$). For a leaf~$i$, we
use Properties~4(b) and 4(c) to obtain that $\tilde{\Procedure}_i$ returns
$\ADP(i,\cdot)$ which is a monotonone piecewiese constant function with at most
$p$~pieces and which is an $\alpha$-approximation of $\DP(i,\cdot)$.

Next, consider a row~$i$ of level~$\ell$. We use $\tilde{\Procedure}_i$ to compute
$\ADP(i,\cdot) = \tilde{\Procedure}_i(\{ \ADP(i',\cdot) \colon i'\in\In(i) \}$.
By induction hypothesis, all solutions $\ADP(i',\cdot)$, $i'\in\In(i)$, are
stored as monotone piecewise constant functions and each of them has at most
$p$~pieces. Since we apply the operations from
Lemma~\ref{lem:operations} only $O(1)$~times, the number of pieces
only grows by a factor $O(1)$. Since we only apply the $(\min,+)$-convolution from
Lemma~\ref{lem:convolution} at most a single time, the number of pieces after
the convolution is bounded by $O(p^2)$.  Thus, we will never operate on functions with
more than $O(p^2)$~pieces.  The bounds from Lemmas~\ref{lem:operations}
and~\ref{lem:convolution} imply that all operations to compute
$\tilde{\Procedure}_i$ can
be performed in time at most $O(p^2 \log(p))$. Furthermore, by induction
hypothesis and since each $i'$ is at level $\ell'\leq \ell-1$, we know that
$\ADP(i',\cdot)$ is an $\alpha^{\ell}$-approximation of $\DP(i',\cdot)$. Using
Properties~(3) and 4(a), we get that $\ADP(i,\cdot)$ is an
$\alpha^{\ell+1}$-approximation of $\DP(i,\cdot)$.

The theorem's approximation guarantee follows from Property~(2) which implies
that $\ell\leq h$ for all DP rows~$i$ in the dependency graph. Furthermore,
above we argued that each solution $\ADP(i,\cdot)$ can be computed in time
$O(p^2 \lg(p))$ which gives a total running time of
$O(\abs{I}\cdot p^2 \log(p))$.

\subsection{Proof of Theorem~\ref{thm:well-behaved-dynamic}}
Consider a row~$i$ for which $\DP(i,\cdot)$ changes. Note that we only have to
compute DP solutions for rows~$i'$ which are reachable from $i$ in the
dependency graph. Since we assume that the dependency graph is a DAG and
$\Reach(i)\leq h$ for all rows~$i$, there can be at most~$h$ such rows. In the
proof of Theorem~\ref{thm:well-behaved-static} we argued that each solution
$\ADP(i,\cdot)$ can be computed in time $O(p^2\log(p))$. This gives the proof of
the theorem.

\subsection{Property of the Räcke Tree}
\label{sec:property-racke-tree}

Let $G=(V_G,E_G)$ be an undirected graph and let $T=(V_T,E_T)$ be a Räcke tree
for $G$. 
We prove that $\mincut_T(A,B)\ge \mincut_G(A,B)$ by showing that for any set of
vertices $S_T\subseteq V_T$, it holds that $\capac_T(S_T)\geq \capac_G(S)$ where
$S\subseteq V_G$ is the set of leaf vertices in $V_T$.

Let $S_T\subseteq V_T$ and consider the cut $(S_T,\bar{S}_T)$ in $T$. 
We use $S$ to denote the restriction of $S_T$ to the leaf vertices and observe that
$(S,\bar{S})$ forms a cut in $G$ as well. Then:
\begin{align*}
\capac_T(S_T)
&=\lsum{(x_t,y_t)\in S_T\times\bar{S}_T}{\capac_T(x_t,y_t)} 
  & \text{(definition of $\capac_T(S_T)$)}\\
&= \lsum{(x_t,y_t)\in S_T\times\bar{S}_T}{\capac_G(V_{x_t}\cap V_{y_t})}
  & \text{(definition of tree edge capacity)}\\
&= \sum_{(x_t,y_t)\in S_T\times\bar{S}_T}\lsum{(x,y)\in V_{x_t}\times
                                                  \bar{V}_{x_t}}{\capac_G(x,y)}
  & (\text{w.l.o.g.\ assume $V_{x_t}\subseteq V_{y_t}$})\\
&= \lsum{\{x,y\}\in E_G}{\capac_G(x,y)}\lsum{(x_t,y_t)\in S_T\times\bar{S}_T}
\mathbbm{1}\{x\in  V_{x_t}\wedge y\in \bar{V}_{x_t}\}
  & \text{(change order of summation)}\\
&\ge 
\sum_{(x,y)\in S\times\bar{S}}\capac_G(x,y)=\capac_G(S)\enspace.
\end{align*}
Here the inequality follows because a pair $(x,y)\in {S}\times\bar{S}$ whose
capacity is
counted on the right hand side corresponds to a graph edge
$\{x,y\}\in E_G$ (between $x\in S$ and $y\in \bar{S}$). This graph edge contributes
to the capacity on 
every edge of the $x$-$y$ path in $T$. One of these edges must be cut 
by $S_T$, i.e., $\mathbbm{1}\{x\in  V_{x_t}\wedge y\in \bar{V}_{x_t}\}=1$ for
this tree edge. Hence, its capacity is also counted on the left hand side.

\subsection{Proof of Lemma~\ref{lem:approx-dp}}
The lower bound is immediate since
$\tilde{\Procedure}_i$ is an $\alpha$-approximation of $\Procedure_i$.
For the upper bound we use induction over $\ell$. For $\ell=0$ observe that
\begin{align*}
\tDP(i,\cdot)
	&= \tilde{\Procedure}_i( \{ \tDP(i') \colon i'\in\In(i) \})  &
                                                                           \text{(definition)} \\
   &\le\alpha \Procedure_i( \{ \tDP(i') \colon i'\in\In(i) \}) &
                                                               \text{($\tilde{\Procedure}_i$
                                                               is $\alpha$-approximate)}\\
&=\alpha \Procedure_i(\emptyset) &       \text{($v_i$ is a leaf)}\\
&=\alpha \DP(i) & 						      \text{(the DP is okay-behaved)}
\end{align*}
For $\ell>0$ we have
\begin{align*}
\tDP(i,\cdot) &= \tilde{\Procedure}_i( \{ \tDP(i') \colon i'\in\In(i) \}) &
                                                                           \text{(definition)} \\
   &\le\alpha \Procedure_i( \{ \tDP(i') \colon i'\in\In(i) \}) &
                                                               \text{($\tilde{\Procedure}_i$
                                                               is $\alpha$-approximate)}\\
&=\alpha \Procedure_i( \{ \alpha^{\ell}\DP(i',\cdot) \colon i'\in\In(i) \}) &       \text{(induction hypothesis)}\\
&=\alpha^{\ell+1} \DP(i,\cdot) & 						      \text{(the DP is okay-behaved)}
\end{align*}
Here the induction hypothesis exploits the fact that all $i'\in\In(i)$,
have level strictly less than $\ell$ in the dependency graph.

\subsection{Proof of Lemma~\ref{lem:tree-dp}}
The claim about the approximation guarantee follows immediately from
Lemma~\ref{lem:approx-dp} and the fact that the root has level at most $h$
(since the longest leaf-root path in the dependency tree has length $h$).
To obtain running time $O(\abs{V_T} \cdot t)$, we compute the solutions
$\ADP(v_1),\dots,\ADP(v_n)$ in this order, i.e., based on the topological
ordering of the dependency DAG. Then by assumption on the ordering of
the rows $i$ and since all $\tilde{\Procedure}_i$ can be computed in time $t$,
the lemma follows.

\subsection{Proof of Lemma~\ref{lem:tree-dp-dynamic}}
Suppose the inserted or deleted edge is incident upon a vertex $i$.
Since the DPs we consider are well-behaved, we only need to recompute DP
solutions for those vertices $j$ such that there exists a directed path from
$i$ to $j$, $j\geq i$, in the dependency graph. By construction of the
dependency graph, there can be at most $h$ such vertices (since the longest
leaf-root path in the dependency graph has length $h$). Therefore, we can
recompute all of these solutions in time $O(h \cdot t)$. After we finished the
recomputation, the guarantees on the approximation ratio are implied by
Lemma~\ref{lem:tree-dp}.

\subsection{Proof of Lemma~\ref{lem:multimin}}
We only prove the case if all functions $f_i$ are monotonically decreasing. The
case for monotonically increasing functions is analogous. Let $\cal P$ denote 
the set of all pieces in functions $f_i$.
 Consider a piece $p\in\cal P$ that starts at $t_1$ ends at $t_2$ and has value
$\alpha$. We construction a piece-wise constant function $f_p:[0,t]\to\Winfty$
with two pieces that has value $\infty$ on $[0,t_1)$, and value $\alpha$ on 
$[t_1,t]$ (this means we extend the piece from $t_2$ to $t$). 

Because the functions are monotonically decreasing we can rewrite $f_{\min}$ as
a minimum of the piece-functions $f_p$, i.e.,
\begin{equation*}
 f_{\min}(x)=\min_{p\in\cal P}f_p(x)\enspace.
\end{equation*}
We now sort all pieces in $\cal P$ by there start-point. By processing the pieces
in sorting order we can build the result function step-by-step. Let $f_{r-1}$
denote the piece-wise constant function encoding the minimum over the first $r-1$
pieces. In order to compute $f_{r}$ we have to compare the last piece of
$f_{r-1}$ to the $r$-th piece $p_r$ in $\cal P$. If the value of $p_r$ is higher 
than $f_{r-1}(\infty)$ (the value of the last piece in $f_{r-1}$) we ignore the
piece $p_r$. Otherwise, we end the current last piece of $f_{r-1}$ at the 
start time $t_r$ of piece $p_r$ and add the piece $p_r$ with its start time,
its value, and an end time of $t$. The running time is dominated by sorting the
pieces and inserting them into a binary search tree when adding them to the 
result function.

\subsection{Proof of Lemma~\ref{lem:monotone-convolution-is-monotone}}
	We can assume w.l.o.g.\ that $f_2$ is monotonically decreasing (this follows
	from the symmetry of $(\min,+)$-convolution). Now the lemma is implied by
	the following computation, where in the third step we use the monotonicity
	of $f_2$, i.e., we use that $f_2(x') \leq f_2(x)$ for all $x'\geq x$:
   \begin{align*}
       f(x')
       &= \min_{\bar{x}\in[0,x']} f_1(\bar{x}) + f_2(x'-\bar{x}) \\
       &\leq \min_{\bar{x}\in[0,x]} f_1(\bar{x}) + f_2(x'-\bar{x}) \\
       &\leq \min_{\bar{x}\in[0,x]} f_1(\bar{x}) + f_2(x-\bar{x}) \\
       &= f(x). \qedhere
   \end{align*}

\end{document}